\documentclass[11pt]{book}
\usepackage{amssymb}
\usepackage{amsmath}
\usepackage{amsthm}
\usepackage{graphics}
\usepackage{epsfig,amssymb,latexsym,verbatim}
\usepackage{graphicx}
\usepackage{multirow}
\usepackage{multicol}
\usepackage{hypernat}
\usepackage{float}
\usepackage{booktabs}
\usepackage{fancyhdr}
\usepackage{enumerate,verbatim}
\usepackage[sort&compress]{natbib}
\usepackage{rotating}
\usepackage[lined,algonl,boxed]{algorithm2e}
\floatstyle{ruled}
\newfloat{algorithm}{tbp}{loa}
\floatname{algorithm}{Algorithm}
\usepackage{graphicx}
\usepackage{hyperref}
\usepackage{subfigure}
\usepackage{graphicx}
\usepackage{ulem}
\usepackage{url}
\usepackage{threeparttable}
\usepackage{xr}
\usepackage{multirow}
\usepackage{multicol}
\usepackage{wrapfig}
\usepackage{lscape}
\usepackage{rotating}
\usepackage{epstopdf}
\usepackage{subfigure}

\def\bs{\boldsymbol}

\newcommand{\hbbeta}{\widehat{\bs{\beta}}}
\newcommand{\hbgamma}{\widehat{\bs{\gamma}}}
\newcommand{\hbtheta}{\widehat{\bs{\theta}}}
\newcommand{\tbbeta}{\widetilde{\bs{\beta}}}
\newcommand{\tbgamma}{\widetilde{\bs{\gamma}}}
\newcommand{\tbtheta}{\widetilde{\bs{\theta}}}
\newcommand{\sigs}[1]{\bs{#1}_{\mathcal{S}}}

\newcommand{\sigsy}[1]{\bs{#1}_{\mathcal{S}^{\prime}}}

\newcommand{\thetheorem}{{\thesection. \arabic{theorem}}}
\newcommand{\thelemma}{{\thesection. \arabic{lemma}}}
\newcommand{\theproposition}{{\thesection. \arabic{proposition}}}
\newcommand{\thecorollary}{{\thesection. \arabic{corollary}}}
\newtheorem{theorem}{{\sc Theorem}}
\newtheorem{lemma}{{\sc Lemma}}

\begin{document}
\renewcommand{\baselinestretch}{1.2}
\markboth{\hfill{\footnotesize\rm Xinyi Li, Li Wang and Dan Nettleton}\hfill}
{\hfill {\footnotesize\rm Sparse Model Identification and Learning for Ultra-high-dimensional Additive Partially Linear Models} \hfill}
\renewcommand{\thefootnote}{}
$\ $\par \fontsize{10.95}{14pt plus.8pt minus .6pt}\selectfont
\vspace{0.8pc} \centerline{\Large\bf Sparse Model Identification and Learning for }
\centerline{\Large\bf Ultra-high-dimensional Additive Partially Linear Models}
\vspace{.4cm} \centerline{Xinyi Li$^{a}$, Li Wang$^{b}$ and Dan Nettleton$^{b}$
\footnote{\emph{Address for correspondence}: Li Wang, Department of Statistics and the Statistical Laboratory, Iowa State University, Ames, IA, USA. Email: lilywang@iastate.edu}} \vspace{.4cm} \centerline{\it $^{a}$SAMSI / University of North Carolina at Chapel Hill and $^{b}$Iowa State University} \vspace{.55cm}
\fontsize{9}{11.5pt plus.8pt minus .6pt}\selectfont

\begin{quotation}
\noindent \textit{Abstract:} The additive partially linear model (APLM) combines the flexibility of nonparametric regression with the parsimony of regression models, and has been widely used as a popular tool in multivariate nonparametric regression to alleviate the ``curse of dimensionality''. A natural question raised in practice is the choice of structure in the nonparametric part, that is, whether the continuous covariates enter into the model in linear or nonparametric form. In this paper, we present a comprehensive framework for simultaneous sparse model identification and learning for ultra-high-dimensional APLMs where both the linear and nonparametric components are possibly larger than the sample size. We propose a fast and efficient two-stage procedure. In the first stage, we decompose the nonparametric functions into a linear part and a nonlinear part. The nonlinear functions are approximated by constant spline bases, and a triple penalization procedure is proposed to select nonzero components using adaptive group LASSO. In the second stage, we refit data with selected covariates using higher order polynomial splines, and apply spline-backfitted local-linear smoothing to obtain asymptotic normality for the estimators. The procedure is shown to be consistent for model structure identification. It can identify zero, linear, and nonlinear components correctly and efficiently. Inference can be made on both linear coefficients and nonparametric functions. We conduct simulation studies to evaluate the performance of the method and apply the proposed method to a dataset on the Shoot Apical Meristem (SAM) of maize genotypes for illustration.

\vspace{9pt} \noindent \textit{Key words and phrases:} Dimension reduction, inference for ultra-high-dimensional data, semiparametric regression, spline-backfitted local polynomial, structure identification, variable selection.
\end{quotation}

\fontsize{10.95}{14pt plus.8pt minus .6pt}\selectfont

\thispagestyle{empty}

\setcounter{chapter}{1}
\setcounter{equation}{0}
\renewcommand{\theequation}{\arabic{equation}}
\renewcommand{\thesection}{\arabic{section}}\setcounter{section}{0}
\section{Introduction} \label{SEC:introduction} \vskip 0.1in

In the past three decades, flexible and parsimonious additive partially linear models (APLMs) have been extensively studied and widely used in many statistical applications, including biology, econometrics, engineering, and social science. Examples of recent work on APLMs include \cite{liang2008additive}, \cite{liu2011estimation}, \cite{ma2011spline}, \cite{wang2011estimation}, \cite{ma2013simultaneous}, \cite{wang2014estimation} and \cite{lian2014generalized}. APLMs are natural extensions of classical parametric models with good interpretability and are becoming more and more popular in data analysis.

Suppose we observe $\{(Y_i,\mathbf{Z}_{(i)},$ $\mathbf{X}_{(i)})\}_{i=1}^{n}$. For subject $i =1,\ldots, n$, $Y_i$ is a univariate response, $\mathbf{Z}_{(i)}=( Z_{i1},\ldots ,$ $Z_{ip_1})^{\top}$ is a $p_1$-dimensional vector of covariates that may be linearly associated with the response, and $\mathbf{X}_{(i)}=(X_{i1},\ldots ,X_{i p_2})^{\top}$ is a $p_2$-dimensional vector of continuous covariates that may have nonlinear associations with the response. We assume $\{(Y_i,\mathbf{Z}_{(i)},$ $\mathbf{X}_{(i)})\}_{i=1}^{n}$ is an i.i.d sample from the distribution of $\left(Y,\bs{Z},\bs{X}\right) $, satisfying the following model:
\begin{align}
    Y_i &= \mu + \mathbf{Z}_{(i)}^{\top}\bs{\alpha} + \sum_{\ell=1}^{p_2}\phi_{\ell}(X_{i\ell}) + \varepsilon_i
    = \mu + \sum_{k=1}^{p_1}Z_{ik}\alpha_k + \sum_{\ell=1}^{p_2}\phi_{\ell}(X_{i\ell}) + \varepsilon_i,
\label{model:aplm}
\end{align}
where $\mu$ is the intercept, $\alpha_k$, $k=1,\ldots,p_1$, are unknown regression coefficients, $\left\{\phi _{\ell}\left( \cdot \right) \right\} _{\ell=1}^{p_2}$ are unknown smooth functions, and each $\phi_{\ell}\left( \cdot \right)$ is centered with ${\rm E} \phi _{\ell}\left(X_{i\ell}\right)=0$ to make model (\ref{model:aplm}) identifiable. The $\mathbf{X}_{(i)}$ is a $p_2$-dimensional vector of zero-mean covariates having density with a compact support. Without loss of generality, we assume that each covariate $\left\{X_{i\ell}\right\}_{\ell=1}^{p_2}$ can be rescaled into an interval $\chi=[a,b]$. The $\varepsilon_i$ terms are iid random errors with mean zero and variance $\sigma^2$. 

The APLM is particularly convenient when $\bs{Z}$ is a vector of categorical or discrete variables, and in this case, the components of $\bs{Z}$ enter the linear part of model (\ref{model:aplm}) automatically, and the continuous variables usually enter the model nonparametrically. In practice, we might have reasons to believe that some of the continuous variables should enter the model linearly rather than nonparametrically. A natural question is how to determine which continuous covariates have a linear effect and which continuous covariates have a nonlinear effect. If the choice of linear components is correctly specified, then the biases in the estimation of these components are eliminated and root-$n$ convergence rates can be obtained for the linear coefficients. However, such prior knowledge is rarely available, especially when the number of covariates is large. Thus, structure identification, or linear and nonlinear detection, is an important step in the process of building an APLM from high-dimensional data.

When the number of covariates in the model is fixed, structure identification in additive models (AMs) has been studied in the literature. \citet{zhang2011linear} proposed a penalization procedure to identify the linear components in AMs in the context of smoothing splines ANOVA. They demonstrated the consistency of the model structure identification and established the convergence rate of the proposed method specifically under the tensor product design. \citet{huang2012semiparametric} proposed another penalized semiparametric regression approach using a group minimax concave penalty to identify the covariates with linear effects. They showed consistency in determining the linear and nonlinear structure in covariates, and obtained the convergence rate of nonlinear function estimators and asymptotic properties of linear coefficient estimators; but they did not perform variable selection at the same time.

For high-dimensional AMs, \citet{lian2015separation} proposed a double penalization procedure to distinguish covariates that enter the nonparametric and parametric parts and to identify significant covariates simultaneously.
They demonstrated the consistency of the model structure identification, and established the convergence rate of nonlinear function estimators and asymptotic normality of linear coefficient estimators. Despite the nice theoretical properties, their method heavily relies on the local quadratic approximation in \cite{fan2001variable}, which is incapable of producing naturally sparse estimates. In addition, employing the local quadratic approximation can be extremely expensive because it requires the repeated factorization of large matrices, which becomes infeasible when the number of covariates is very large.

Note that all the aforementioned papers \citep{zhang2011linear,huang2012semiparametric,lian2015separation} about structure identification focus on the AM with continuous explanatory variables. However, in many applications, a canonical partitioning of the variables exists. In particular, if there are categorical or discrete explanatory variables, as in the case of the SAM data studies (see the details in Section \ref{sec:application}) and in many genome-wide association studies, we may want to keep discrete explanatory variables separate from the other design variables and let discrete variables enter the linear part of the model directly. In addition, if there is some prior knowledge of certain parametric forms for some specific covariates, such as a linear form, we may lose efficiency if we simply model all the covariates nonparametrically.

The above practical and theoretical concerns motivate our further investigation of the simultaneous  variable selection and structure selection problem for flexible and parsimonious APLMs, in which the features of the data suitable for parametric modeling are modeled parametrically and nonparametric components are used only where needed. We consider the setting where both the dimension of the linear components and the dimension of nonlinear components is ultra-high. We propose an efficient and stable penalization procedure for simultaneously identifying linear and nonlinear components, removing insignificant predictors, and estimating the remaining linear and nonlinear components. We prove the proposed \textit{S}parse \textit{M}odel \textit{I}dentification, \textit{L}earning and \textit{E}stimation (referred to as SMILE) procedure is consistent. We propose an iterative group coordinate descent approach to solve the penalized minimization problem efficiently. Our algorithm is very easy to implement because it only involves simple arithmetic operations with no complicated numerical optimization steps, matrix factorizations, or inversions. In one simulation example with $n=500$ and $p_1=p_2=5000$, it takes less than one minute to complete the entire model identification and variable selection process on a regular PC.

After variable selection and structure detection, we would like to provide an inferential tool for the linear and nonparametric components. The spline method is fast and easy to implement; however, the rate of convergence is only established in mean squares sense, and there is no asymptotic distribution or uniform convergence, so no measures of confidence can be assigned to the estimators. In this paper, we propose a two-step spline-backfitted local-linear smoothing (SBLL) procedure for APLM estimation, model selection and simultaneous inference for all the components. In the first stage, we approximate the nonparametric functions $\phi_{\ell} (\cdot)$, $\ell= 1,\ldots,p_2$, with undersmoothed constant spline functions. We perform model selection for the APLM using a triple penalized procedure to select important variables and identify the linear vs. nonlinear structure for the continuous covariates, which is crucial to obtain efficient estimators for the non-zero components. We show that the proposed model selection and structure identification for both parametric and nonparametric terms are consistent, and the estimators of the nonzero linear coefficients and nonzero nonparametric functions are both $L_2$-norm consistent. In the second stage, we refit the data with covariates selected in the first step using higher-order polynomial splines to achieve root-$n$ consistency of the coefficient estimators in the linear part, and apply a one-step local-linear backfitting to the projected nonparametric components obtained from the refitting. Asymptotic normality for both linear coefficient estimators and nonlinear component estimators, as well as simultaneous confidence bands (SCBs) for all nonparametric components, are provided.

The rest of the paper is organized as follows. In Section \ref{sec:method}, we describe the first-stage spline smoothing and propose a triple penalized regularization method for simultaneous model identification and variable selection. The theoretical properties of selection consistency and rates of convergence for the coefficient estimators and nonparametric estimators are developed. Section \ref{sec:nonparmetric} introduces the spline-backfitted local-linear estimators and SCBs for the nonparametric components. The performance of the estimators is assessed by simulations in Section \ref{sec:simulation} and illustrated by application to the SAM data in Section \ref{sec:application}. Some concluding remarks are given in Section \ref{sec:conclusion}. Section A of the online Supplemental Materials evaluates the effect of different smoothing parameters on the performance of the proposed method. Technical details are provided in Section B of the Supplemental Materials.

\setcounter{chapter}{2} \renewcommand{\theproposition}{{2.%
\arabic{proposition}}} \renewcommand{\thesection}{\arabic{section}}
\renewcommand{\thesubsection}{2.\arabic{subsection}}
\setcounter{lemma}{0} \setcounter{section}{1}
\setcounter{theorem}{0} \setcounter{proposition}{0} \setcounter{corollary}{0}

\section{Methodology}
\label{sec:method}

\subsection{Model Setup}

In the following, the functional form (linear vs. nonlinear) for each continuous covariate in model (\ref{model:aplm}) is assumed to be unknown.
In order to decide the form of $\phi_{\ell}$, for each $\ell=1, \ldots p_2$, we can decompose $\phi_{\ell}$ into a linear part and a nonlinear part: $\phi_{\ell}(x) = \beta_{\ell}x + g_{\ell}(x)$,
where $g_{\ell}(x)$ is some unknown smooth nonlinear function (see Assumption (A1) in Appendix \ref{SUBSEC:assump}). For model identifiability, we assume that $\mathrm{E}(X_{i\ell}) = 0$, $\mathrm{E}\{g_{\ell}(X_{i\ell})\} = 0$ and ${\rm E}\{g_{\ell}^{\prime}(X_{i\ell})\} = 0$. The first two constraints $\mathrm{E}(X_{i\ell}) = 0$ and ${\rm E}\{g_{\ell}(X_{i\ell})\} = 0$, are required to guarantee identifiability for the APLM, that is, ${\rm E}\{\phi_{\ell}(X_{i\ell})\} = 0$. The constraint $\mathrm{E}\{g_{\ell}^{\prime}(X_{i\ell})\} = 0$ ensures there is no linear form in nonlinear function $g_{\ell}(x)$. Note that these constraints are also in accordance with the definition of nonlinear contrast space in \cite{zhang2011linear}, which is a subspace of the orthogonal decomposition of RKHS. In the following, we assume $Y_i$ values are centered so that we can express the APLM in (\ref{model:aplm}) without an intercept parameter as
\begin{align}
    Y_i &= \sum_{k=1}^{p_1}Z_{ik}\alpha_k + \sum_{\ell=1}^{p_2}X_{i\ell}\beta_{\ell} + \sum_{\ell=1}^{p_2}g_{\ell}(X_{i\ell}) + \varepsilon_i.
\label{EQ:aplm}
\end{align}

In the following, we define predictor variable $Z_k$ as irrelevant in model (\ref{EQ:aplm}), if and only if $\alpha_k=0$, and $X_{\ell}$ as irrelevant if and only if $\beta_{\ell}=0$ and $g_{\ell}(x_{\ell})=0$ for all $x_{\ell}$ on its support. A predictor variable is defined as relevant if and only if it is not irrelevant. Suppose that only an unknown subset of predictor variables is relevant. We are interested in identifying  such subsets of relevant predictors consistently while simultaneously estimating their coefficients and/or functions.

For covariates $\bs{Z}$, we define
\begin{align*}
  \mbox{Active index set for $\bs{Z}$}:& ~\mathcal{S}_z=\{k = 1,\ldots, p_1: \alpha_k \neq 0\}, \\
  \mbox{Inactive index set for $\bs{Z}$}:& ~\mathcal{N}_z=\{k = 1,\ldots, p_1: \alpha_k = 0\}.
\end{align*}

For continuous covariate $X_{\ell}$, we say it is a linear covariate if $\beta_{\ell} \neq 0$ and $g_{\ell}(x_{\ell}) = 0$ for all $x_{\ell}$ on its support, and $X_{\ell}$ is a nonlinear covariate if $g_{\ell}(x_{\ell})\neq 0$. Explicitly, we define the following index sets for $\bs{X}$:
\begin{align*}
  \mbox{Active pure linear index set for $\bs{X}$}:& ~\mathcal{S}_{x, PL}=\{\ell=1,\ldots, p_2: \beta_{\ell} \neq 0,~ g_{\ell} \equiv 0\}, \\
  \mbox{Active nonlinear index set for $\bs{X}$}:& ~\mathcal{S}_{x, N}=\{\ell=1,\ldots, p_2: g_{\ell} \neq 0\}, \\
  \mbox{Inactive index set for $\bs{X}$}:& ~\mathcal{N}_x=\{\ell=1,\ldots, p_2: \beta_{\ell} = 0,~ g_{\ell} \equiv 0\}.
\end{align*}
Note that the active nonlinear index set for $\bs{X}$, $\mathcal{S}_{x, N}$, can be decomposed as $\mathcal{S}_{x, N} = \mathcal{S}_{x, LN} \cup \mathcal{S}_{x, PN}$, where $\mathcal{S}_{x, LN}=\{\ell=1,\ldots, p_2: \beta_{\ell} \neq 0, ~g_{\ell} \neq 0\}$ is the index set for covariates whose linear and nonlinear terms in (\ref{EQ:aplm}) are both nonzero, and $\mathcal{S}_{x, PN}=\{\ell=1,\ldots, p_2: \beta_{\ell} = 0, ~g_{\ell} \neq 0\}$ is the index set for active pure nonlinear index set for $\bs{X}$.

Therefore, the model selection problem for model (\ref{EQ:aplm}) is equivalent to the problem of identifying $\mathcal{S}_z$, $\mathcal{N}_z$, $\mathcal{S}_{x, PL}$, $\mathcal{S}_{x, LN}$, $\mathcal{S}_{x, PN}$ and $\mathcal{N}_x$. To achieve this, we propose to minimize
\begin{equation}
	\sum_{i=1}^n  \Bigg\{ Y_i - \sum_{k=1}^{p_1}Z_{ik}\alpha_k - \sum_{\ell=1}^{p_2}X_{i\ell}\beta_{\ell} - \sum_{\ell=1}^{p_2}g_{\ell}(X_{i\ell}) \Bigg\}^2 + \sum_{k=1}^{p_1} p_{\lambda_{n1}}(|\alpha_k|) + \sum_{\ell=1}^{p_2} p_{\lambda_{n2}}(|\beta_{\ell}|) + \sum_{\ell=1}^{p_2} p_{\lambda_{n3}}(\|g_{\ell}\|_2),
\label{EQ:loss_penalty}
\end{equation}
where $\Vert g_{\ell}\Vert _{2}^2=\mathrm{E}\{g_{\ell}^{2}(X_{\ell})\}$, and $p_{\lambda _{n1}}\left( \cdot \right)$, $p_{\lambda _{n2}}\left( \cdot \right)$ and $p_{\lambda _{n3}}\left( \cdot \right)$ are penalty functions explained in detail in Section \ref{subsec:selection}. The tuning parameters $\lambda_{n1}$, $\lambda_{n2}$ and $\lambda_{n3}$ decide the complexity of the selected model. The smoothness of predicted nonlinear functions is controlled by $\lambda_{n3}$, and $\lambda_{n1}$, $\lambda_{n2}$ and $\lambda_{n3}$ go to $\infty$ as $n$ increases to $\infty$.

\subsection{Spline Basis Approximation} \label{SUBSEC:spline}

We approximate the smooth functions $\left\{g _{\ell}\left( \cdot \right): \ell=1, \ldots, p_2 \right\}$ in (\ref{EQ:aplm}) by polynomial splines for their simplicity in computation. For example, for each $\ell=1, \ldots, p_2$, let $\upsilon _{0,\ell}, \ldots, \upsilon _{N_n+1,\ell}$ be knots that partition $[a,b]$ with
$a=\upsilon _{0,\ell}<\upsilon _{1,\ell}<\ldots <\upsilon_{N_n,\ell}<\upsilon _{N_n+1,\ell}=b$.
The space of polynomial splines of order $d\geq 1$, $\mathcal{B}^{(d)}_{\ell}[a,b]$, consisting of functions $s(\cdot)$ satisfying (i) the restriction of $s(\cdot)$ to subintervals $[\upsilon _{J,\ell},\upsilon_{J+1,\ell})$, $J = 1, \ldots, N_n + d$, and $\left[\upsilon _{N_n,\ell},\upsilon_{N_n+1,\ell}\right]$, is a polynomial of  $(d-1)$-degree (or less); (ii) for $d\geq 2$ and $0 \leq d^{\prime} \leq d-2$, $s(\cdot)$ is $d^{\prime}$ times continuously differentiable on $[a,b]$. Below we denote $b_{J,\ell}^{(d)}(\cdot)$, $J=1,\ldots, N_n+d$, the basis functions of $\mathbb{B}^{(d)}_{\ell}[a,b]$.

To ensure $\mathrm{E}\{g_{\ell}(X_{i\ell})\} = 0$ and $\mathrm{E}\{g_{\ell}^{\prime}(X_{i\ell})\} = 0$, we consider the following normalized first-order B-splines, referred to as piecewise constant splines.
We define for any $\ell =1, \ldots, p_2$ the piecewise constant B-spline function as the indicator
function $I_{J,\ell}\left( x_{\ell}\right) $ of the $\left( N_n+1\right)
$ equally-spaced subintervals of $[a,b]$
with length $H=H_{n}=(b-a) / \left( N_n+1\right)$, that is,
\begin{align*}
    I_{J,\ell}\left( x_{\ell}\right) &=\left\{
    \begin{array}{ll}
        1 & a + JH\leq x_{\ell}<a + \left( J+1\right) H, \\
        0 & \mbox{otherwise},
    \end{array}
    \right. \, J=0,1, \ldots, N_n-1, ~~\\
	I_{N_n,\ell}\left( x_{\ell}\right) &=\left\{
	\begin{array}{ll}
		1 & a + N_n H\leq x_{\ell}\leq b, \\
		0 & \mbox{otherwise}.
	\end{array}
	\right.
\end{align*}
Define the following centered spline basis
\[
b_{J,\ell}^{(1)}\left(x_{\ell}\right) =I_{J,\ell}\left( x_{\ell}\right) -(\| I_{J,\ell}\|_2 /
\|I_{J-1,\ell}\|_2)I_{J-1,\ell}\left(x_{\ell}\right) , \, \forall ~ J =1, \ldots, N_n, \,  \ell= 1, \ldots, p_2,
\]
with the standardized version given for any $\ell= 1, \ldots, p_2$,%
\begin{equation}
    B_{J,\ell}^{(1)}\left( x_{\ell}\right) =b_{J,\ell}^{(1)}\left( x_{\ell}\right)/\| b_{J,\ell}^{(1)}\| _{2}, \, \forall ~ J=1, \ldots, N_n.
\label{DEF:BJalpha}
\end{equation}
So $\mathrm{E} \{B_{J,\ell}^{(1)}(X_{i\ell})\} = 0$, $\mathrm{E} \{B_{J,\ell}^{(1)}(X_{i\ell})\}^2 = 1$.
In practice, we use the empirical distribution of $X_{1\ell},\ldots,X_{n\ell}$ to perform the centering and scaling in the definitions of $b_{J,\ell}^{(1)}(x_{\ell})$ and $B_{J,\ell}^{(1)}( x_{\ell})$.

We approximate the nonparametric function $g_{\ell}(x_{\ell})$, $\ell= 1,\ldots, p_2$, using the above normalized piecewise constant splines
\begin{equation}
    g_{\ell}(x_{\ell}) \approx g_{\ell s}(x_{\ell})= \sum_{J=1}^{N_n}\gamma _{J,\ell}B_{J,\ell}^{(1)}(x_{\ell})=\mathbf{B}_{\ell}^{(1)\top}(x_{\ell}) \bs{\gamma}_{\ell},
\label{EQ:approx}
\end{equation}
where $\mathbf{B}_{\ell}^{(1)}(x_{\ell})=(B_{1,\ell}^{(1)}(x_{\ell}),\ldots ,B_{N_n,\ell}^{(1)}(x_{\ell}))^{\top}$, and $\bs{\gamma}_{\ell}=\left( \gamma_{1,l},\ldots ,\gamma _{N_n,\ell}\right)^{\top}$ is  a vector of the spline coefficients.  By using the centered constant spline basis functions, we can guarantee that $n^{-1}\sum_{i=1}^{n} g_{\ell s}(X_{i\ell}) = 0$, and $n^{-1}\sum_{i=1}^{n}g_{\ell s}^{\prime}(X_{i\ell}) = 0$ except at the location of the knots.

Denote a length $N_n$ vector $\mathbf{B}_{i\ell}^{(1)} = (B_{1,\ell}^{(1)}(X_{i\ell}), \ldots, B_{N_n, \ell}^{(1)}(X_{i\ell}))^{\top}$.
For any vector $\bs{a} \in \mathbb{R}^p$, denote $\Vert \bs{a} \Vert = (\sum_{\ell = 1}^p a_{\ell}^2)^{1/2}$ as the $L_2$ norm of $\bs{a}$. Following from (\ref{EQ:approx}), to minimize (\ref{EQ:loss_penalty}), it is approximately equivalent to consider the problem of minimizing
\[
    \sum_{i=1}^n \Bigg\{ Y_i - \sum_{k=1}^{p_1}Z_{ik}\alpha_k - \sum_{\ell=1}^{p_2}X_{i\ell}\beta_{\ell} - \sum_{\ell=1}^{p_2}\mathbf{B}_{i\ell}^{(1)}\bs{\gamma}_{\ell} \Bigg\}^2  + \sum_{k=1}^{p_1} p_{\lambda_{n1}}(|\alpha_k|) + \sum_{\ell=1}^{p_2} p_{\lambda_{n2}}(|\beta_{\ell}|) + \sum_{\ell=1}^{p_2} p_{\lambda_{n3}}(\|\bs{\gamma}_{\ell}\|).
\]

\subsection{Adaptive Group LASSO Regularization} \label{subsec:selection}

We use adaptive LASSO \citep{zou2006adaptive} and adaptive group LASSO \citep{huang2010variable} for variable selection and estimation. Other popular choices include methods based on the Smoothly Clipped Absolute Deviation penalty  \citep{fan2001variable} or the minimax concave penalty \citep{zhang2010nearly}.
Specifically, we start with group LASSO estimators obtained from the following minimization:
\begin{align}
	(\widetilde{\bs{\alpha}}, \widetilde{\bs{\beta}}, \widetilde{\bs{\gamma}}) = \underset{{\bs{\alpha},\bs{\beta}, \bs{\gamma}}}{\arg\min} \sum_{i=1}^n \left\{ Y_i - \sum_{k=1}^{p_1}Z_{ik}\alpha_k - \sum_{\ell=1}^{p_2}X_{i\ell}\beta_{\ell} - \sum_{\ell=1}^{p_2}\mathbf{B}_{i\ell}^{(1)}\bs{\gamma}_{\ell} \right\}^2\notag\\
	+ \widetilde{\lambda}_{n1} \sum_{k=1}^{p_1} |\alpha_k| + \widetilde{\lambda}_{n2} \sum_{\ell=1}^{p_2} |\beta_{\ell}| + \widetilde{\lambda}_{n3} \sum_{\ell=1}^{p_2} \|\bs{\gamma}_{\ell}\|.
\label{EQ:loss_glasso_spline}
\end{align}
Then, let $w_k^{\alpha} = |\widetilde{\alpha}_k|^{-1} I\{|\widetilde{\alpha}_k| > 0\} + \infty \times I\{|\widetilde{\alpha}_k| = 0\}$, $w_{\ell}^{\beta} = |\widetilde{\beta}_{\ell}|^{-1} I\{|\widetilde{\beta}_{\ell}| > 0\} + \infty \times I\{|\widetilde{\beta}_{\ell}| = 0\}$, $w_{\ell}^{\bs{\gamma}} = \|\widetilde{\bs{\gamma}}_{\ell}\|^{-1} I\{\|\widetilde{\bs{\gamma}}_{\ell}\| > 0\} + \infty \times I\{\|\widetilde{\bs{\gamma}}_{\ell}\| = 0\}$, where by convention, $\infty \times 0 = 0$. The adaptive group LASSO objective function is defined as
\begin{align}
    L(\bs{\alpha}, \bs{\beta}, \bs{\gamma}; \lambda_{n1}, \lambda_{n2}, \lambda_{n3}) = \sum_{i=1}^n \left\{ Y_i - \sum_{k=1}^{p_1}Z_{ik}\alpha_k - \sum_{\ell=1}^{p_2}X_{i\ell}\beta_{\ell} - \sum_{\ell=1}^{p_2}\mathbf{B}_{i\ell}^{(1)}\bs{\gamma}_{\ell} \right\}^2  \nonumber\\
    + \lambda_{n1} \sum_{k=1}^{p_1} w_k^{\alpha} |\alpha_k| + \lambda_{n2} \sum_{\ell=1}^{p_2} w_{\ell}^{\beta} |\beta_{\ell}| + \lambda_{n3}\sum_{\ell=1}^{p_2} w_{\ell}^{\bs{\gamma}} \|\bs{\gamma}_{\ell}\|.
\label{EQ:loss_aglasso_spline}
\end{align}
The adaptive group LASSO estimators are minimizers of (\ref{EQ:loss_aglasso_spline}), denoted by
\[
    (\widehat{\bs{\alpha}}, \widehat{\bs{\beta}}, \widehat{\bs{\gamma}}) = \underset{{\bs{\alpha},\bs{\beta}, \bs{\gamma}}}{\arg\min} ~ L(\bs{\alpha}, \bs{\beta}, \bs{\gamma}; \lambda_{n1}, \lambda_{n2}, \lambda_{n3}).
\]

The model structure selected is defined by
\begin{align*}
    \widehat{\mathcal{S}}_z &= \left\{1 \leq k \leq p_1: |\widehat{\alpha}_k| > 0 \right\}, ~
    \widehat{\mathcal{S}}_{x, PL} = \left\{\ell: |\widehat{\beta}_{\ell}| > 0, \|\widehat{\bs{\gamma}}_{\ell}\| = 0, 1 \leq \ell \leq p_2 \right\} , \nonumber \\
    \widehat{\mathcal{S}}_{x, LN} &= \left\{\ell: |\widehat{\beta}_{\ell}| > 0, \|\widehat{\bs{\gamma}}_{\ell}\| > 0, 1 \leq \ell \leq p_2 \right\} , ~
    \widehat{\mathcal{S}}_{x, PN} = \left\{\ell: |\widehat{\beta}_{\ell}| = 0, \|\widehat{\bs{\gamma}}_{\ell}\| > 0, 1 \leq \ell \leq p_2 \right\}.
\end{align*}

The spline estimators of each component function are
\begin{equation*}
    \widehat{g}_{\ell}\left( x_{\ell}\right) =\sum_{J=1}^{N_n}\widehat{\gamma}%
    _{J,\ell}B_{J,\ell}^{(1)}\left( x_{\ell}\right)
    -n^{-1}\sum_{i=1}^{n}\sum_{J=1}^{N_n}\widehat{\gamma}_{J,\ell}B_{J,\ell}^{(1)}\left( X_{i\ell}\right) .
\end{equation*}
Accordingly, the spline estimators for the original component functions $\phi_{\ell}$'s are $\widehat{\phi}_{\ell}\left( x_{\ell}\right) = \widehat{\beta}_{\ell} x_{\ell} + \widehat{g}_{\ell}\left( x_{\ell}\right)$.

The following theorems establish the asymptotic properties of the adaptive group LASSO estimators. Theorem \ref{THM:selection} shows the proposed method can consistently distinguish nonzero components from zero components. Theorem \ref{THM:consistency} gives the convergence rates of the estimators. We only state the main results here. To facilitate the development of the asymptotic properties, we assume the following sparsity condition:
\begin{enumerate}
	\item[(A1)] (\textit{Sparsity}) The numbers of nonzero components $|\mathcal{S}_z|$, $|\mathcal{S}_{x,PL}|$ and $|\mathcal{S}_{x,N}|$ are fixed, and there exist positive constants $c_{\alpha}$, $c_{\beta}$ and $c_g$ such that $\min_{k \in \mathcal{S}_z} |{\alpha}_{0k} | \geq c_{\alpha}$, $\min_{\ell \in \mathcal{S}_{x, PL}} | {\beta}_{0\ell} | \geq c_{\beta}$, and $\min_{\ell \in \mathcal{S}_{x, N}} \Vert g_{0\ell} \Vert_2 \geq c_g$.
\end{enumerate}
\noindent Other regularity conditions and proofs are provided in Appendix \ref{SUBSEC:assump}-- \ref{SUBSEC:KKT}.

\begin{theorem}
\label{THM:selection} 
    Suppose that Assumptions (A1), (A2)--(A6) in Appendix \ref{SUBSEC:assump} hold. As $n\rightarrow \infty$, we have $\widehat{\mathcal{S}}_z=\mathcal{S}_z$, $\widehat{\mathcal{S}}_{x, PL}=\mathcal{S}_{x, PL}$, $\widehat{\mathcal{S}}_{x, LN}=\mathcal{S}_{x, LN}$ and $\widehat{\mathcal{S}}_{x, PN}=\mathcal{S}_{x, PN}$
with probability approaching one.
\end{theorem}


In the following, to avoid confusion, we use $\bs{\alpha}_{0}=(\alpha_{01},\ldots, \alpha_{0p_1})^{\top}$, $\bs{\beta}_{0}=(\beta _{01},\ldots ,$ $\beta _{0p_2})^{\top}$ to denote the true parameters in model (\ref{EQ:aplm}), and $\bs{g}_{0}=(g_{01},\ldots, g_{0p_2})^{\top}$ to denote the nonlinear functions in model (\ref{EQ:aplm}). Let $\bs{\alpha}_{0}=(\bs{\alpha}_{0,\mathcal{S}_z}^{\top},\bs{\alpha}_{0,\mathcal{N}_z}^{\top})^{\top}$, where $\bs{\alpha}_{0, \mathcal{S}_z}$ consists of all nonzero components of $\bs{\alpha }_{0}$, and $\bs{\alpha}_{0, \mathcal{N}_z}=\bs{0}$ without loss of generality; similarly, let $\bs{\beta}_{0}=(\bs{\beta}_{0,\mathcal{S}_{x, L}}^{\top},\bs{\beta}_{0,\mathcal{N}_x}^{\top})^{\top}$, where $\bs{\beta}_{0, \mathcal{S}_{x, L}}$ consists of all nonzero components of $\bs{\beta}_{0}$, and $\bs{\beta}_{0, \mathcal{N}_x}=\bs{0}$ without loss of generality.

\begin{theorem}
\label{THM:consistency} 
Suppose that Assumptions (A1), (A2)--(A6) in Appendix \ref{SUBSEC:assump} hold. Then
\begin{eqnarray*}
    \sum\limits_{k \in \mathcal{S}_z} \, | \widehat{\alpha}_k - \alpha_{0k} |^2 =  O_P\left(n^{-1} N_n \right) + O\left(N_n^{-2}\right) + O_P\left(n^{-2}\sum_{j=1}^{3} \lambda_{nj}^2\right), \\
    \sum\limits_{\ell \in \mathcal{S}_{x, L}} |\widehat{\beta}_{\ell} - \beta_{0\ell}|^2 =  O_P\left( n^{-1} N_n\right) + O\left(N_n^{-2}\right) +  O_P\left(n^{-2}\sum_{j=1}^{3} \lambda_{nj}^2\right), \\
    \sum\limits_{\ell \in \mathcal{S}_{x, N}} \, \Vert \widehat{g}_{\ell} - g_{0\ell} \Vert_2^2 =  O_P\left( n^{-1} N_n\right) + O\left(N_n^{-2}\right) +  O_P\left(n^{-2}\sum_{j=1}^{3} \lambda_{nj}^2\right).
\end{eqnarray*}
\end{theorem}

\setcounter{chapter}{3} \renewcommand{\thetheorem}{{\arabic{theorem}}}
\renewcommand{\thelemma}{{3.\arabic{lemma}}}
\renewcommand{\theproposition}{{3.\arabic{proposition}}} %
\renewcommand{\thesection}{\arabic{section}}
\renewcommand{\thesubsection}{3.\arabic{subsection}}
\setcounter{section}{2}

\section{Two-stage SBLL Estimator and Inference} \label{sec:nonparmetric}

After model selection, our next step is to conduct statistical inference for the nonparametric component functions of those important variables. Although the one-step penalized estimation in Section \ref{subsec:selection} can quickly identify the nonzero nonlinear components,  the asymptotic distribution is not available for the resulting estimators.

To obtain estimators whose asymptotic distribution can be used for inference, we first refit the data using selected model,
\begin{equation}
    Y_i = \sum_{k \in \widehat{\mathcal{S}}_z} Z_{ik}{\alpha}_k +
    \sum_{j \in \widehat{\mathcal{S}}_{x, PL}} X_{ij}{\beta}_j +
    \sum_{\ell \in {\widehat{\mathcal{S}}}_{x, N}} {\phi}_{\ell}\left(
    X_{i\ell}\right) + \epsilon_i.
\label{model:aplm_refit}
\end{equation}
We approximate the smooth functions $\left\{\phi _{\ell}\left( \cdot \right): \ell \in {\widehat{\mathcal{S}}}_{x, N} \right\}$ in (\ref{model:aplm_refit}) by polynomial splines introduced in Section \ref{SUBSEC:spline}.
Let $\mathcal{B} _{\ell}^{(d)}$ be the space of polynomial splines of order $d$, and $\mathcal{B} _{\ell}^{0}=\{b\in \mathcal{B}_{\ell}^{(d)}: \mathrm{E}\{b(X_{\ell})\} = 0, ~\mathrm{E}\{b^2(X_{\ell})\} < \infty\}$. Working with $\mathcal{B} _{\ell}^{0}$ ensures that the spline functions are centered, see for example \citet[][]{xue2006additive, wang2007spline, wang2014estimation}. Let $\left\{B_{J,\ell}^{(d)}\left( \cdot \right) \right\} _{j=1}^{M_n}$ be a set of standardized spline basis functions for $\mathcal{B} _{\ell}^{0}$ with dimension $M_n=N_{n}+d-1$, where $B_{J,\ell}^{(d)}(x_{\ell}) = b_{J,\ell}^{(d)}(x_{\ell}) / \|b_{J,\ell}^{(d)}\|_2$, $J = 1, \ldots, M_n$, so that $\mathrm{E} \{B_{J,\ell}^{(d)}(X_{\ell})\} \equiv 0$, $\mathrm{E} \{B_{J,\ell}^{(d)}(X_{\ell})\}^2 \equiv 1$. Specifically, if $d=1$, $M_n = N_n$ and $B_{J,\ell}^{(1)}(\cdot)$ is the standardized piecewise constant spline function defined in (\ref{DEF:BJalpha}).

We propose a one-step backfitting using refitted pilot spline estimators in the first stage followed by local-linear estimators. The refitted coefficients are defined as
\begin{align}
	(\widehat{\bs{\alpha}}^{\ast}, \widehat{\bs{\beta}}^{\ast}, \widehat{\bs{\gamma}}^{\ast}) = \underset{{\bs{\alpha},\bs{\beta}, \bs{\gamma}}}{\arg\min}
	\sum_{i=1}^n \left(Y_i - \sum_{k \in \widehat{\mathcal{S}}_z}Z_{ik}\alpha_k - \sum_{j \in \widehat{\mathcal{S}}_{x, PL}} X_{ij}\beta_j - \sum_{\ell \in {\widehat{\mathcal{S}}}_{x, N}}\mathbf{B}_{i\ell}^{(d)}\bs{\gamma}_{\ell} \right)^2.
\label{EQ:refit}
\end{align}
Then the refitted spline estimator for nonlinear functions $\phi_{\ell}(\cdot)$ is
\begin{equation}
    \widehat{\phi}^{\ast}_{\ell}\left( x_{\ell}\right) = \mathbf{B}_{\ell}^{(d)}\left( x_{\ell}\right)\widehat{\bs{\gamma}}^{\ast}_{\ell}, \quad \ell \in {\widehat{\mathcal{S}}}_{x, N}.
\label{DEF:phihatstar}
\end{equation}

Next we establish the asymptotic normal distribution for the parametric estimators. To make $\bs{\beta}_{0, \mathcal{S}_Z}$ estimable at the $\sqrt{n}$ rate, we need a condition to ensure $\bs{X}$ and $\bs{Z}$ are not functionally related. Define $\mathcal{F}_{+}=\left\{f(\bs{x}) =\sum_{\ell \in \mathcal{S}_{x,N}} f_{\ell}(x_{\ell}),~ \mathrm{E}\{f_{\ell}(X_{\ell})\} =0,~\left \| f_{\ell}\right\| _{2}<\infty \right\} $ as the Hilbert space of theoretically centered $L_{2}$ additive functions.  For any $k \in \mathcal{S}_z$, let $z_{k}$ be the coordinate mapping that maps $\bs{Z}$ to its $k$-th component so that $z_{k}(\bs{Z})=Z_{k}$, and let $\psi_k^z=\mathrm{argmin}_{\psi\in \mathcal{F}_{+}}\|z_{k}-\psi\|_{2}^{2}=\mathrm{argmin}_{\psi\in \mathcal{F}_{+}}\mathrm{E}\{Z_{k}-\psi(\bs{X})\}^{2}$ be the orthogonal projection of $z_{k}$ onto $\mathcal{F}_{+}$. Let $\widetilde{\bs{Z}}_{\mathcal{S}_z}=\left\{\psi_k^z(\bs{X}), k \in \mathcal{S}_z \right\}^{\top}$. Similarly, for any $\ell \in \mathcal{S}_{x,PL}$, let $x_{\ell}$ be the coordinate mapping that maps $\bs{X}$ to its $\ell$-th component so that $x_{\ell}(\bs{X})=X_{\ell}$, and let
\begin{equation}
	\psi_{\ell}^x=\mathrm{argmin}_{\psi\in \mathcal{F}_{+}}\|x_{\ell}-\psi\|_{2}^{2}=\mathrm{argmin}_{\psi\in \mathcal{F}_{+}}\mathrm{E}\{X_{\ell}-\psi(\bs{X})\}^{2}
\label{DEF:psi}
\end{equation}
be the orthogonal projection of $x_{\ell}$ onto $\mathcal{F}_{+}$. Let $\widetilde{\bs{X}}_{\mathcal{S}_{x, PL}}=\left\{\psi_{\ell}^x(\bs{X}), \ell \in \mathcal{S}_{x,PL} \right\}^{\top}$. Define $\bs{Z}_{\mathcal{S}_z} = \left( \bs{Z}_k, k \in \mathcal{S}_z \right)^{\top}$ and $\bs{X}_{\mathcal{S}_{x, PL}} = \left( \bs{X}_{\ell}, \ell \in \mathcal{S}_{x, PL} \right)^{\top}$. Denote vector $\bs{T}$ and $\widetilde{\bs{T}}$ as $\bs{T} = (\bs{Z}_{\mathcal{S}_z}, \bs{X}_{\mathcal{S}_{x, PL}})^{\top}$, $\widetilde{\bs{T}} = \left(\widetilde{\bs{Z}}_{\mathcal{S}_z}, \widetilde{\bs{X}}_{\mathcal{S}_{x, PL}} \right)$.

\begin{theorem}
\label{THM:normality}
Under the Assumptions (A1), (A2)--(A6), (A3$^{\prime}$) and (A6$^{\prime}$) in Appendix \ref{SUBSEC:assump},
\[
	(n \bs{\Sigma})^{1/2}
	\begin{pmatrix}
		\widehat{\bs{\alpha}}_{\mathcal{S}_z}^{\ast} - \bs{\alpha}_{0, \mathcal{S}_z} \\
		\hbbeta_{\mathcal{S}_{x, PL}}^{\ast} - \bs{\beta}_{0, \mathcal{S}_{x, PL}}
	\end{pmatrix}
	\stackrel{D}{%
\longrightarrow } \mathcal{N}(\bs{0}, \mathbf{I}),
\]
where  $\mathbf{I}$ is an identity matrix and $\bs{\Sigma}=\sigma^{-2}{\rm E}[(\bs{T}-\widetilde{\bs{T}}) (\bs{T}-\widetilde{\bs{T}})^{\top}]$.
\end{theorem}
The proof of Theorem \ref{THM:normality} is similar to the proof of \cite{liu2011estimation} and \cite{li2017ultra} and thus omitted. {Let $\mathbf{Z}_{\mathcal{S}_z}=(Z_{ik},k\in \mathcal{S}_z)_{i=1}^{n}$ and $\mathbf{B}_{\mathcal{S}}^{(d)} = (B_{J,\ell}^{(d)}(X_{i\ell}), 1 \leq \ell \leq p_2$, $\ell \in \mathcal{S}_{x,N}, J = 1, \ldots, N_n)_{i=1}^{n}$.
If $\mathcal{S}_z$ and $\mathcal{S}_x$ are given, $\bs{\Sigma}$ can be consistently estimated by
$\widehat{\bs{\Sigma}}_{n}=(n\widehat{\sigma}^{2})^{-1}(\mathbf{Z}_{\mathcal{S}_z}-\widehat{\mathbf{Z}}_{\mathcal{S}_z})^{\top}(\mathbf{Z}_{\mathcal{S}_z}-\widehat{\mathbf{Z}}_{\mathcal{S}_z})$,
where $\widehat{\mathbf{Z}}_{\mathcal{S}_z}^{\top}=\mathbf{Z}_{\mathcal{S}_z}^{\top}\mathbf{B}_{\mathcal{S}}^{(d)}\mathbf{U}_{22}^{-1} \mathbf{B}_{\mathcal{S}}^{(d)\top}$ with $\mathbf{U}_{22}$ given in (\ref{EQN:us}) in the Supplemental Materials and $\widehat{\sigma}^{2}=(n-|\mathcal{S}_z|-|\mathcal{S}_x|)^{-1}\|\mathbf{Y}-\widehat{\mathbf{Y}}\|^{2}$.
In practice, we replace $\mathcal{S}_z$ and $\mathcal{S}_x$ with $\widehat{\mathcal{S}}_z$ and $\widehat{\mathcal{S}}_x$, respectively, to obtain the corresponding estimate.}

Let $\Omega_n = \{\widehat{\mathcal{S}}_z = \mathcal{S}_z, \widehat{\mathcal{S}}_{x, PL} = \mathcal{S}_{x, PL} \}$. In the selection step, we estimate $\mathcal{S}_z$ and $\mathcal{S}_{x, PL}$ consistently, that is, $P\left(\Omega_n\right) \rightarrow 1$. Within the event $\Omega_n$, that is, $\widehat{\mathcal{S}}_z = \mathcal{S}_z$ and $\widehat{\mathcal{S}}_{x, PL} = \mathcal{S}_{x, PL}$,
the estimator $(\widehat{\bs{\alpha}}_{\mathcal{S}_z}^{\ast\top}, \hbbeta_{\mathcal{S}_{x, PL}}^{\ast\top})^{\top}$ is root-$n$ consistent according to Theorem \ref{THM:normality}.
Since $\Omega_n$ is shown to have probability tending to one, we can conclude that $(\widehat{\bs{\alpha}}_{\widehat{\mathcal{S}}_z}^{\ast\top}, \hbbeta_{\widehat{\mathcal{S}}_{x, PL}}^{\ast\top})^{\top}$ is also root-$n$ consistent.

These refitted pilot estimators defined in (\ref{EQ:refit}) and (\ref{DEF:phihatstar}) are then used to define new pseudo-responses $\widehat{Y}%
_{i\ell}$, which are estimates of the unobservable ``oracle'' responses $Y_{i\ell}$%
. Specifically,
\begin{align}
    \widehat{Y}_{i\ell}=Y_{i}-\left\{
    \sum_{k \in \widehat{\mathcal{S}}_z} Z_{ik}\widehat{\alpha}^{\ast}_k +
    \sum_{\ell^{\prime} \in \widehat{\mathcal{S}}_{x, PL}} X_{i\ell^{\prime}}\widehat{\beta}_{\ell^{\prime}}^{\ast} +
    \sum_{\ell^{\prime\prime} \in \widehat{\mathcal{S}}_{x, N} \setminus \{\ell\}}  \widehat{\phi}_{\ell^{\prime\prime}}^{\ast}\left(
    X_{i\ell^{\prime\prime}}\right)\right\} , \nonumber \\
    Y_{i\ell}=Y_{i}-\left\{
    \sum_{k \in \mathcal{S}_z} Z_{ik}\alpha_{0k} +
    \sum_{\ell^{\prime} \in \mathcal{S}_{x, PL}} X_{i\ell^{\prime}}\beta_{0\ell^{\prime}} +
    \sum_{\ell^{\prime\prime} \in \mathcal{S}_{x, N} \setminus \{\ell\}} \phi_{0\ell^{\prime\prime}}\left(
    X_{i\ell^{\prime\prime}}\right)\right\} .
\label{DEF:YilhatYil}
\end{align}
Denote $K(\cdot)$ a continuous kernel function, and let $K_{h_{\ell}}(t)=K(t/h)/h$ be a rescaling of $K$, where $h$ is usually called the bandwidth.
Next, we define the spline-backfitted local-linear (SBLL) estimator of $\phi_{\ell}\left(
x_{\ell}\right) $ as $\widehat{\phi}_{\ell}^{\mathrm{SBLL}}\left( x_{\ell}\right) $ based on $\left\{
X_{i\ell},\widehat{Y}_{i\ell}\right\} _{i=1}^{n}$, which attempts to mimic the
would-be SBLL estimator $\widehat{\phi}_{\ell}^o\left( x_{\ell}\right) $
of $\phi_{\ell}\left( x_{\ell}\right) $ based on $\left\{X_{i\ell},Y_{i\ell}\right\}
_{i=1}^{n}$ if the unobservable ``oracle'' responses $\left\{ Y_{i\ell}\right\}
_{i=1}^{n}$ were available:
\begin{equation}
\left(\widehat{\phi}_{\ell}^o\left( x_{\ell}\right), \widehat{\phi}_{\ell}^{\mathrm{SBLL}}\left( x_{\ell}\right) \right)= \left(1 ~~ 0\right) \left(\mathbf{X}_{\ell}^{\ast\top}\mathbf{W}_{\ell}\mathbf{X}_{\ell}^{\ast}\right)^{-1} \mathbf{X}_{\ell}^{\ast\top}\mathbf{W}_{\ell} (\mathbf{Y}_{\ell}, \widehat{\mathbf{Y}}_{\ell}),
\label{DEF:alphahat-SBLL}
\end{equation}
where $\mathbf{Y}_{\ell} = \left(Y_{1\ell}, \ldots, Y_{n\ell}\right)^{\top}$ and $\widehat{\mathbf{Y}}_{\ell} = (\widehat{Y}_{1\ell}, \ldots, $ $\widehat{Y}_{n\ell})^{\top}$, with $\widehat{Y}_{i\ell}$ and $Y_{i\ell}$ as defined in (\ref{DEF:YilhatYil}), respectively; and the weight and ``design'' matrices are
\[
    \mathbf{W}_{\ell} = n^{-1}\text{diag} \{K_{h_{\ell}}(X_{i\ell} - x_{\ell})\}_{i = 1}^n, \quad \mathbf{X}_{\ell}^{\ast\top} =
     \begin{pmatrix}
         1 & ,\ldots  , & 1 \\
         X_{1\ell} - x_{\ell} & , \ldots , & X_{n\ell} - x_{\ell}
     \end{pmatrix}.
\]

Asymptotic properties of smoothers of $\widehat{\phi}_{\ell}^o\left( x_{\ell}\right), ~ \ell\in \mathcal{S}_{x, N}$, can be easily established. Specifically, let $\mu_{2}(K)=\int u^{2}K\left( u\right) du$, and let $f_{\ell}$ be the probability density function of $X_{\ell}$, then under Assumptions (B1) and (B2) in Appendix \ref{SUBSEC:assump},
\begin{equation}
\sqrt{nh_{\ell}}\left\{\widehat{\phi}_{\ell}^o\left( x_{\ell}\right) -\phi_{0\ell}(x_{\ell}) -b_{\ell}(x_{\ell}) h_{\ell}^{2}\right\} \stackrel{D}{%
\longrightarrow }N\left\{0,v_{\ell}^{2}(x_{\ell}) \right\} , ~ \ell\in \mathcal{S}_{x,N},
\label{EQN:alphao_normal}
\end{equation}
where
\begin{equation}
b_{\ell}(x_{\ell}) = \mu_{2}(K)
\phi_{0\ell}^{\prime \prime }(x_{\ell})/2 , \quad
v_{\ell}^{2}(x_{\ell}) = \|K\|_2^2
f_{\ell}^{-1}(x_{\ell}) \sigma^{2}.
\label{EQ:oracle}
\end{equation}

The following theorem states that the asymptotic uniform magnitude of the difference between $\widehat{\phi}_{\ell}^{\mathrm{SBLL}}\left( x_{\ell}\right) $ and $\widehat{\phi}_{\ell}^o\left( x_{\ell}\right) $ is of order $o_{P}\{(nh_{\ell})^{-1/2}\} $, which is dominated by the asymptotic uniform size of $\widehat{\phi}_{\ell}^o\left( x_{\ell}\right) -\phi_{0\ell}(x_{\ell})$. As a result, $\widehat{\phi}_{\ell}^{\mathrm{SBLL}}\left( x_{\ell}\right) $ will have the same asymptotic distribution as $\widehat{\phi}_{\ell}^o\left( x_{\ell}\right) $. We say $x_{\ell}\in\chi_{\ell}$ is a boundary point if and only if $x_{\ell}=a+ch_{\ell}$ or $x_{\ell}=b-ch_{\ell}$ for some $0\leq c<1$ and an interior point otherwise. Let $\chi_{h_{\ell}}$ be the interior of the support $\chi$.



\begin{theorem}
\label{THM:nonlinear-normality}
Suppose the assumptions in Theorem \ref{THM:normality} hold. In addition, if Assumptions (B1) and (B2) in Appendix \ref{SUBSEC:assump} are satisfied, then the SBLL estimator $\widehat{\phi}_{\ell}^{\mathrm{SBLL}}\left( x_{\ell}\right)$ given in (\ref{DEF:alphahat-SBLL}) satisfies
\begin{equation}
\label{EQN:SBLL_Order}
\sup_{x_{\ell}\in \chi_{h_{\ell}}}\left| \widehat{\phi}_{\ell}^{\mathrm{SBLL}}\left( x_{\ell}\right) -\widehat{\phi}_{\ell}^o\left( x_{\ell}\right)  \right| =o_{P}\{(nh_{\ell})^{-1/2}\}, ~ \ell\in \mathcal{S}_{x,N}.
\end{equation}
Hence with $b_{\ell}(x_{\ell}) $ and $v_{\ell}^{2}(x_{\ell})$ as defined in (\ref{EQ:oracle}), for any $x_{\ell}$ in its interior support $x_{\ell} \in \chi_{h_{\ell}}$,%
\begin{equation}
\sqrt{nh_{\ell}}\left\{\widehat{\phi}_{\ell}^{\mathrm{SBLL}}\left( x_{\ell}\right) -\phi_{0\ell}(x_{\ell}) -b_{\ell}(x_{\ell}) h_{\ell}^{2}\right\} \stackrel{D}{\longrightarrow }\mathcal{N} \left\{0,v_{\ell}^{2}(x_{\ell}) \right\} , ~ \ell\in \mathcal{S}_{x,N}.
\label{EQN:sbll_normal}
\end{equation}
In addition, the estimator $\widehat{\phi}_{\ell}^{\mathrm{SBLL}}\left( x_{\ell}\right)$ satisfies, for any $t$ and $ \ell\in \mathcal{S}_{x,N}$,
\begin{align}
    \lim_{n \rightarrow \infty} & \Pr\left\{\sqrt{\ln(h_{\ell}^{-2})} \left(\sup_{x_{\ell} \in \chi_{h_{\ell}}} \frac{\sqrt{nh_{\ell}}}{v_{\ell}(x_{\ell})} |\widehat{\phi}_{\ell}^{\mathrm{SBLL}}\left( x_{\ell}\right) -\phi_{0\ell}(x_{\ell})| - \tau_n\right) < t\right\}
    = e^{-2 e^{-t}},
\label{EQN:scb_form}
\end{align}
where $\tau_n = \sqrt{\ln(h_{\ell}^{-2})} + \ln\{\|K^{\prime}\|_2 / (2\pi \|K\|_2)\} / \sqrt{\ln(h_{\ell}^{-2})}$.
\end{theorem}

Theorem \ref{THM:nonlinear-normality} provides analytical expressions for constructing asymptotic confidence intervals and SCBs under certain conditions. Under Assumptions (A1)--(A6), (A3$^{\prime}$), (A6$^{\prime}$), (B1) and (B2) in Appendix \ref{SUBSEC:assump}, for any $\alpha\in(0,1)$, an asymptotic $100(1-\alpha)\%$ pointwise confidence interval for $\phi_{0\ell}(x_{\ell})$ over the interval $\chi_{h_{\ell}}$ is
\[
\widehat{\phi}_{\ell}^{\mathrm{SBLL}}\left( x_{\ell}\right) -\widehat{b}_{\ell}(x_{\ell})h_{\ell}^{2}\pm \widehat{v}_{\ell}(x_{\ell})(nh_{\ell})^{-1/2}, ~ \ell\in \mathcal{S}_{x,N}.
\]
Under Assumptions (A1)--(A6), (A2$^\prime$) (A3$^{\prime}$), (A6$^{\prime}$), (B1) and (B2)  in the Appendix, for any $\alpha\in(0,1)$, an asymptotic $100(1-\alpha)\%$ SCB for $\phi_{0\ell}(x_{\ell})$ over the interval $\chi_{h_{\ell}}$ is
\[
\widehat{\phi}_{\ell}^{\mathrm{SBLL}}\left( x_{\ell}\right) \pm \widehat{v}_{\ell}(x_{\ell})(nh_{\ell})^{-1/2}\left[\tau_n - \{\ln(h_{\ell}^{-2})\}^{-1/2} \ln\left\{-\frac{1}{2} \ln(1 - \alpha)\right\} \right], ~ \ell\in \mathcal{S}_{x,N}.
\]

\setcounter{chapter}{4}
\renewcommand{\thesection}{\arabic{section}}
\renewcommand{\thetable}{{\arabic{table}}} \setcounter{table}{0}
\renewcommand{\thefigure}{\arabic{figure}} \setcounter{figure}{0}
\renewcommand{\thesubsection}{4.\arabic{subsection}}
\setcounter{section}{3}

\section{Implementation and Simulation} \label{sec:simulation}

In this section we discuss practical implementations for the SMILE procedure.
To meet the zero mean requirement specified in Assumption (A4),
we use the centralized $X_{i\ell}^{\ast}$ instead of $X_{i\ell}$ directly, for each $\ell=1, \ldots, p_2$.
At the risk of abusing the notation, we still use symbol $X$ instead of $X^{\ast}$ to avoid creating too many new symbols. To implement the proposed procedure, one needs to select the penalty parameters, the knots for a spline at the selection stage and refitting stage, and the bandwidth for a kernel at the backfitting stage.

\textit{Knot selection.} For spline smoothing involved in both selection and refitting, we suggest placing knots on a grid of evenly spaced sample quantiles. Based on extensive simulation experiments in Section A of the Supplementary Materials, we find that the number of knots often has little effect on the model selection results. Therefore, we recommend using a small number of knots at the model selection stage to reduce the computing cost, especially when the sample size is too small compared to the number of covariates. In practice, $2\sim 5$ interior knots is usually adequate to identify the model structure.

At the refitting stage, Assumption (A6$^\prime$) in the Supplementary Materials suggests the number of interior knots $M_n$ for a refitting spline needs to satisfy: $\{n^{1/(2d)} \vee n^{4/(10d-5)}\} \ll M_n \ll n^{1/3}$, where $d$ is the degree of the polynomial spline basis functions used in the refitting. The widely used quadratic/cubic splines and any polynomial splines of degree $d\geq 2$ all satisfy this condition. Therefore, in practice we suggest take the following rule-of-thumb number of interior knots
\begin{equation*}
\min\{\lfloor n^{1/(2d) \vee 4/(10d-5)}\ln(n)\rfloor, \lfloor n/(4s)\rfloor\}+1,
\end{equation*}
where $s$ is the number of nonlinear components selected at the first stage, and the term $\lfloor n/(4s)\rfloor$ is to guarantee that we have at least four observations in each subinterval between two adjacent knots to avoid getting (near) singular design matrices in the spline refitting.

\textit{Bandwidth selection.} Note that Condition (B2) in the Supplementary Materials requires that the bandwidths in the backfitting are of order $n^{-1/5}$. Thus, the bandwidth selection can be done using a standard routine in the literature. In our numerical studies, we find that the rule-of-thumb bandwidth selector \citep{Fan:Gijbels:96} often works very well in both estimation and SCB construction.

Section A in the Supplementary Materials provides detailed investigations on how the smoothing parameters affect the proposed SMILE method and evaluates the practical performance in finite-sample simulation studies. Next we present our algorithm and discuss how to choose the penalty parameters.

\subsection{Algorithm} \label{subsec:algorithm}

In this section we discuss practical implementations for the SMILE procedure.
To meet the zero mean requirement specified in Assumption (A4),
we use the centralized $X_{i\ell}^{\ast}$ instead of $X_{i\ell}$ directly, for each $\ell=1, \ldots, p_2$.
At the risk of abusing the notation, we still use symbol $X$ instead of $X^{\ast}$ to avoid creating too many new symbols. To implement the proposed procedure, one needs to select the penalty parameters, the knots for a spline at the selection stage and refitting stage, and the bandwidth for a kernel at the backfitting stage.

\textit{Knot selection.} For spline smoothing involved in both selection and refitting, we suggest placing knots on a grid of evenly spaced sample quantiles. Based on extensive simulation experiments in Section A of the Appendix, we find that the number of knots often has little effect on the model selection results. Therefore, we recommend using a small number of knots at the model selection stage to reduce the computing cost, especially when the sample size is small compared to the number of covariates. In practice, $2\sim 5$ interior knots is usually adequate to identify the model structure.

At the refitting stage, Assumption (A6$^\prime$) in the Appendix suggests the number of interior knots $M_n$ for a refitting spline needs to satisfy: $\{n^{1/(2d)} \vee n^{4/(10d-5)}\} \ll M_n \ll n^{1/3}$, where $d$ is the degree of the polynomial spline basis functions used in the refitting. The widely used quadratic/cubic splines and any polynomial splines of degree $d\geq 2$ all satisfy this condition. Therefore, in practice we suggest take the following rule-of-thumb number of interior knots
\begin{equation*}
\min\{\lfloor n^{1/(2d) \vee 4/(10d-5)}\ln(n)\rfloor, \lfloor n/(4s)\rfloor\}+1,
\end{equation*}
where $s$ is the number of nonlinear components selected at the first stage, and the term $\lfloor n/(4s)\rfloor$ is to guarantee that we have at least four observations in each subinterval between two adjacent knots to avoid (near) singular design matrices in the spline refitting.

\textit{Bandwidth selection.} Note that Condition (B2) in the Appendix requires that the bandwidths in the backfitting are of order $n^{-1/5}$. Thus, the bandwidth selection can be done using a standard routine in the literature. In our numerical studies, we find that the rule-of-thumb bandwidth selector \citep{Fan:Gijbels:96} often works very well in both estimation and SCB construction.

Section A in the Appendix provides detailed investigations on how the smoothing parameters affect the proposed SMILE method and evaluates the practical performance in finite-sample simulation studies. Next we present our algorithm and discuss how to choose the penalty parameters.

\subsection{Algorithm} \label{subsec:algorithm}
The minimization of (\ref{EQ:loss_aglasso_spline}) can be solved by the group coordinate descent algorithm \citep{huang2012selective}, implemented using \textsf{R} package \texttt{grpreg} \citep{grpreg}. As for the selection of penalty parameters, we consider two criteria widely used in high-dimensional settings, modified Bayesian information criteria \citep[BIC; see][]{lee2014model} and the extended BIC \citep[EBIC; see][]{chen2008extended, chen2009tournament}:
\begin{align*}
	\text{BIC}(\bs{\lambda}) &= \ln(RSS_{\bs{\lambda}}) + df_{\bs{\lambda}} \times \frac{\ln (p_1 + p_2  + p_2N_n) \times \ln(n)}{2n}, \\
	\text{EBIC}(\bs{\lambda}) &= \ln(RSS_{\bs{\lambda}}) + df_{\bs{\lambda}} \times \frac{\ln(n)}{n} + df_{\bs{\lambda}} \times \frac{\ln (p_1 + p_2  + p_2N_n)}{n},
\end{align*}
where $RSS_{\bs{\lambda}}$ is the residual sum of squares associated with penalty parameters $\bs{\lambda} = (\lambda_1, \lambda_2,$ $\lambda_3)^{\top}$ and $df_{\bs{\lambda}}$ is the number of estimated nonzero coefficients for the given $\bs{\lambda}$. The simulation results are similar based on these two criteria, so in the following, we choose $\lambda_1$ and $\lambda_2$ by modified BIC and $\lambda_3$ by EBIC for illustration using an approach described below.

The classical coordinate descent algorithm deals with the optimization problem with one tuning parameter, and there are several ways to address the triple-penalization or multiple-penalization issue. A natural idea is to solve the optimization problem by searching over a three-dimensional grid for tuning parameters, which can be computationally expensive.
To pose a balance between computational efficiency and precision, we propose to solve the triple-penalization problem in two steps. In the first step, BIC is minimized with a common smoothing parameter $\lambda$, i.e., we set $\lambda_1=\lambda_2=\lambda_3=\lambda$, and we choose $\lambda$ by minimizing BIC($\lambda$) over a grid of $\lambda$ values. Using the selected common smoothing parameter, we obtain the initial estimators $\widehat{\bs{\alpha}}^{(0)}$, $\widehat{\bs{\beta}}^{(0)}$ and $\widehat{\bs{\gamma}}^{(0)}$. In Step 2, $\bs{\alpha}$, $\bs{\beta}$ and $\bs{\gamma}$ estimates are obtained one at a time by minimizing (\ref{EQ:loss_aglasso_spline}). More precisely, an $\bs{\alpha}$ estimate is obtained with $\bs{\beta}$, $\bs{\gamma}$ fixed at current estimates, where $\lambda_1$ is set equal to its minimum BIC value and $\lambda_2=\lambda_3=0$. One cycles in this way through $\bs{\alpha}$, $\bs{\beta}$ and $\bs{\gamma}$ estimation steps for a fixed number of iterations. Three iterations generally works well in practice. Algorithm 1 outlines the iterative group coordinate descent algorithm.
\normalem
\begin{algorithm}
\SetKwInOut{Input}{Input}
\SetKwInOut{Output}{Output}
\caption{Iterative group coordinate descent algorithm}
\Input{Data $\left\{(Y_i, Z_{i1},\ldots,Z_{ip_1}, X_{i1},\ldots,X_{ip_2},\mathbf{B}_{i1}^{(1)},\ldots,\mathbf{B}_{ip_2}^{(1)})\right\}_{i=1}^n$\\
$\widehat{\bs{\alpha}}^{(0)}$, $\widehat{\bs{\beta}}^{(0)}$ and $\widehat{\bs{\gamma}}^{(0)}$: initial parameters of interest \\
$\delta_0$: convergence criterion\\}
\Output{$\widehat{\bs{\alpha}}$, $\widehat{\bs{\beta}}$ and $\widehat{\bs{\gamma}}$: Estimates of $\bs{\alpha}$, $\bs{\beta}$ and $\bs{\gamma}$}
\BlankLine
\While{$\left\|\left(\widehat{\bs{\alpha}}^{(m+1)\top}, \widehat{\bs{\beta}}^{(m+1)\top}, \widehat{\bs{\gamma}}^{(m+1)\top}\right)^{\top}-\left(\widehat{\bs{\alpha}}^{(m)\top}, \widehat{\bs{\beta}}^{(m)\top}, \widehat{\bs{\gamma}}^{(m)\top}\right)^{\top}\right\|^2 > \delta_0$}{
(i) Given $\widehat{\bs{\beta}}^{(m)}$ and $\widehat{\bs{\gamma}}^{(m)}$, obtain $w_1^{\alpha\,(m+1)},\ldots, w_{p_1}^{\alpha\,(m+1)}$ by minimizing objective function (\ref{EQ:loss_glasso_spline}) with $\widetilde\lambda_1$ selected via the modified BIC;\\
(ii) Given $\widehat{\bs{\beta}}^{(m)}$, $\widehat{\bs{\gamma}}^{(m)}$ and $w_1^{\alpha\,(m+1)},\ldots, w_{p_1}^{\alpha\,(m+1)}$, obtain $\widehat{\bs{\alpha}}^{(m+1)}$ by minimizing objective function (\ref{EQ:loss_aglasso_spline}) with $\lambda_1$ selected via the modified BIC;\\
(iii) Given $\widehat{\bs{\alpha}}^{(m+1)}$ and $\widehat{\bs{\gamma}}^{(m)}$, obtain $w_1^{\beta\,(m+1)},\ldots, w_{p_2}^{\beta\,(m+1)}$ by minimizing objective function (\ref{EQ:loss_glasso_spline}) with $\widetilde\lambda_2$ selected via the modified BIC;\\
(iv) Given $\widehat{\bs{\alpha}}^{(m+1)}$, $\widehat{\bs{\gamma}}^{(m)}$ and $w_1^{\beta\,(m+1)},\ldots, w_{p_2}^{\beta\,(m+1)}$, obtain $\widehat{\bs{\beta}}^{(m+1)}$ by minimizing objective function (\ref{EQ:loss_aglasso_spline}) with $\lambda_2$ selected via the modified BIC;\\
(v) Given $\widehat{\bs{\alpha}}^{(m+1)}$ and $\widehat{\bs{\beta}}^{(m+1)}$, obtain $w_1^{\gamma\,(m+1)},\ldots, w_{p_2}^{\gamma\,(m+1)}$ by minimizing objective function (\ref{EQ:loss_glasso_spline}) with $\widetilde\lambda_3$ selected via EBIC;\\
(vi) Given $\widehat{\bs{\alpha}}^{(m+1)}$, $\widehat{\bs{\beta}}^{(m+1)}$ and $w_1^{\gamma\,(m+1)},\ldots, w_{p_2}^{\gamma\,(m+1)}$, obtain $\widehat{\bs{\gamma}}^{(m+1)}$ by minimizing objective function (\ref{EQ:loss_aglasso_spline}) with $\lambda_3$ selected via EBIC.}
\vskip -.1in
Set $\widehat{\bs{\alpha}}=\widehat{\bs{\alpha}}^{(m+1)}$, $\widehat{\bs{\beta}}=\widehat{\bs{\beta}}^{(m+1)}$ and $\widehat{\bs{\gamma}}=\widehat{\bs{\gamma}}^{(m+1)}$.
\label{ALGO:igcd}
\end{algorithm}
\ULforem

\subsection{Simulation Studies}
In this section, we investigate the performance of the proposed sparse model identification and learning estimator, abbreviated as SMILE, in terms of model selection, estimation accuracy and inference performance in a simulation study. We compare SMILE with the sparse APLM estimator with adaptive group LASSO penalty (SAPLM) proposed in \cite{li2017ultra}, the ordinary linear least squares estimator with the adaptive LASSO penalty (SLM), and the oracle estimator (ORACLE), which uses the same estimation techniques as SMILE except that no penalization or data-driven variable selection is used because all active and inactive index sets are treated as known. Note that SAPLM ignores the potential linear structure in covariate $\bs{X}$, and estimates the effects of each component of $\bs{X}$ with all nonparametric forms; in contrast, SLM ignores the potential nonlinear structure in covariate $\bs{X}$ and requires selected components of covariates $\bs{Z}$ and $\bs{X}$ to enter the model in a linear form. In terms of the performances of SCBs, we compare SMILE with SAPLM and ORACLE. In our simulation, ORACLE works as a benchmark for estimation comparison. It is worth pointing out that the ORACLE estimator is only computable in simulations, not real examples.

We generate simulated datasets using the APLM structure
\begin{align*}
    Y_i = \sum_{k=1}^{p_1}Z_{ik}\alpha_k + \sum_{\ell=1}^{p_2} \phi_{\ell}(X_{i\ell}) + \varepsilon_i,
\end{align*}
where $\alpha_1 = 3$, $\alpha_2 = 4$, $\alpha_3 = -2$, $\alpha_4 = \ldots = \alpha_{p_1} = 0$,
$\phi_1(x) = 9x$, $\phi_2(x) = -1.5 \cos^2(\pi x) + 3 \sin^2(\pi x) - \mathrm{E}\{-1.5 \cos^2(\pi X_2 + 3 \sin^2(\pi X_2)\}$, $\phi_3(x) = 6x + 18x^{2} - \mathrm{E}(6X_3 + 18X_3^{2})$, and $\phi_4(x) = \ldots = \phi_{p_2}(x) = 0$.
     Notice that $\phi_1(x)$ is actually a linear function. So there are three variables in the active index set for $\bs{Z}$, one variable in the active pure linear index set for $\bs{X}$, one variable in the active pure nonlinear index set for $\bs{X}$, and one variable in the active linear \& nonlinear index set for $\bs{X}$.

%
We simulate $Z_{ik}^{\ast}$ independently from the $\mathrm{Unif}[0, 1]$ and $X_{i\ell}$ independently from the $\mathrm{Unif}[-.5,$ $.5]$, and set $Z_{ik} = I(Z_{ik}^{\ast} > 0.75)$, for $i = 1, \ldots, n$, $k = 1, \ldots, p_1$, $\ell=1, \ldots, p_2$.
To make an ultra-high-dimensional scenario, we let the sample size $n = 300$ and $n = 500$, and consider three different dimensions: $p_1= p_2= p$, where $p$ is taken to be $1000$,  $2000$ and $5000$. The error term $\varepsilon_i$ is simulated from $\mathcal{N}(0, \sigma^{2})$ with $\sigma= 0.5$ and $1.0$.

To approximate the nonlinear functions, we use the constant B-spline ($d=1$) with four interior knots for selection and use the cubic B-spline ($d=4$) with four interior knots in the refitting step. For both selection and refitting, the knots are on a grid of evenly spaced sample quantiles. To construct the SCBs, in our simulation studies below,  we choose the Epanechnikov kernel function with the rule-of-thumb bandwidth described in Section 4.2 in \cite{Fan:Gijbels:96}, which usually works well in our experimental investigation. More simulation studies have been conducted with different choices for spline knots and kernel bandwidth selectors; see Section A of the Appendix.

We evaluate the methods on the accuracy of variable selection, prediction and inference. In detail, we adopt the following criteria for evaluation:
\begin{enumerate} [{\normalfont (B-i)}]
    \item Percent of covariates in $\bs{Z}$ with nonzero linear coefficients that are correctly identified (``CorrZ''); \\[-18pt]
    \item Percent of covariates in $\bs{Z}$ with zero linear coefficients that are correctly identified (``CorrZ0'');\\[-18pt]
    \item Percent of covariates in $\bs{X}$ with nonzero purely linear functions that are correctly identified (``CorrL'');\\[-18pt]
    \item Percent of covariates in $\bs{X}$ with nonzero purely nonlinear functions that are correctly identified (``CorrN''); \\[-18pt]
    \item Percent of covariates in $\bs{X}$ with nonzero linear and nonlinear functions that are correctly identified (``CorrLN'); \\[-18pt]
    \item Percent of covariates in $\bs{X}$ with zero functions that are correctly identified (``CorrX0'');
\end{enumerate}
\begin{enumerate} [{\normalfont (C-i)}]
    \item Percent of covariates in $\bs{Z}$ with nonzero linear coefficients incorrectly identified as having zero linear coefficients (``Zto0'');\\[-18pt]
    \item Percent of covariates in $\bs{X}$ with nonzero purely linear functions incorrectly identified as having nonlinear functions (``LtoN''); \\[-18pt]
    \item Percent of covariates in $\bs{X}$ with nonzero purely nonlinear functions incorrectly identified as having linear functions (``NtoL''); \\[-18pt]
    \item Percent of covariates in $\bs{X}$ with nonzero linear or nonzero nonlinear functions incorrectly identified as having both zero linear and zero nonlinear functions (``Xto0'');
\end{enumerate}
\begin{enumerate} [{\normalfont (D-i)}]
    \item Mean squared errors (MSE) for linear coefficients $\alpha_1$, $\alpha_2$, $\alpha_3$ and $\beta_1$; \\[-18pt]
    \item Average MSE (AMSE) for $\phi_1$, $\phi_2$ and $\phi_3$, defined as $n^{-1} \sum_{i = 1}^n \{\widehat{\phi}_{\ell}^{\mathrm{SBLL}}(x_{i\ell}) - \phi_{\ell}(x_{i\ell})\}^2$;\\[-18pt]
    \item 10-fold cross-validation mean squared prediction error (CV-MSPE) for the response variable, defined as $10^{-1}\sum_{m = 1}^{10} |\kappa_m|^{-1}\sum_{i \in \kappa_{m}} (\widehat{Y}_i - Y_i)^2$, where $\kappa_1, \ldots, \kappa_{10}$ comprise a random partition of the dataset into $10$ disjoint subsets of approximately equal size, and $\widehat{Y}_i$ is the prediction obtained from all data aside from the subset containing the $i$th observation;\\[-18pt]
    \item The coverage rates of the proposed 95\% SCB for functions $\phi_2$ and $\phi_3$ (Coverage).
\end{enumerate}
All these performance measures are computed based on 1000 replicates. Note that Criteria (B-i)--(B-vi) measure the frequency of getting the correct model structure; Criteria (C-i)--(C-iv) measure the frequency of getting an incorrect model structure; Criteria (D-i)--(D-iii) focus on the estimation and prediction accuracy for the model components; and Criterion (D-iv) measures the inferential performance.

The model selection results are provided in Tables \ref{TAB:SIM1-selection1} and \ref{TAB:SIM1-selection2}, respectively. SMILE can effectively identify informative linear and nonlinear components as well as correctly discover the linear and nonlinear structure in covariate $\bs{X}$, while SAPLM neglects linear structure in $\bs{X}$ and SLM fails in representing the nonlinear part of covariate $\bs{X}$. For SMILE, the numbers of correctly selected nonzero covariates in $\bs{Z}$, linear, nonlinear, linear-and-nonlinear components in $\bs{X}$, nonzero covariates are very close to ORACLE (100\% for corrZ, corrL, corrN, corrLN, corrZ0 and corrX0, respectively); and the numbers of incorrectly identified components approach to $0$ as the sample size $n$ increases, as shown in Table \ref{TAB:SIM1-selection2}. SMILE is close in the selection of covariates $\bs{Z}$ to the SAPLM estimator, and it far outperforms SAPLM in identifying the linear-and-nonlinear structure of covariate $\bs{X}$. From the results in Tables \ref{TAB:SIM1-selection1} and \ref{TAB:SIM1-selection2}, it is also evident that model misspecification leads to poor variable selection performance for SLM. Especially for the selection of covariates in $\bs{X}$, which is our main focus for real data analysis, SLM fails to select the right nonlinear components in each simulation.

\begin{table}[htbp]
\caption{Statistics (B-i)--(B-vi) comparing the SMILE, SAPLM and SLM.}
\label{TAB:SIM1-selection1}
\renewcommand{\arraystretch}{0.9}
\begin{center}
\begin{tabular}{ccclrrrrrr}
	\hline \hline
	Size & Noise & & & \multicolumn{2}{c}{Z Part} & \multicolumn{ 4}{c}{X Part} \\
	\cmidrule(r){5-6} \cmidrule(l){7-10}
	$n$ & sig & $p$ & Method & corrZ & corrZ0 & corrL & corrN & corrLN & corrX0 \\
	\hline
	300   & 0.5   & 1000  & SMILE & 100   & 99.99960 & 100   & 100   & 100   & 99.99940 \\
          &       &       & SAPLM & 100   & 100   & 0     & 100   & 0     & 100 \\
          &       &       & SLM   & 98.6  & 99.99920 & 100   & 0     & 0     & 99.99850 \\
          &       & 2000  & SMILE & 100   & 99.99995 & 100   & 100   & 100   & 99.99985 \\
          &       &       & SAPLM & 100   & 100   & 0     & 100   & 0     & 100 \\
          &       &       & SLM   & 97.3  & 99.99950 & 100   & 0     & 0     & 99.99915 \\
          &       & 5000  & SMILE & 100   & 99.99996 & 100   & 100   & 100   & 100 \\
          &       &       & SAPLM & 100   & 100   & 0     & 100   & 0     & 100 \\
          &       &       & SLM   & 96.63333 & 99.99988 & 100   & 0     & 0     & 99.99974 \\
          & 1.0   & 1000  & SMILE & 100   & 99.99920 & 100   & 100   & 100   & 99.99990 \\
          &       &       & SAPLM & 100   & 99.99920 & 0     & 100   & 0     & 100 \\
          &       &       & SLM   & 96.56667 & 99.99799 & 100   & 0     & 0     & 99.99719 \\
          &       & 2000  & SMILE & 99.93333 & 99.99995 & 100   & 99.8  & 99.8  & 99.99975 \\
          &       &       & SAPLM & 100   & 99.99970 & 0     & 100   & 0     & 100 \\
          &       &       & SLM   & 95.7  & 99.99975 & 100   & 0     & 0     & 99.99905 \\
          &       & 5000  & SMILE & 99.86667 & 99.99996 & 100   & 99.5  & 99.5  & 99.99996 \\
          &       &       & SAPLM & 100   & 99.99990 & 0     & 100   & 0     & 100 \\
          &       &       & SLM   & 93.73333 & 99.99982 & 100   & 0     & 0     & 99.99978 \\
    500   & 0.5   & 1000  & SMILE & 100   & 99.99990 & 100   & 100   & 100   & 99.99980 \\
          &       &       & SAPLM & 100   & 100   & 0     & 100   & 0     & 100 \\
          &       &       & SLM   & 100   & 99.99990 & 100   & 0     & 0     & 99.99960 \\
          &       & 2000  & SMILE & 100   & 99.99995 & 100   & 100   & 100   & 100 \\
          &       &       & SAPLM & 100   & 100   & 0     & 100   & 0     & 100 \\
          &       &       & SLM   & 100   & 99.99985 & 100   & 0     & 0     & 99.99985 \\
          &       & 5000  & SMILE & 100   & 99.99996 & 100   & 100   & 100   & 100 \\
          &       &       & SAPLM & 100   & 100   & 0     & 100   & 0     & 100 \\
          &       &       & SLM   & 99.96667 & 99.99994 & 100   & 0     & 0     & 99.99994 \\
          & 1.0   & 1000  & SMILE & 100   & 99.99950 & 100   & 100   & 100   & 99.99970 \\
          &       &       & SAPLM & 100   & 100   & 0     & 100   & 0     & 100 \\
          &       &       & SLM   & 99.96667 & 99.99940 & 100   & 0     & 0     & 99.99930 \\
          &       & 2000  & SMILE & 100   & 99.99980 & 100   & 100   & 100   & 99.99990 \\
          &       &       & SAPLM & 100   & 99.99990 & 0     & 100   & 0     & 100 \\
          &       &       & SLM   & 99.93333 & 99.99990 & 100   & 0     & 0     & 99.99960 \\
          &       & 5000  & SMILE & 100   & 99.99994 & 100   & 100   & 100   & 100 \\
          &       &       & SAPLM & 100   & 100   & 0     & 100   & 0     & 100 \\
          &       &       & SLM   & 99.76667 & 100   & 100   & 0     & 0     & 99.99994 \\
	\hline \hline
\end{tabular}
\end{center}
\end{table}

\begin{table}[htbp]
\caption{Statistics (C-i)--(C-iv) comparing the SMILE, SAPLM and SLM.}
\label{TAB:SIM1-selection2}
\renewcommand{\arraystretch}{0.9}
\begin{center}
\begin{tabular}{ccclrrrr}
	\hline \hline
	Size & Noise & & & Z Part & \multicolumn{3}{c}{X Part} \\
	\cmidrule(r){5-5} \cmidrule(l){6-8}
	  $n$ & sig & $p$ & Method & Zto0 & LtoN & NtoL & Xto0 \\
	\hline
	300   & 0.5   & 1000  & SMILE & 0     & 0     & 0     & 0 \\
          &       &       & SAPLM & 0     & 100   & 0     & 0 \\
          &       &       & SLM   & 1.4   & 0     & 100   & 33.33333 \\
          &       & 2000  & SMILE & 0     & 0     & 0     & 0 \\
          &       &       & SAPLM & 0     & 100   & 0     & 0 \\
          &       &       & SLM   & 2.7   & 0     & 100   & 33.33333 \\
          &       & 5000  & SMILE & 0     & 0     & 0     & 0 \\
          &       &       & SAPLM & 0     & 100   & 0     & 0 \\
          &       &       & SLM   & 3.36667 & 0     & 100   & 33.33333 \\
          & 1.0   & 1000  & SMILE & 0     & 0     & 0     & 0 \\
          &       &       & SAPLM & 0     & 100   & 0     & 0 \\
          &       &       & SLM   & 3.43333 & 0     & 100   & 33.33333 \\
          &       & 2000  & SMILE & 0.06667 & 0     & 0     & 0.06667 \\
          &       &       & SAPLM & 0     & 100   & 0     & 0 \\
          &       &       & SLM   & 4.3   & 0     & 100   & 33.33333 \\
          &       & 5000  & SMILE & 0.13333 & 0     & 0     & 0.16667 \\
          &       &       & SAPLM & 0     & 100   & 0     & 0 \\
          &       &       & SLM   & 6.26667 & 0     & 100   & 33.33333 \\
    500   & 0.5   & 1000  & SMILE & 0     & 0     & 0     & 0 \\
          &       &       & SAPLM & 0     & 100   & 0     & 0 \\
          &       &       & SLM   & 0     & 0     & 100   & 33.33333 \\
          &       & 2000  & SMILE & 0     & 0     & 0     & 0 \\
          &       &       & SAPLM & 0     & 100   & 0     & 0 \\
          &       &       & SLM   & 0     & 0     & 100   & 33.33333 \\
          &       & 5000  & SMILE & 0     & 0     & 0     & 0 \\
          &       &       & SAPLM & 0     & 100   & 0     & 0 \\
          &       &       & SLM   & 0.03333 & 0     & 100   & 33.33333 \\
          & 1.0   & 1000  & SMILE & 0     & 0     & 0     & 0 \\
          &       &       & SAPLM & 0     & 100   & 0     & 0 \\
          &       &       & SLM   & 0.03333 & 0     & 100   & 33.33333 \\
          &       & 2000  & SMILE & 0     & 0     & 0     & 0 \\
          &       &       & SAPLM & 0     & 100   & 0     & 0 \\
          &       &       & SLM   & 0.06667 & 0     & 100   & 33.33333 \\
          &       & 5000  & SMILE & 0     & 0     & 0     & 0 \\
          &       &       & SAPLM & 0     & 100   & 0     & 0 \\
          &       &       & SLM   & 0.23333 & 0     & 100   & 33.33333 \\
	\hline \hline
\end{tabular}
\end{center}
\end{table}

The estimation and prediction results are displayed in Table \ref{TAB:SIM1-estimation}. Specifically, we present the MSEs for linear coefficients $\alpha_1$, $\alpha_2$, $\alpha_3$ and $\beta_1$ and AMSEs for functions $\phi_1$, $\phi_2$ and $\phi_3$ and the CV-MSPEs for predicting $Y$. The case with known active covariates (ORACLE) is also reported in each setting and serves as a gold standard. SMILE performs the best in predicting $Y$ and estimating the coefficients of covariates $\bs{Z}$, as indicated by CV-MSPE and MSEs for $\alpha_1$, $\alpha_2$ and $\alpha_3$ that are closest to ORACLE in most simulation settings, while SLM is much higher (around 2 $\sim$ 18 times higher). As for the linear structure in $\bs{X}$, as shown in MSE for $\beta_1$ and AMSE for $\phi_1$, the performance of SMILE is comparable to SAPLM and SLM, even though restricted to the selection bias; as the sample size $n$ increases, the performance of SMILE is perfect and matches with ORACLE. Note that the SAPLM estimator is incapable in estimating $\beta_1$ in this case. The estimation of nonlinear functions $\phi_2$ and $\phi_3$ is also good for SMILE, and matches with ORACLE as sample size $n$ increases. The inferior performance of SAPLM and the poor performance of SLM, in both estimation and prediction, illustrates the importance and necessity of identifying correct model structure.

\begin{table}[htbp]
\small
\caption{Estimation results comparing the ORACLE, SMILE, SAPLM and SLM.}
\label{TAB:SIM1-estimation}
\begin{center}
\renewcommand{\arraystretch}{0.725}
\begin{tabular}{ccclrrrrrrrr}
	\hline \hline
	& & & & \multicolumn{4}{c}{MSE $(\times 10^{-2})$} & \multicolumn{3}{c}{AMSE $(\times 10^{-2})$} & CV- \\
	\cmidrule(r){5-8} \cmidrule(l){9-11}
	$n$ & $\sigma$ & $p$ & Method & $\alpha_1$ & $\alpha_2$ & $\alpha_3$ & $\beta_1$ & $\phi_1$ & $\phi_2$ & $\phi_3$ & MSPE \\
	\hline
          300 &\!\!\! 0.5 &\!\!\! 1000  &\!\!\! ORACLE &\!\!\! 0.47  &\!\!\! 0.48  &\!\!\! 0.50  &\!\!\! 1.06  &\!\!\! 0.09  &\!\!\! 0.98  &\!\!\! 0.83  &\!\!\! 0.28 \\
          &\!\!\! &\!\!\! &\!\!\! SMILE &\!\!\! 0.47  &\!\!\! 0.48  &\!\!\! 0.50  &\!\!\! 1.06  &\!\!\! 0.11  &\!\!\! 0.94  &\!\!\! 0.77  &\!\!\! 0.28 \\
          &\!\!\! &\!\!\! &\!\!\! SAPLM &\!\!\! 0.49  &\!\!\! 0.48  &\!\!\! 0.54  &\!\!\! -   &\!\!\! 0.41  &\!\!\! 0.99  &\!\!\! 0.83  &\!\!\! 0.28 \\
          &\!\!\! &\!\!\! &\!\!\! SLM &\!\!\! 9.13  &\!\!\! 8.86  &\!\!\! 25.99 &\!\!\! 18.65 &\!\!\! 1.63  &\!\!\! 253.82 &\!\!\! 178.85 &\!\!\! 4.79 \\
          &\!\!\! &\!\!\! 2000  &\!\!\! ORACLE &\!\!\! 0.47  &\!\!\! 0.46  &\!\!\! 0.47  &\!\!\! 1.08  &\!\!\! 0.09  &\!\!\! 0.99  &\!\!\! 0.85  &\!\!\! 0.27 \\
          &\!\!\! &\!\!\! &\!\!\! SMILE &\!\!\! 0.47  &\!\!\! 0.45  &\!\!\! 0.47  &\!\!\! 1.08  &\!\!\! 0.19  &\!\!\! 0.95  &\!\!\! 0.79  &\!\!\! 0.27 \\
          &\!\!\! &\!\!\! &\!\!\! SAPLM &\!\!\! 0.49  &\!\!\! 0.47  &\!\!\! 0.53  &\!\!\! -   &\!\!\! 0.43  &\!\!\! 1.01  &\!\!\! 0.86  &\!\!\! 0.28 \\
          &\!\!\! &\!\!\! &\!\!\! SLM &\!\!\! 10.51 &\!\!\! 8.70  &\!\!\! 41.05 &\!\!\! 21.29 &\!\!\! 1.84  &\!\!\! 252.75 &\!\!\! 180.40 &\!\!\! 4.80 \\
          &\!\!\! &\!\!\! 5000  &\!\!\! ORACLE &\!\!\! 0.45  &\!\!\! 0.44  &\!\!\! 0.53  &\!\!\! 1.06  &\!\!\! 0.09  &\!\!\! 0.97  &\!\!\! 0.81  &\!\!\! 0.27 \\
          &\!\!\! &\!\!\! &\!\!\! SMILE &\!\!\! 0.45  &\!\!\! 0.44  &\!\!\! 0.53  &\!\!\! 1.06  &\!\!\! 0.16  &\!\!\! 0.94  &\!\!\! 0.75  &\!\!\! 0.27 \\
          &\!\!\! &\!\!\! &\!\!\! SAPLM &\!\!\! 0.47  &\!\!\! 0.47  &\!\!\! 0.57  &\!\!\! -   &\!\!\! 0.42  &\!\!\! 0.98  &\!\!\! 0.81  &\!\!\! 0.28 \\
          &\!\!\! &\!\!\! &\!\!\! SLM &\!\!\! 9.29  &\!\!\! 8.66  &\!\!\! 49.64 &\!\!\! 19.90 &\!\!\! 1.73  &\!\!\! 252.44 &\!\!\! 179.41 &\!\!\! 4.83 \\
          &\!\!\! 1.0 &\!\!\! 1000  &\!\!\! ORACLE &\!\!\! 1.94  &\!\!\! 1.98  &\!\!\! 1.82  &\!\!\! 4.48  &\!\!\! 0.37  &\!\!\! 2.97  &\!\!\! 2.63  &\!\!\! 1.08 \\
          &\!\!\! &\!\!\! &\!\!\! SMILE &\!\!\! 1.94  &\!\!\! 1.98  &\!\!\! 1.82  &\!\!\! 4.48  &\!\!\! 0.56  &\!\!\! 2.80  &\!\!\! 2.29  &\!\!\! 1.09 \\
          &\!\!\! &\!\!\! &\!\!\! SAPLM &\!\!\! 1.98  &\!\!\! 2.01  &\!\!\! 1.96  &\!\!\! -   &\!\!\! 1.44  &\!\!\! 2.98  &\!\!\! 2.53  &\!\!\! 1.09 \\
          &\!\!\! &\!\!\! &\!\!\! SLM &\!\!\! 11.33 &\!\!\! 10.55 &\!\!\! 51.08 &\!\!\! 22.51 &\!\!\! 1.95  &\!\!\! 253.34 &\!\!\! 180.31 &\!\!\! 5.59 \\
          &\!\!\! &\!\!\! 2000  &\!\!\! ORACLE &\!\!\! 1.90  &\!\!\! 1.77  &\!\!\! 1.82  &\!\!\! 4.16  &\!\!\! 0.35  &\!\!\! 3.04  &\!\!\! 2.57  &\!\!\! 1.08 \\
          &\!\!\! &\!\!\! &\!\!\! SMILE &\!\!\! 1.91  &\!\!\! 1.84  &\!\!\! 2.62  &\!\!\! 4.23  &\!\!\! 0.73  &\!\!\! 3.30  &\!\!\! 3.20  &\!\!\! 1.11 \\
          &\!\!\! &\!\!\! &\!\!\! SAPLM &\!\!\! 1.98  &\!\!\! 1.85  &\!\!\! 1.98  &\!\!\! -   &\!\!\! 1.40  &\!\!\! 3.07  &\!\!\! 2.49  &\!\!\! 1.09 \\
          &\!\!\! &\!\!\! &\!\!\! SLM &\!\!\! 11.04 &\!\!\! 10.19 &\!\!\! 60.67 &\!\!\! 23.02 &\!\!\! 1.99  &\!\!\! 252.71 &\!\!\! 179.98 &\!\!\! 5.61 \\
          &\!\!\! &\!\!\! 5000  &\!\!\! ORACLE &\!\!\! 1.71  &\!\!\! 1.89  &\!\!\! 1.93  &\!\!\! 4.03  &\!\!\! 0.33  &\!\!\! 2.93  &\!\!\! 2.54  &\!\!\! 1.08 \\
          &\!\!\! &\!\!\! &\!\!\! SMILE &\!\!\! 1.82  &\!\!\! 1.92  &\!\!\! 3.52  &\!\!\! 4.05  &\!\!\! 0.39  &\!\!\! 3.97  &\!\!\! 4.67  &\!\!\! 1.28 \\
          &\!\!\! &\!\!\! &\!\!\! SAPLM &\!\!\! 1.77  &\!\!\! 1.97  &\!\!\! 2.09  &\!\!\! -   &\!\!\! 1.43  &\!\!\! 2.96  &\!\!\! 2.44  &\!\!\! 1.25 \\
          &\!\!\! &\!\!\! &\!\!\! SLM &\!\!\! 16.78 &\!\!\! 10.55 &\!\!\! 80.30 &\!\!\! 23.28 &\!\!\! 2.01  &\!\!\! 252.41 &\!\!\! 180.60 &\!\!\! 5.72 \\
          500 &\!\!\! 0.5 &\!\!\! 1000  &\!\!\! ORACLE &\!\!\! 0.27  &\!\!\! 0.28  &\!\!\! 0.28  &\!\!\! 0.67  &\!\!\! 0.06  &\!\!\! 0.67  &\!\!\! 0.58  &\!\!\! 0.27 \\
          &\!\!\! &\!\!\! &\!\!\! SMILE &\!\!\! 0.27  &\!\!\! 0.28  &\!\!\! 0.28  &\!\!\! 0.67  &\!\!\! 0.07  &\!\!\! 0.65  &\!\!\! 0.55  &\!\!\! 0.27 \\
          &\!\!\! &\!\!\! &\!\!\! SAPLM &\!\!\! 0.29  &\!\!\! 0.29  &\!\!\! 0.31  &\!\!\! -   &\!\!\! 0.27  &\!\!\! 0.67  &\!\!\! 0.58  &\!\!\! 0.27 \\
          &\!\!\! &\!\!\! &\!\!\! SLM &\!\!\! 5.08  &\!\!\! 4.91  &\!\!\! 5.88  &\!\!\! 10.70 &\!\!\! 0.97  &\!\!\! 253.20 &\!\!\! 180.13 &\!\!\! 4.66 \\
          &\!\!\! &\!\!\! 2000  &\!\!\! ORACLE &\!\!\! 0.27  &\!\!\! 0.27  &\!\!\! 0.31  &\!\!\! 0.65  &\!\!\! 0.05  &\!\!\! 0.65  &\!\!\! 0.55  &\!\!\! 0.26 \\
          &\!\!\! &\!\!\! &\!\!\! SMILE &\!\!\! 0.27  &\!\!\! 0.27  &\!\!\! 0.31  &\!\!\! 0.65  &\!\!\! 0.06  &\!\!\! 0.63  &\!\!\! 0.52  &\!\!\! 0.26 \\
          &\!\!\! &\!\!\! &\!\!\! SAPLM &\!\!\! 0.28  &\!\!\! 0.28  &\!\!\! 0.34  &\!\!\! -   &\!\!\! 0.27  &\!\!\! 0.66  &\!\!\! 0.55  &\!\!\! 0.27 \\
          &\!\!\! &\!\!\! &\!\!\! SLM &\!\!\! 5.25  &\!\!\! 4.99  &\!\!\! 5.90  &\!\!\! 11.93 &\!\!\! 1.07  &\!\!\! 252.96 &\!\!\! 179.22 &\!\!\! 4.66 \\
          &\!\!\! &\!\!\! 5000  &\!\!\! ORACLE &\!\!\! 0.29  &\!\!\! 0.25  &\!\!\! 0.29  &\!\!\! 0.62  &\!\!\! 0.05  &\!\!\! 0.67  &\!\!\! 0.57  &\!\!\! 0.26 \\
          &\!\!\! &\!\!\! &\!\!\! SMILE &\!\!\! 0.29  &\!\!\! 0.25  &\!\!\! 0.29  &\!\!\! 0.62  &\!\!\! 0.17  &\!\!\! 0.64  &\!\!\! 0.54  &\!\!\! 0.26 \\
          &\!\!\! &\!\!\! &\!\!\! SAPLM &\!\!\! 0.30  &\!\!\! 0.26  &\!\!\! 0.32  &\!\!\! -   &\!\!\! 0.28  &\!\!\! 0.67  &\!\!\! 0.57  &\!\!\! 0.27 \\
          &\!\!\! &\!\!\! &\!\!\! SLM &\!\!\! 5.30  &\!\!\! 4.87  &\!\!\! 6.35  &\!\!\! 11.96 &\!\!\! 1.07  &\!\!\! 252.99 &\!\!\! 179.99 &\!\!\! 4.66 \\
          &\!\!\! 1.0 &\!\!\! 1000  &\!\!\! ORACLE &\!\!\! 1.18  &\!\!\! 1.08  &\!\!\! 1.09  &\!\!\! 2.43  &\!\!\! 0.20  &\!\!\! 1.90  &\!\!\! 1.62  &\!\!\! 1.05 \\
          &\!\!\! &\!\!\! &\!\!\! SMILE &\!\!\! 1.18  &\!\!\! 1.08  &\!\!\! 1.09  &\!\!\! 2.43  &\!\!\! 0.56  &\!\!\! 1.83  &\!\!\! 1.47  &\!\!\! 1.05 \\
          &\!\!\! &\!\!\! &\!\!\! SAPLM &\!\!\! 1.21  &\!\!\! 1.12  &\!\!\! 1.15  &\!\!\! -   &\!\!\! 0.87  &\!\!\! 1.92  &\!\!\! 1.60  &\!\!\! 1.06 \\
          &\!\!\! &\!\!\! &\!\!\! SLM &\!\!\! 6.45  &\!\!\! 5.26  &\!\!\! 7.42  &\!\!\! 12.11 &\!\!\! 1.09  &\!\!\! 253.05 &\!\!\! 180.33 &\!\!\! 5.41 \\
          &\!\!\! &\!\!\! 2000  &\!\!\! ORACLE &\!\!\! 1.12  &\!\!\! 1.02  &\!\!\! 1.12  &\!\!\! 2.45  &\!\!\! 0.20  &\!\!\! 1.94  &\!\!\! 1.66  &\!\!\! 1.04 \\
          &\!\!\! &\!\!\! &\!\!\! SMILE &\!\!\! 1.12  &\!\!\! 1.02  &\!\!\! 1.12  &\!\!\! 2.45  &\!\!\! 0.22  &\!\!\! 1.84  &\!\!\! 1.49  &\!\!\! 1.04 \\
          &\!\!\! &\!\!\! &\!\!\! SAPLM &\!\!\! 1.15  &\!\!\! 1.05  &\!\!\! 1.21  &\!\!\! -   &\!\!\! 0.85  &\!\!\! 1.94  &\!\!\! 1.63  &\!\!\! 1.05 \\
          &\!\!\! &\!\!\! &\!\!\! SLM &\!\!\! 6.12  &\!\!\! 5.99  &\!\!\! 7.62  &\!\!\! 13.76 &\!\!\! 1.22  &\!\!\! 252.81 &\!\!\! 180.10 &\!\!\! 5.43 \\
          &\!\!\! &\!\!\! 5000  &\!\!\! ORACLE &\!\!\! 1.12  &\!\!\! 1.05  &\!\!\! 1.16  &\!\!\! 2.46  &\!\!\! 0.20  &\!\!\! 1.96  &\!\!\! 1.67  &\!\!\! 1.05 \\
          &\!\!\! &\!\!\! &\!\!\! SMILE &\!\!\! 1.12  &\!\!\! 1.05  &\!\!\! 1.16  &\!\!\! 2.46  &\!\!\! 0.22  &\!\!\! 1.87  &\!\!\! 1.48  &\!\!\! 1.05 \\
          &\!\!\! &\!\!\! &\!\!\! SAPLM &\!\!\! 1.14  &\!\!\! 1.08  &\!\!\! 1.22  &\!\!\! -   &\!\!\! 0.87  &\!\!\! 1.97  &\!\!\! 1.64  &\!\!\! 1.06 \\
          &\!\!\! &\!\!\! &\!\!\! SLM &\!\!\! 6.16  &\!\!\! 5.64  &\!\!\! 9.37  &\!\!\! 12.28 &\!\!\! 1.10  &\!\!\! 252.69 &\!\!\! 180.26 &\!\!\! 5.43 \\
	\hline \hline
\end{tabular}
\end{center}
\end{table}

Next we investigate the coverage rates of the proposed SCB. For each replication, we test whether the true functions are covered by the SCB at the simulated values of the covariate in the interval $[-0.5+h,0.5-h]$, where $h$ is the bandwidth. Table \ref{TAB:SIM1-scb} shows the empirical coverage probabilities for a nominal 95\% confidence level out of 500 replications. For comparison, we also provide the SCBs from the SAPLM and ORACLE estimators. From Table \ref{TAB:SIM1-scb}, we observe that coverage probabilities for the SMILE, SAPLM and ORACLE SCBs all approach the nominal levels as $n$ increases, which provides positive confirmation of Theorem \ref{THM:nonlinear-normality}. In most cases, SMILE performs as well as or better than SAPLM, and arrives at about the nominal coverage when $n = 500$ and $\sigma=1.0$. Figure \ref{FIG:SIM-band} depicts the true function $\phi_{\ell}$, the corresponding SMILE $\widehat{\phi}_{\ell}^{\mathrm{SBLL}}$ and the $95\%$ SCB for $\phi_{\ell}$ based on $\widehat{\phi}_{\ell}^{\mathrm{SBLL}}$, for $\ell = 2, 3$, which are based on a typical run with $n = 500$, $p = 1000$ and $\sigma = 1.0$.

\begin{table}[htbp]
\caption{Coverage rates comparing the ORACLE, SMILE and SAPLM.} \vspace*{-.3cm}
\label{TAB:SIM1-scb}
\begin{center}
\renewcommand{\arraystretch}{0.75}
\begin{tabular}{ccccccccc}
	\hline \hline
	Size & Noise & & \multicolumn{3}{c}{$\phi_2$ Coverage (\%)} & \multicolumn{3}{c}{$\phi_3$ Coverage (\%)} \\
	\cmidrule(r){4-6} \cmidrule(l){7-9}
	$n$ & $\sigma$ & $p$ & ORACLE & SMILE & SAPLM & ORACLE & SMILE & SAPLM \\
	\hline
	300   & 0.5   & 1000  & 93.7  & 94.5  & 93.9  & 92.4  & 92.6  & 91.7 \\
          &       & 2000  & 92.6  & 93.3  & 92.6  & 92.3  & 93.8  & 92.5 \\
          &       & 5000  & 92.3  & 93.0    & 92.7  & 93.3  & 92.3  & 91.7 \\
          & 1     & 1000  & 96.0    & 95.6  & 94.7  & 96.1  & 96.4  & 95.3 \\
          &       & 2000  & 95.4  & 95.7  & 94.9  & 96.1  & 96.2  & 95.5 \\
          &       & 5000  & 95.1  & 95.6  & 94.2  & 95.9  & 96.4  & 94.8 \\
	500   & 0.5   & 1000  & 92.9  & 93.8  & 93.5  & 92.7  & 90.6  & 92.0 \\
          &       & 2000  & 92.5  & 92.7  & 92.3  & 92.0    & 92.0    & 92.3 \\
          &       & 5000  & 92.5  & 92.6  & 91.8  & 91.5  & 89.9  & 90.4 \\
          & 1     & 1000  & 97.1  & 96.7  & 96.3  & 96.0    & 96.0    & 95.2 \\
          &       & 2000  & 95.2  & 95.0    & 94.5  & 95.2  & 94.6  & 94.3 \\
          &       & 5000  & 94.7  & 95.1  & 95.0    & 96.2  & 96.0    & 95.5 \\
	\hline \hline
\end{tabular}
\end{center}
\end{table}

\begin{figure}[htbp]
\begin{center}
	\subfigure 
	{\includegraphics[scale=0.62]{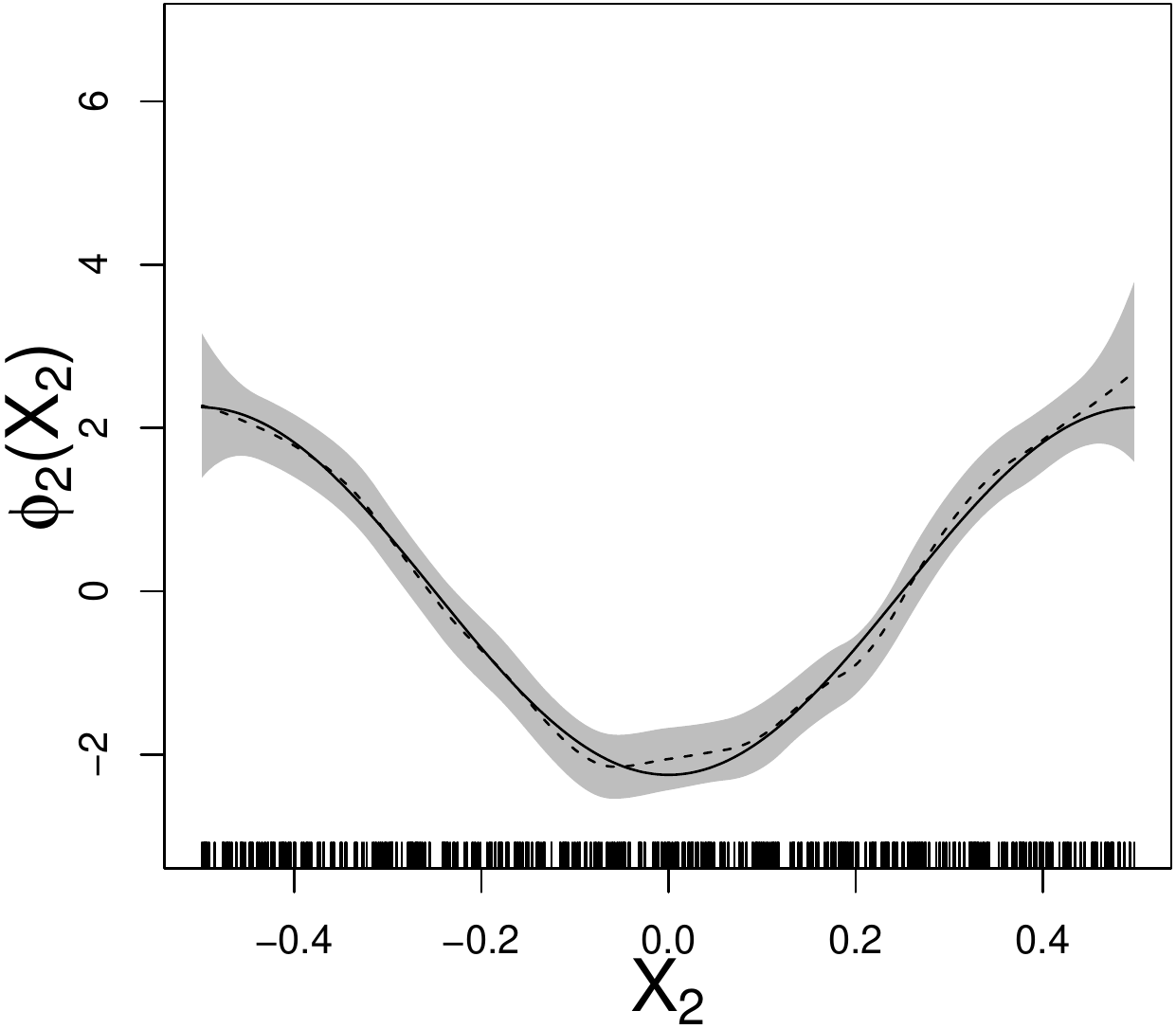}\label{FIG:cb_phi2}}
	\subfigure 
	{\includegraphics[scale=0.62]{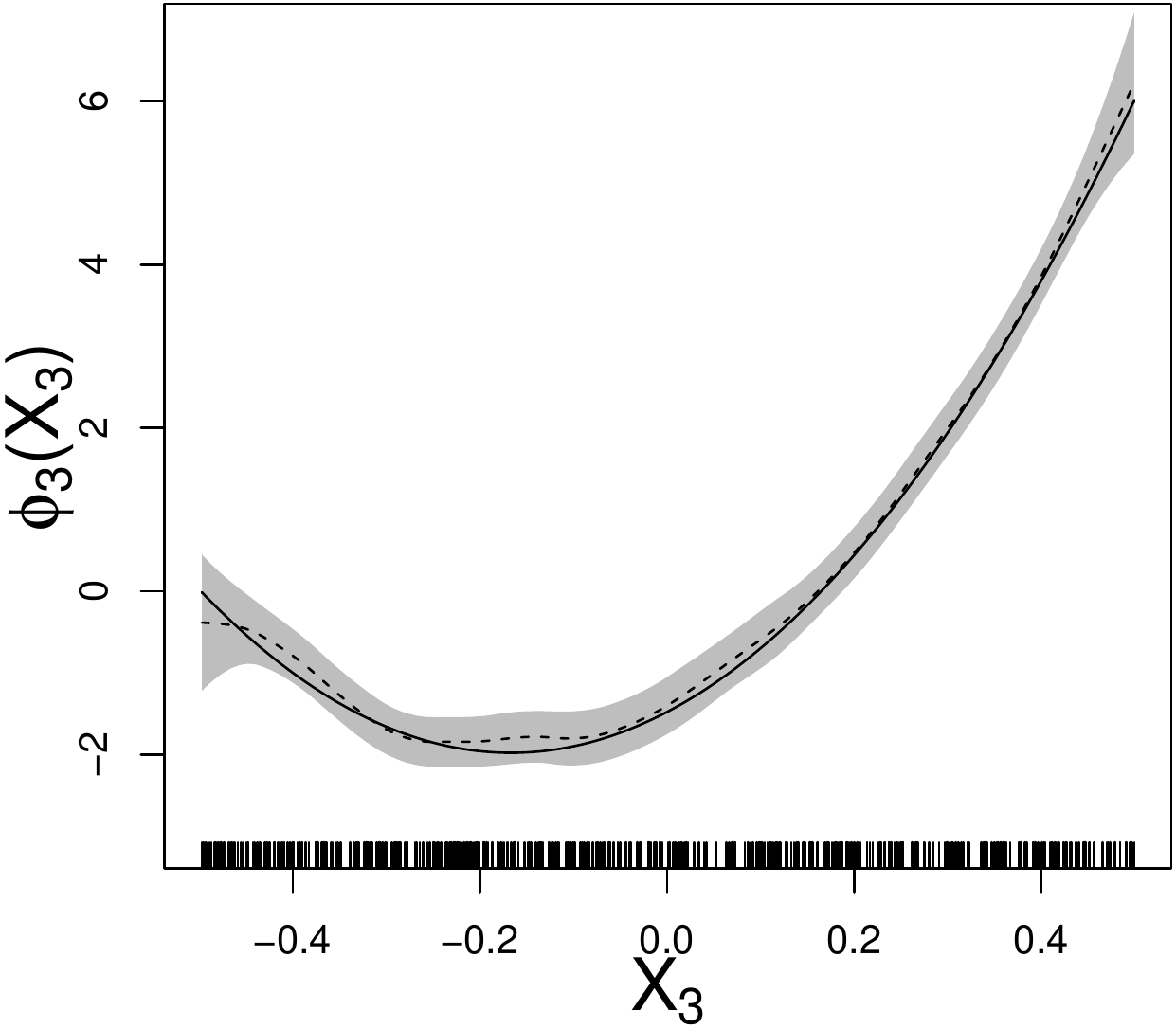}\label{FIG:cb_phi3}}
\end{center} \vspace*{-.6cm}
\caption{Plots of the SMILE (dashed curve) and the 95\% SCB (shaded area) of the nonparametric component $\phi_{\ell}(x_{\ell})$, $\ell = 2, 3$ (solid curve).}
\label{FIG:SIM-band}
\end{figure}

Appendices B--D contain the results of additional simulations which show that our proposed SMILE procedure performs well relative to competing methods under a wider range of conditions.

\setcounter{chapter}{5}
\renewcommand{\thetable}{{\arabic{table}}}
\renewcommand{\thefigure}{{\arabic{figure}}}
\renewcommand{\thesection}{\arabic{section}}
\renewcommand{\thesubsection}{5.\arabic{subsection}}
\setcounter{section}{4}

\section{Application} \label{sec:application}

We illustrate the application of our proposed method in the ultra-high-dimensional setting by using the SAM data generated by \citet[][]{leiboff2015genetic}.  The maize SAM is a small pool of stem cells located in the plant shoot that generate all the above-ground tissues of maize plants. \citet[][]{leiboff2015genetic} showed that SAM volume is correlated with a variety of agronomically important traits in adult plants.  The goal of our analysis is to model and predict SAM volume as a function of single nucleotide polymorphism (SNP) genotypes and messenger RNA transcript abundance levels using data from maize inbred lines.  Following the preprocessing steps described in Section B.5 in the Supplementary Materials in \cite{li2017ultra}, linear sure independent screening \citep{fan2008sure} for SNP genotypes, and nonlinear independent screening \citep{fan2011nonparametric} for RNA transcripts, the dataset we analyze c
onsists of log-scale SAM volume measurements, binary SNP genotypes at $p_1=5203$ markers, and log-scale measures of abundance for $p_2=1020$ transcripts for each of $n=368$ maize inbred lines.

\citet{li2017ultra} used the APLM  to model the relationship between the log SAM volume response and predictors determined by SNP genotypes and RNA transcript abundance levels. Because the SNP genotypes are binary, they naturally entered the linear part of the APLM, and for convenience all the RNA transcripts were included in the nonlinear part of the APLM in \cite{li2017ultra}. As discussed before, failing to account for exactly linear features makes the APLM less efficient statistically and computationally. In the following we apply our proposed SMILE method to distinguish among RNA transcripts entering the nonparametric and parametric parts of the APLM and to identify significant SNP genotypes and RNA transcripts simultaneously.

To compare the results of SMILE to the sparse APLM and the sparse linear regression model, we also analyze the data using the SAPLM and SLM estimators presented in \cite{li2017ultra}. Parallel to the settings in Section \ref{sec:simulation}, we use constant B-splines with four quantile knots for model structure identification, and use cubic B-splines with one quantile knot for nonlinear function approximation. We use the iterative algorithm proposed in Section \ref{subsec:algorithm} for penalty parameter selection and estimation.

As shown in Table \ref{TAB: RNA_sel}, SMILE identified 169 SNPs, 10 RNA transcripts linearly associated with log SAM size and 2 RNA transcripts that have nonlinear association with log SAM size. In contrast, SAPLM selected 177 SNPs and 3 RNA transcripts, and SLM selected 167 SNPs and 32 RNA transcripts.  To evaluate the predictive performance of the two methods, we computed 10-fold cross-validation mean squared prediction error (CV-MSPE) for each method. The SMILE-estimated nonlinear function for the selected nonlinear RNA transcript  is plotted, along with 95\% SCBs, in Figure \ref{FIG:SAM-band}.

\begin{table}[htbp]
\caption{Selected SNPs and Transcripts by SMILE, SAPLM and SLM.}
\label{TAB: RNA_sel}
\begin{center}
\begin{tabular}{lccc}
	\hline\hline
	RNA Transcripts Selected& SMILE & SAPLM & SLM \\
	\hline
	$X_{725}$ & \checkmark & \checkmark & \checkmark \\
	$X_{127}$, $X_{136}$, $X_{141}$, $X_{208}$, $X_{289}$, $X_{312}$, $X_{493}$,$X_{749}$,$X_{855}$& \checkmark & & \checkmark \\
	$X_{153}^{\ast}$, $X_{677}^{\ast}$ & \checkmark & & \\
	$X_{157}$,$X_{701}$ & & \checkmark & \\
	$X_{209}$,$X_{314}$, $X_{320}$, $X_{321}$, $X_{342}$, $X_{419}$,$X_{472}$,$X_{489}$,$X_{553}$, & & & \checkmark \\
	$X_{589}$,$X_{601}$,$X_{615}$, $X_{783}$,$X_{785}$,$X_{793}$,$X_{846}$,$X_{863}$, $X_{940}$,& & & \checkmark \\
	$X_{946}$,$X_{978}$,$X_{1002}$,$X_{1018}$& & & \checkmark \\
	\hline
	Number of SNP Genotypes & 169 & 177 & 167 \\
	Number of Linear RNA Transcripts & 10 & 0 & 32 \\
	Number of Functional RNA Transcripts & 2 & 3 & 0 \\
	\hline
	CV MSPE& 0.060 & 0.102 & 0.132 \\
	CV Mean Number of SNPs & 153.9 & 175.9 & 83.1 \\
	CV Mean Number of Linear Transcripts & 8.7 & 0 & 17.7 \\
	CV Mean Number of Nonlinear Transcripts & 1.9 & 3.8 & 0 \\
	\hline\hline
	\multicolumn{4}{l}{$\ast$ nonlinear association identified by SMILE for $X_{153}$ and $X_{677}$}
\end{tabular}
\end{center}
\end{table}

\begin{figure}[htbp]
\begin{center}
	\subfigure
	{\includegraphics[scale=0.58]{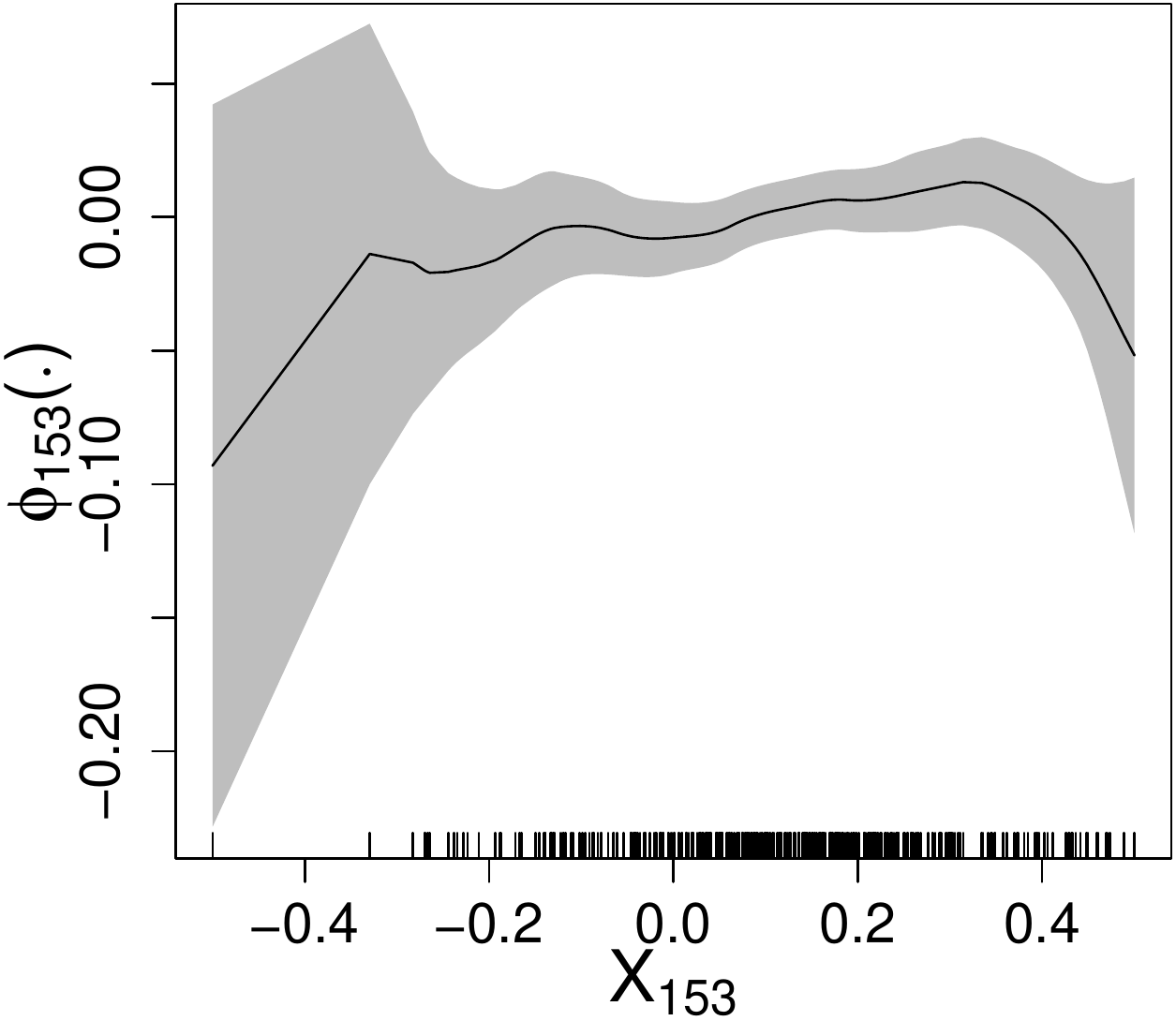}\label{X153}
	\includegraphics[scale=0.58]{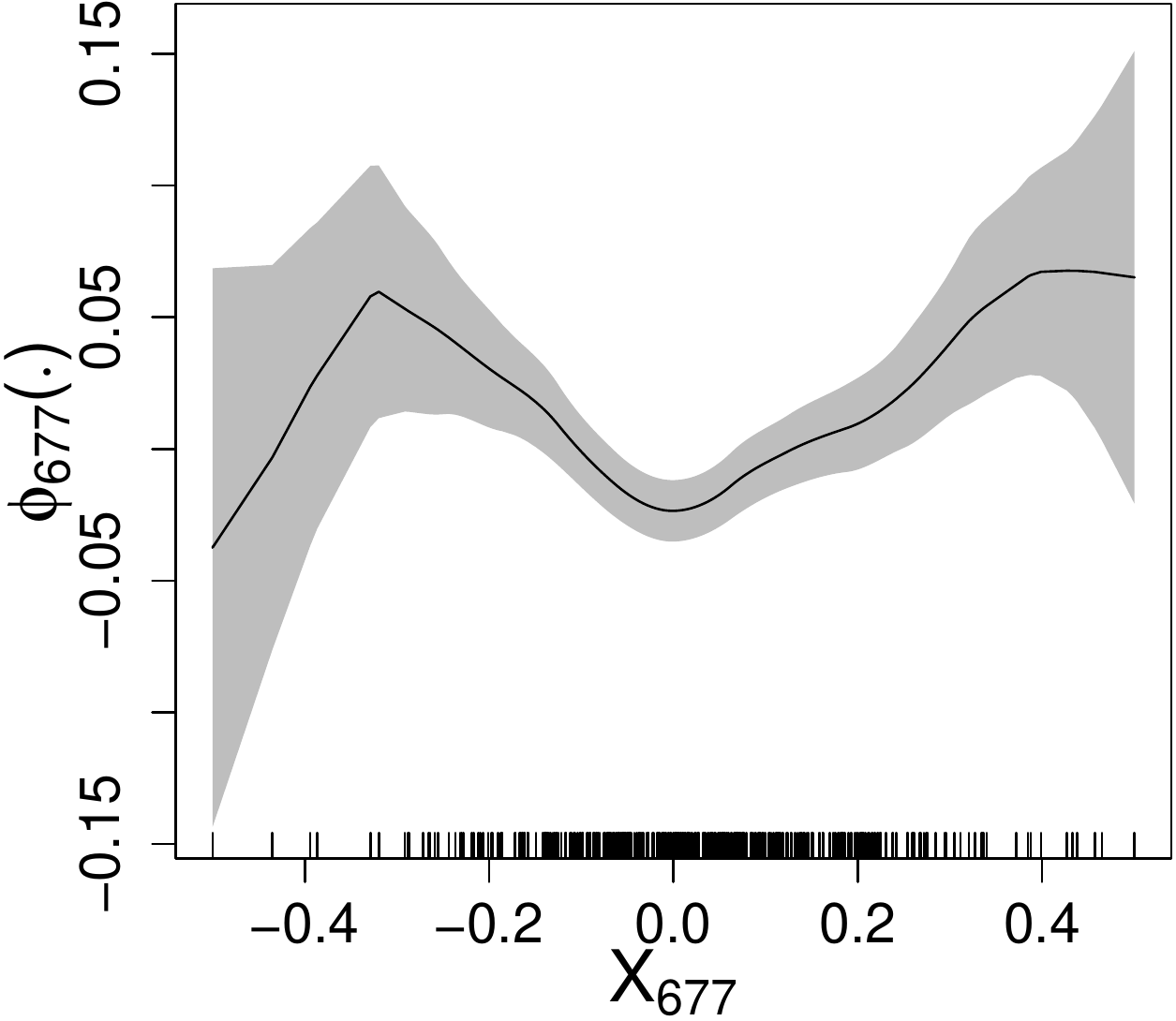}\label{X677}
	}
\end{center} \vspace*{-.8cm}
\caption{Plot of the SMILE (solid curve) and the 95\% confidence band (shaded area) for the selected RNA transcript.}
\label{FIG:SAM-band}
\end{figure}

\setcounter{chapter}{6}
\renewcommand{\thesection}{\arabic{section}}
\setcounter{section}{5}

\section{Discussion} \label{sec:conclusion}

This paper focuses on the simultaneous sparse model identification and learning for ultra-high-dimensional APLMs which strikes a delicate balance between the simplicity of the standard linear regression models and the flexibility of the additive regression models. We proposed a two-stage penalization method, called SMILE, which can efficiently select nonzero components and identify the linear-and-nonlinear structure in the functional terms, as well as simultaneously estimate and make inference for both linear coefficients and nonlinear functions. First, we have devised a groupwise penalization method in the APLM for simultaneous variable selection and structure identification. After identifying important covariates and the functional forms for the selected covariates, we have further constructed SCBs for the nonzero nonparametric functions based on refined spline-backfitted local-linear estimators. Our simulation studies and applications demonstrate the proposed SMILE procedure can be more efficient than penalized linear regression and the penalized APLM without model identification, and can improve predictions.

Our work differs from previous works in practical, theoretical and computational aspects: (i) We perform variable selection and model structure identification simultaneously, for both the linear components in $\bs{Z}$, and the linear and nonlinear forms for the components of $\bs{X}$. In contrast, existing works either performs only model structure identification or performs variable selection only for components in $\bs{X}$.
(ii) Besides the consistency of model structure identification, we also provide inference tools for both the regression coefficients and the component functions.
(iii) Compared to the local quadratic approximation approach used in \cite{lian2015separation}, which cannot  provide exactly zero solutions and is inefficient for fitting large regression problems, our proposed iterative group coordinate descent algorithm takes advantage of sparsity in computation and is able to deal with the triple penalization problem very efficiently. (See \cite{breheny2015group} for a detailed comparison of these two algorithms.) Our algorithm is easy to implement and can provide analysis results for large data sets with thousands of dimensions within seconds.

Our work deals with independent observations but can be extended to longitudinal data settings through marginal models or mixed-effects models. In addition, although we consider continuous response variables in our work, or approach can be readily extended to generalized additive partially linear models, to deal with different types of responses. Currently, the APLM assumes that the effects of all covariates are additive, which may overlook the potential interaction between covariates. Our method can be extended to models that can accommodate interactions between covariates, for example, APLMs with interaction terms. We leave such extensions to future work. Another limitation of our work is a reliance on the assumption of constant error variance. However, heteroscedasticity may be encountered in the analysis of genomic data sets. It is of interest to develop a new methodology that allows non-constant error variance for high-dimensional estimation and model selection, and this is another challenge we leave for future work.

\section*{Acknowledgment}

This work was supported by the Iowa State University Plant Sciences Institute Scholars Program. In addition, Wang's research was supported by NSF grant DMS-1542332, and Nettleton's research was supported by NSF grant IOS-1238142. We sincerely thank the Editor, the Associate Editor and the anonymous reviewers for their insightful comments that have lead to significant improvements on the paper.

\setcounter{chapter}{8}
\renewcommand{\theequation}{A.\arabic{equation}}
\renewcommand{\thesection}{A.\arabic{section}}
\renewcommand{\thesubsection}{A.\arabic{subsection}}
\renewcommand{\thetheorem}{A.\arabic{theorem}}
\renewcommand{\thelemma}{A.\arabic{lemma}}
\renewcommand{\theproposition}{A.\arabic{proposition}}
\renewcommand{\thecorollary}{A.\arabic{corollary}}
\renewcommand{\thefigure}{A.\arabic{figure}}
\renewcommand{\thetable}{A.\arabic{table}}
\setcounter{equation}{0}
\setcounter{theorem}{0}
\setcounter{lemma}{0}
\setcounter{figure}{0}
\setcounter{table}{0}
\setcounter{proposition}{0}
\setcounter{section}{0}
\setcounter{subsection}{0}
\markboth{}{}

\vskip .4in \noindent \textbf{\Large Appendices}

\section*{A. Effect of Smoothing Parameters on Performance of SMILE}

To implement the proposed SMILE procedure, one needs to select the knots for a spline at the selection stage and refitting stage, and the bandwidth for a kernel at the backfitting stage. In this section, we study how these smoothing parameters affect the proposed SMILE method and evaluate the practical performance in the finite-sample simulation studies described in Section 4.2 of the main paper. In the literature of polynomial spline smoothing, the knots for a spline are generally put on a grid of equally spaced sample quantiles \citep{ruppert2002selecting}. Therefore, we only need to investigate the effect of the number of knots on the performance of SMILE.

At the first stage (model selection), we use piecewise constant splines with the number of interior knots $N=2,3,\ldots,8$ in the simulation. Figure \ref{FIG:knots} shows the effect of $N$ on the accuracy of model selection based on the criteria defined in the main paper: (B-i)--(B-vi) and (C-i)--(C-iv). From Figure \ref{FIG:knots}, it appears that the value $N$ has little effect on the selection results. For all combinations of $n$, $p$ and $\sigma$, no matter which $N$ is used, the ``corrZ0", ``corrL", ``corrX0" are all $100\%$, and the ``LtoN" and ``Nto0" are all $0\%$. The values of ``corrZ", ``corrN", ``corrLN" and ``Zto0" and ``Xto0" are not exactly the same when using different values of $N$, but they are almost constant for $N=2,3,\ldots,8$. Especially when the sample size $n=500$, the proposed SMILE is able to identify the true model structure regardless of $p=1000,2000$ or $5000$.  When $n=300$ and $p=5000$, the selection results become slightly worse when we increase to $N\geq 6$.

\begin{figure}[htbp]
\begin{center}
\begin{tabular}{ccc}
	corrZ & corrZ0 & corrL  \\
	\includegraphics[scale=0.39]{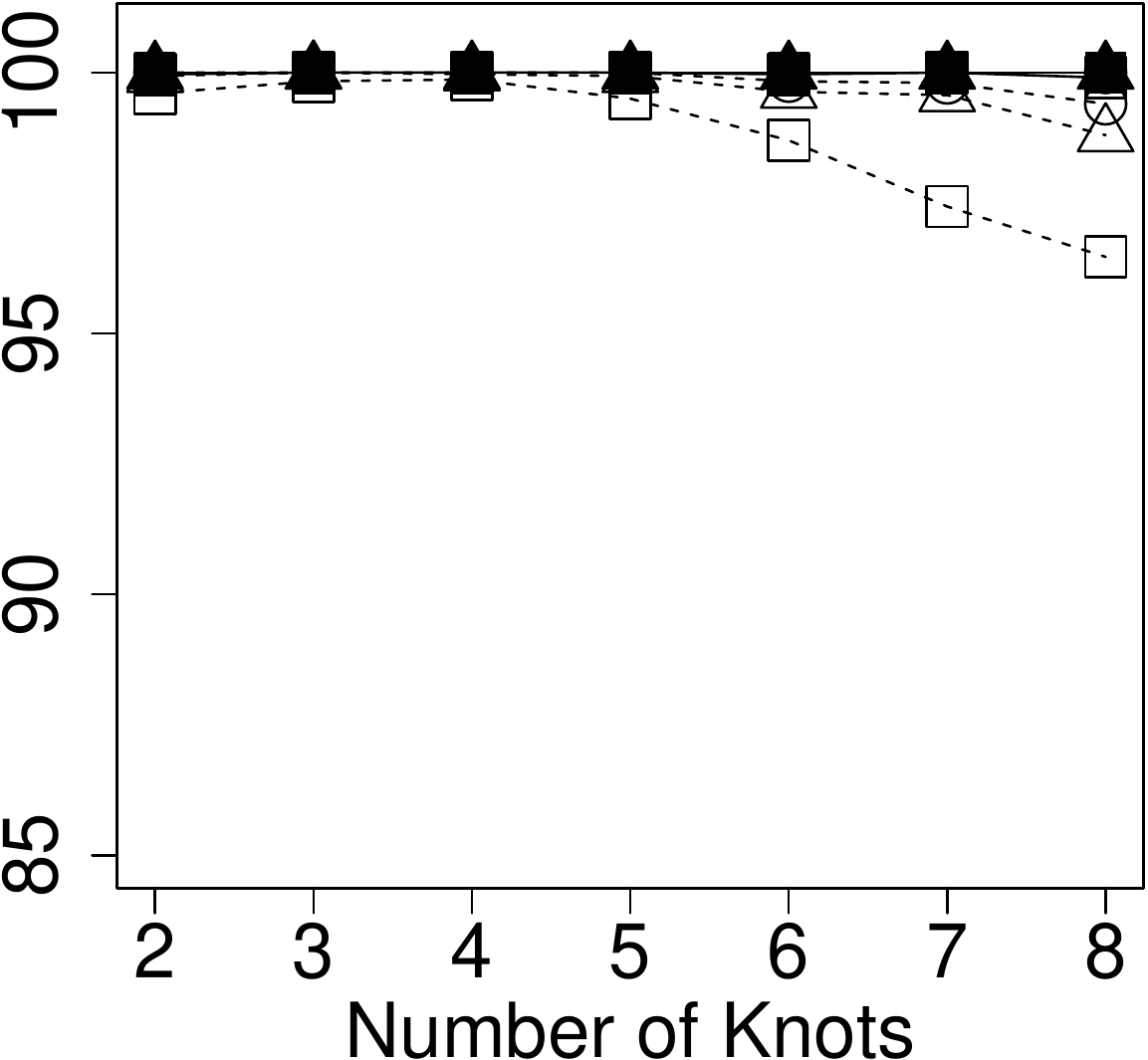}  &
	\includegraphics[scale=0.39]{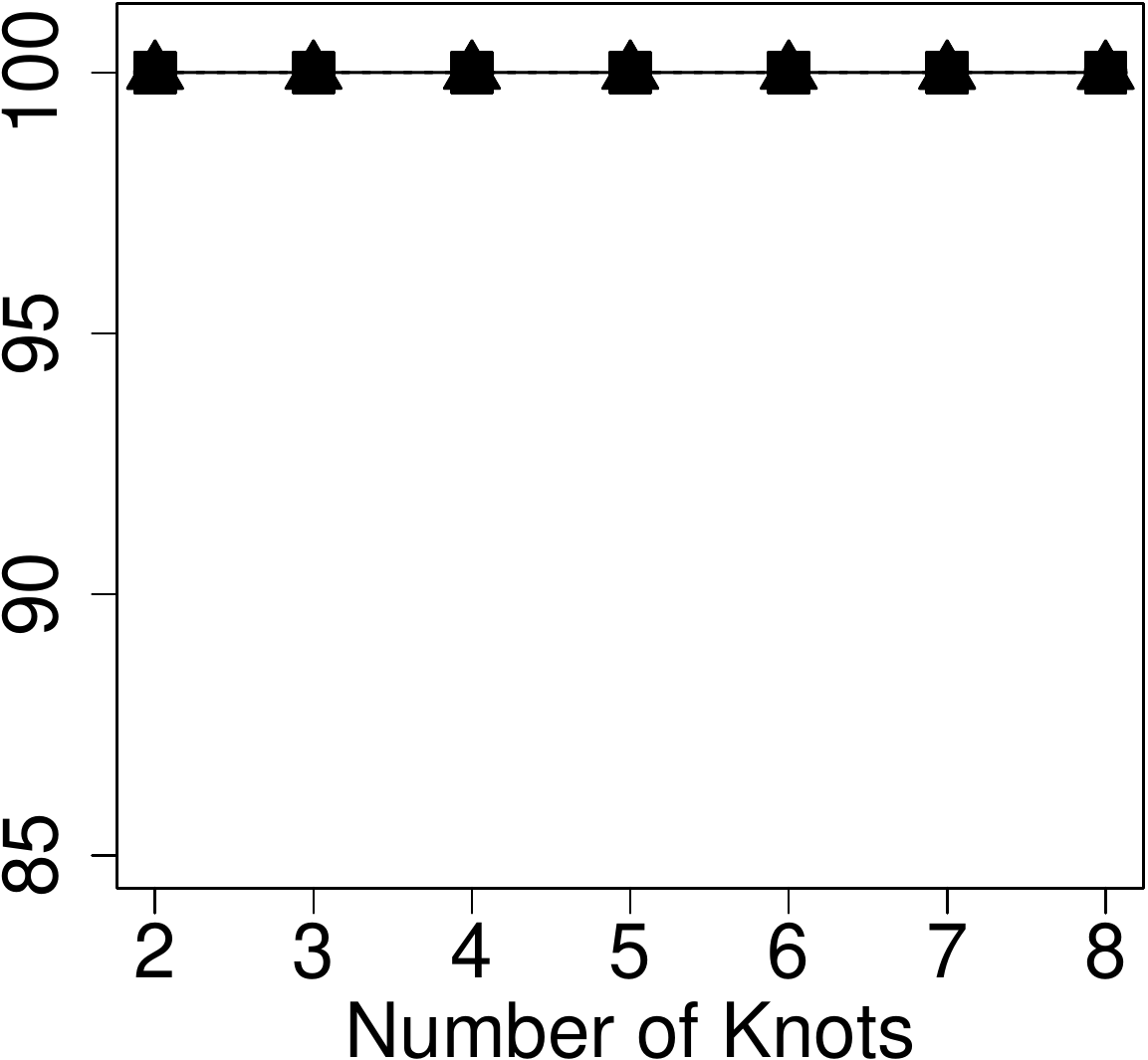} &
	\includegraphics[scale=0.39]{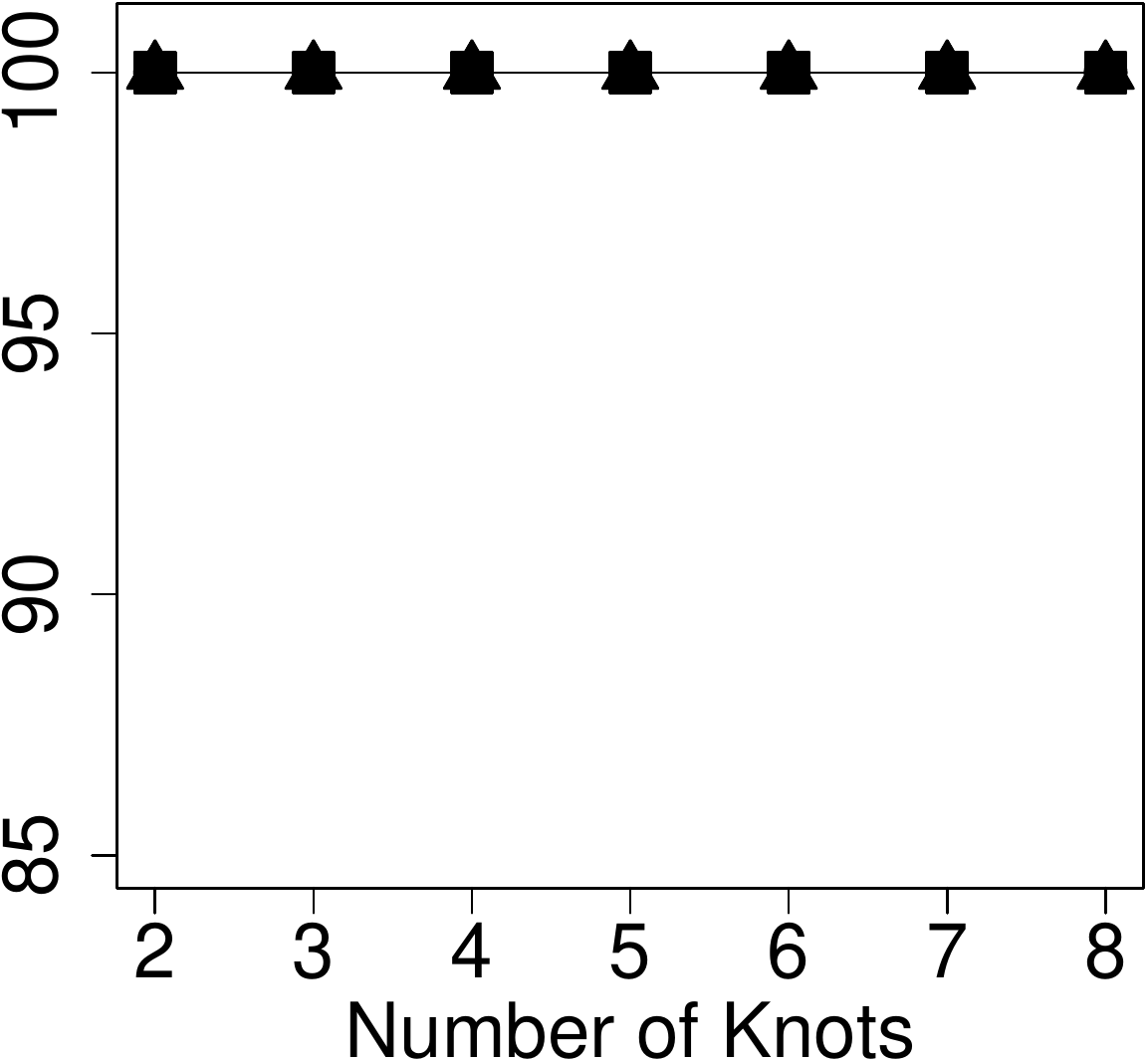} \\[6pt]
	corrN & corrLN & corrX0  \\
	\includegraphics[scale=0.39]{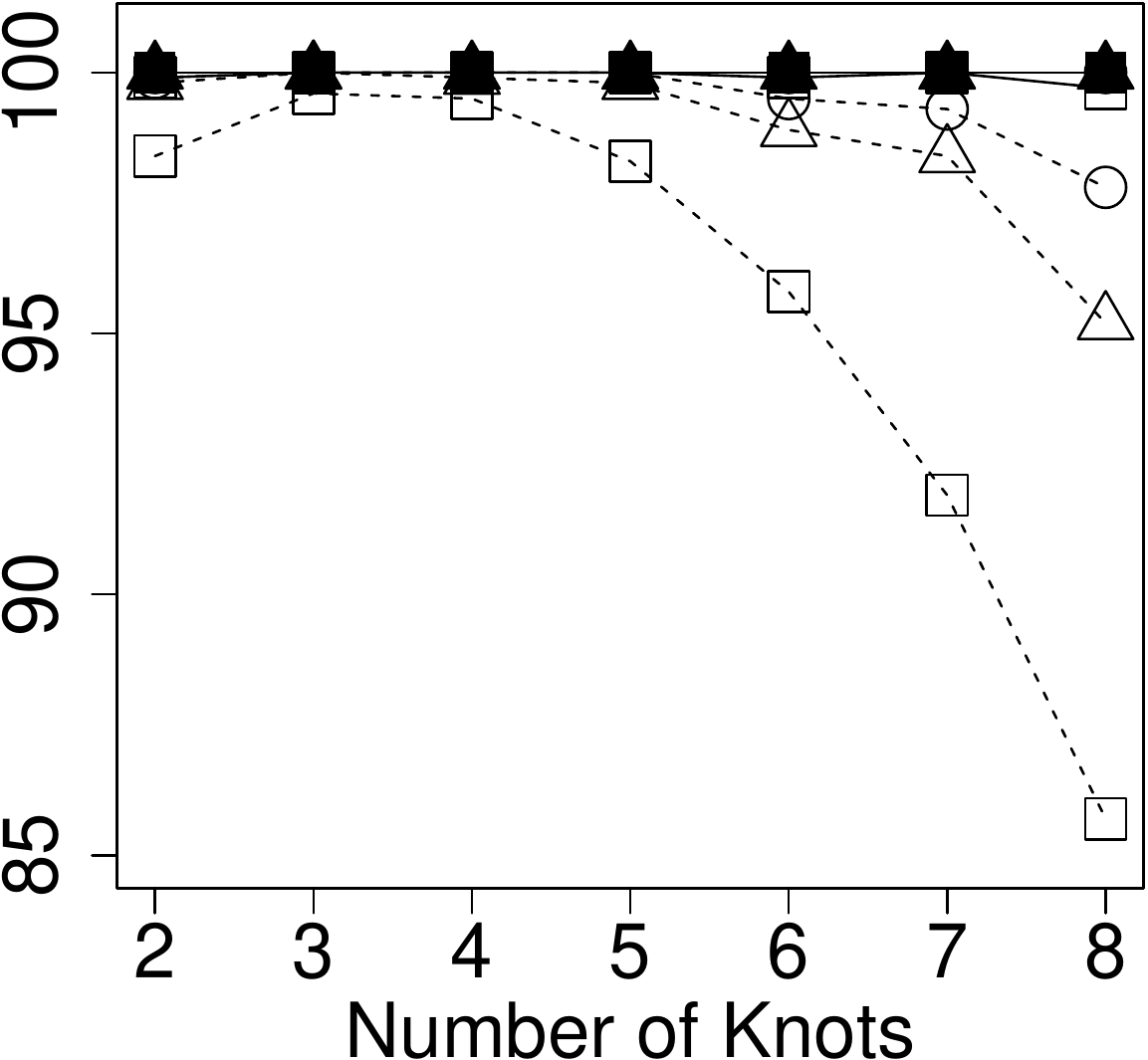}  &
	\includegraphics[scale=0.39]{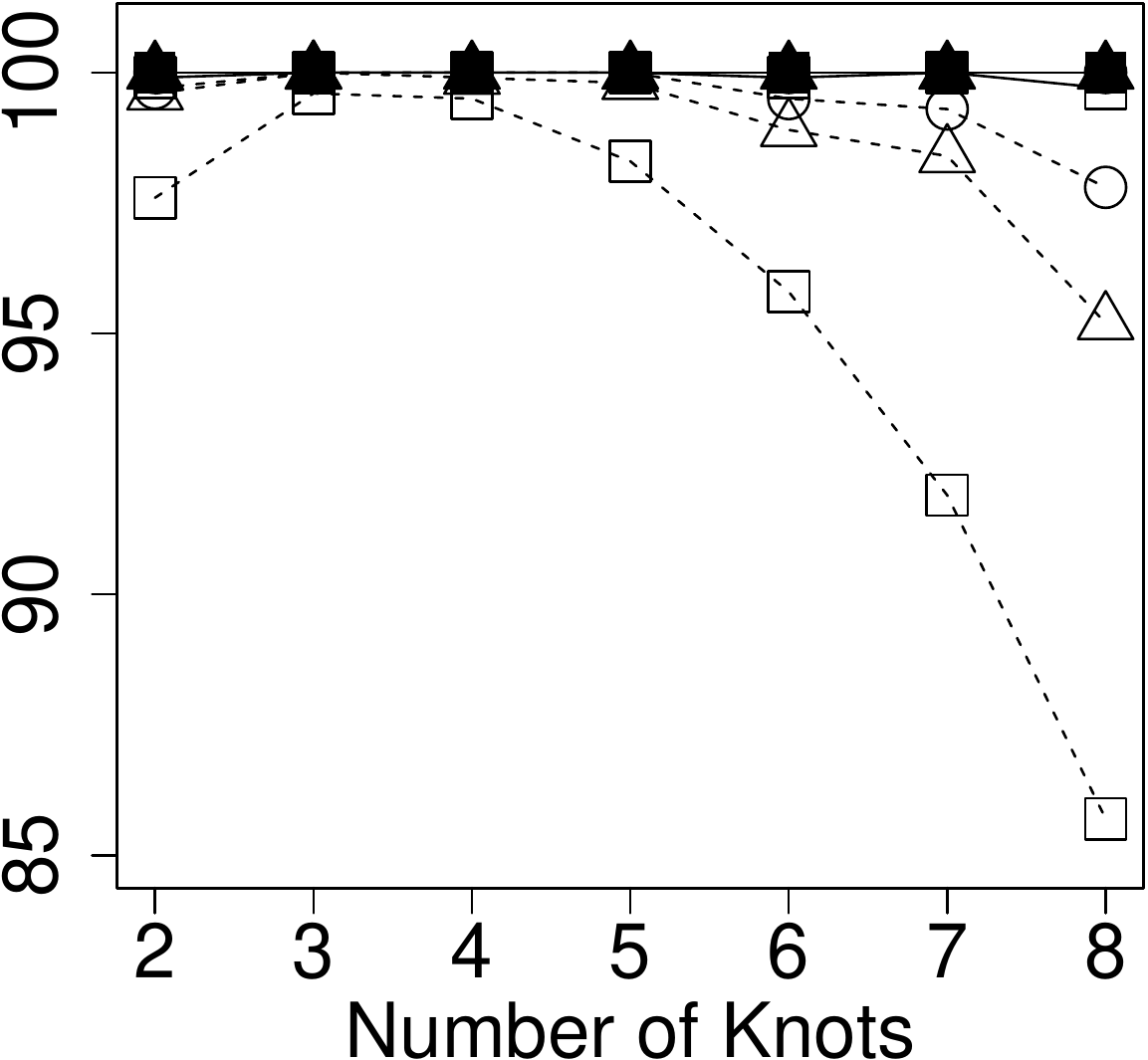} &
	\includegraphics[scale=0.39]{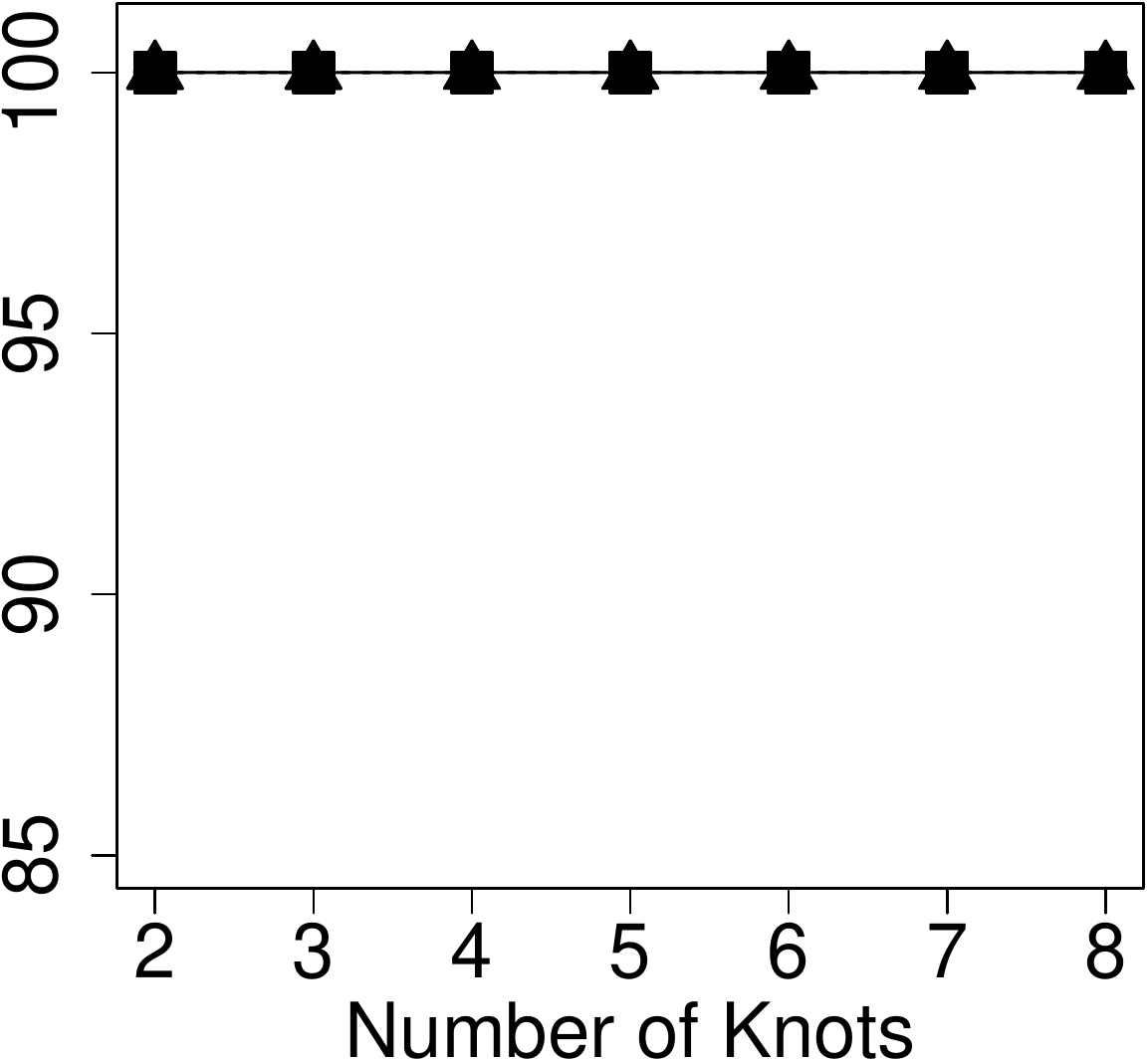} \\[6pt]
	Zto0 & LtoN & Nto0  \\	
	\includegraphics[scale=0.39]{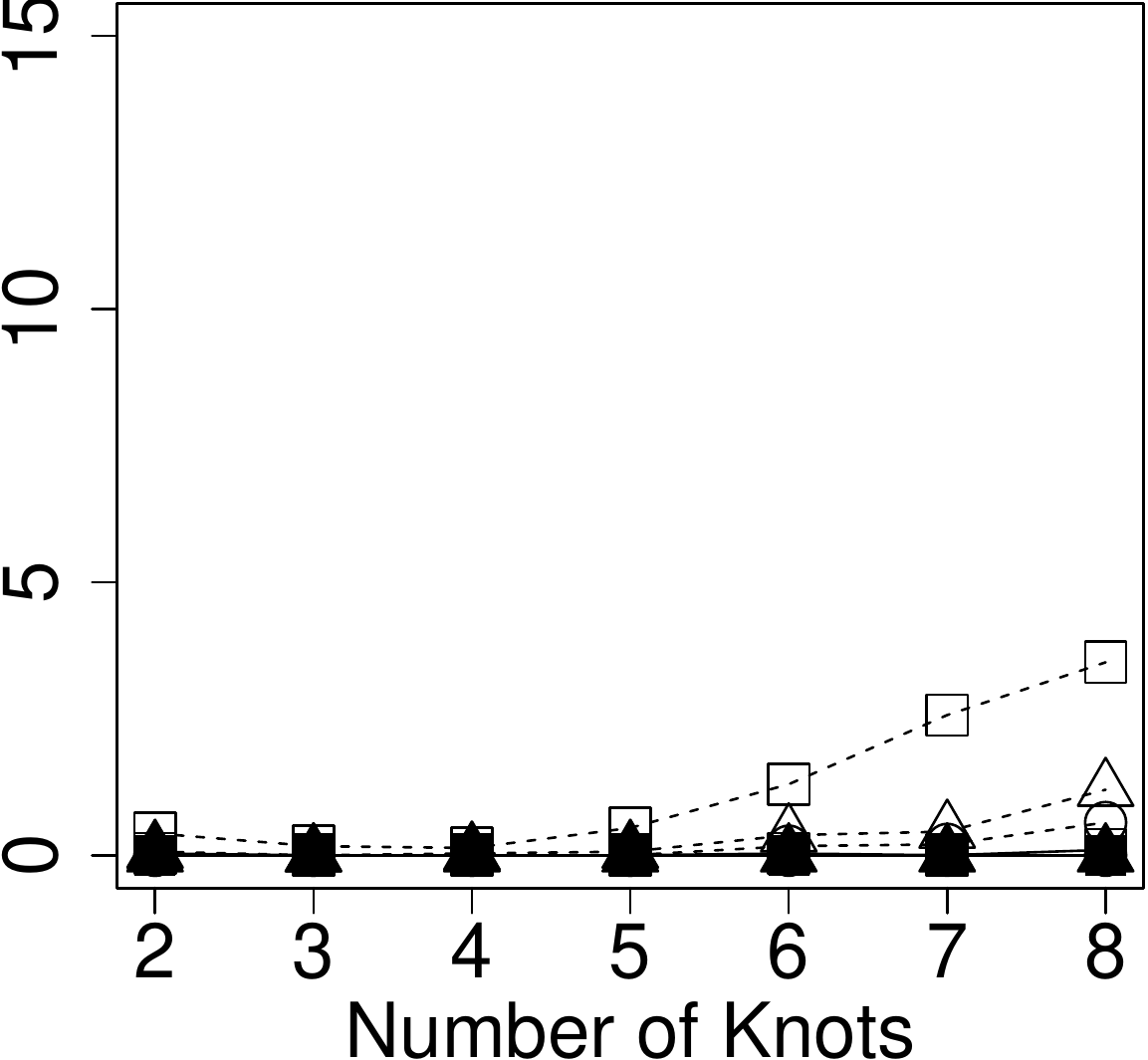}  &
	\includegraphics[scale=0.39]{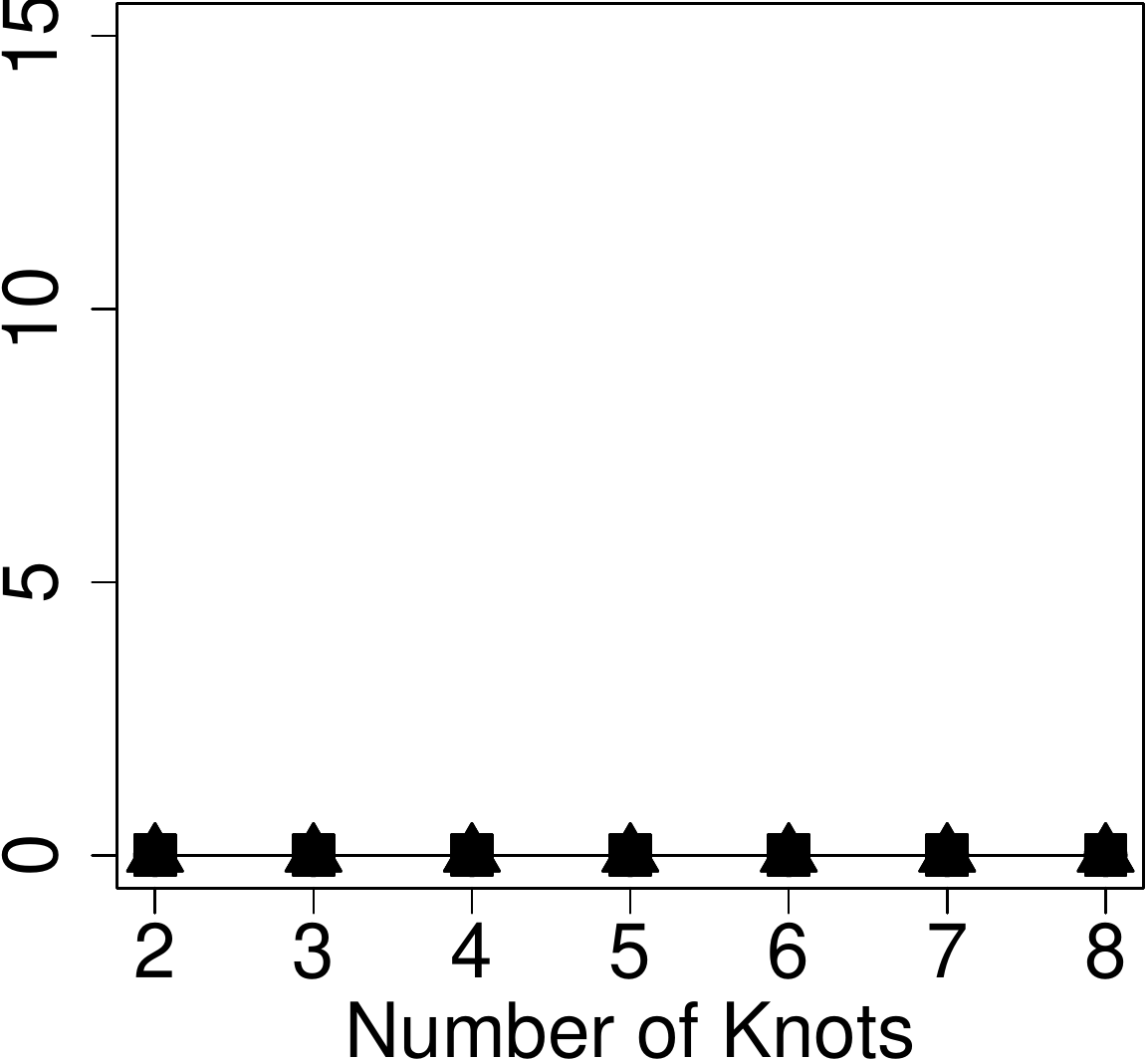}  &
	\includegraphics[scale=0.39]{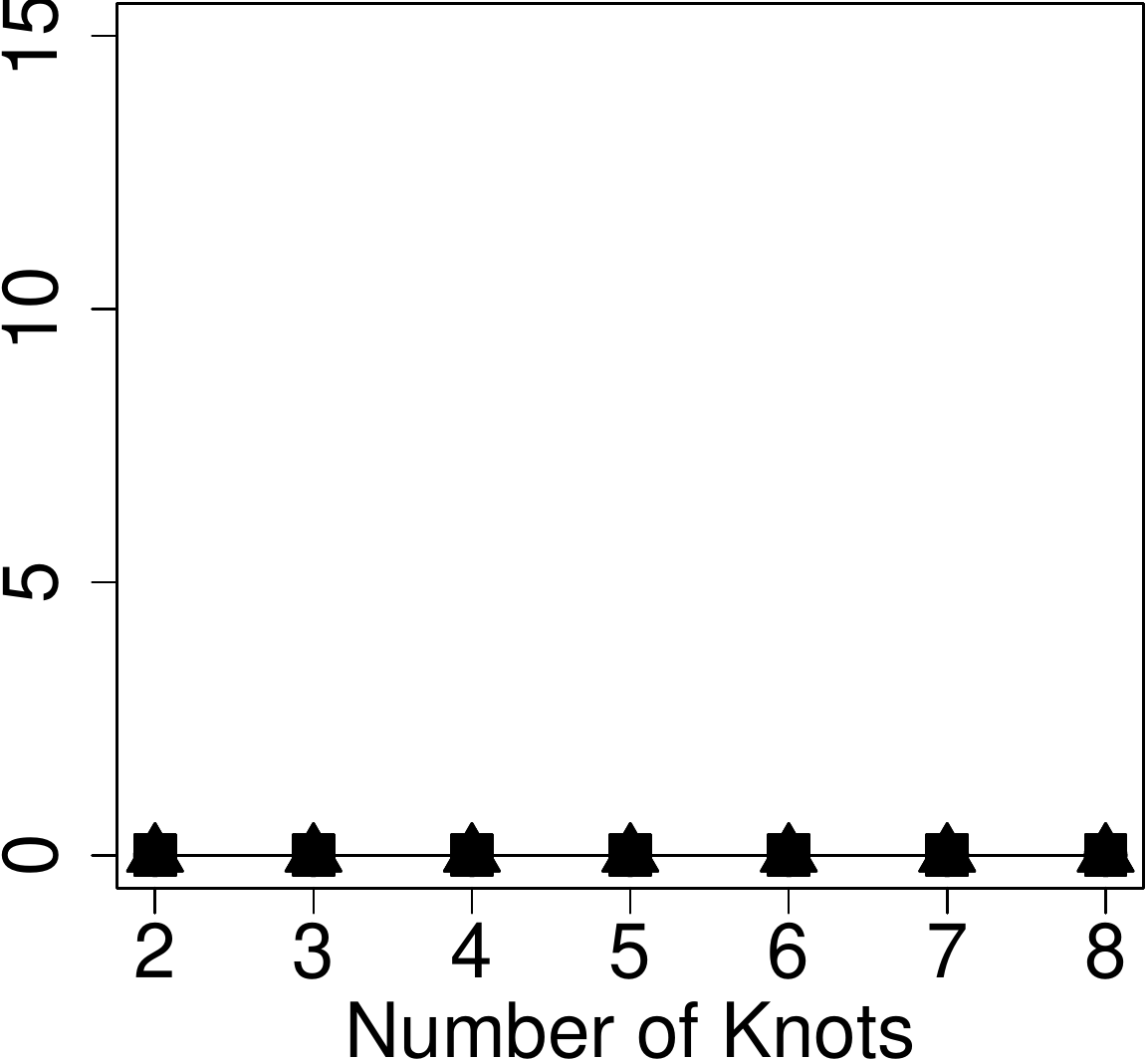} \\[6pt]
	Xto0 & Legend &   \\
	\includegraphics[scale=0.39]{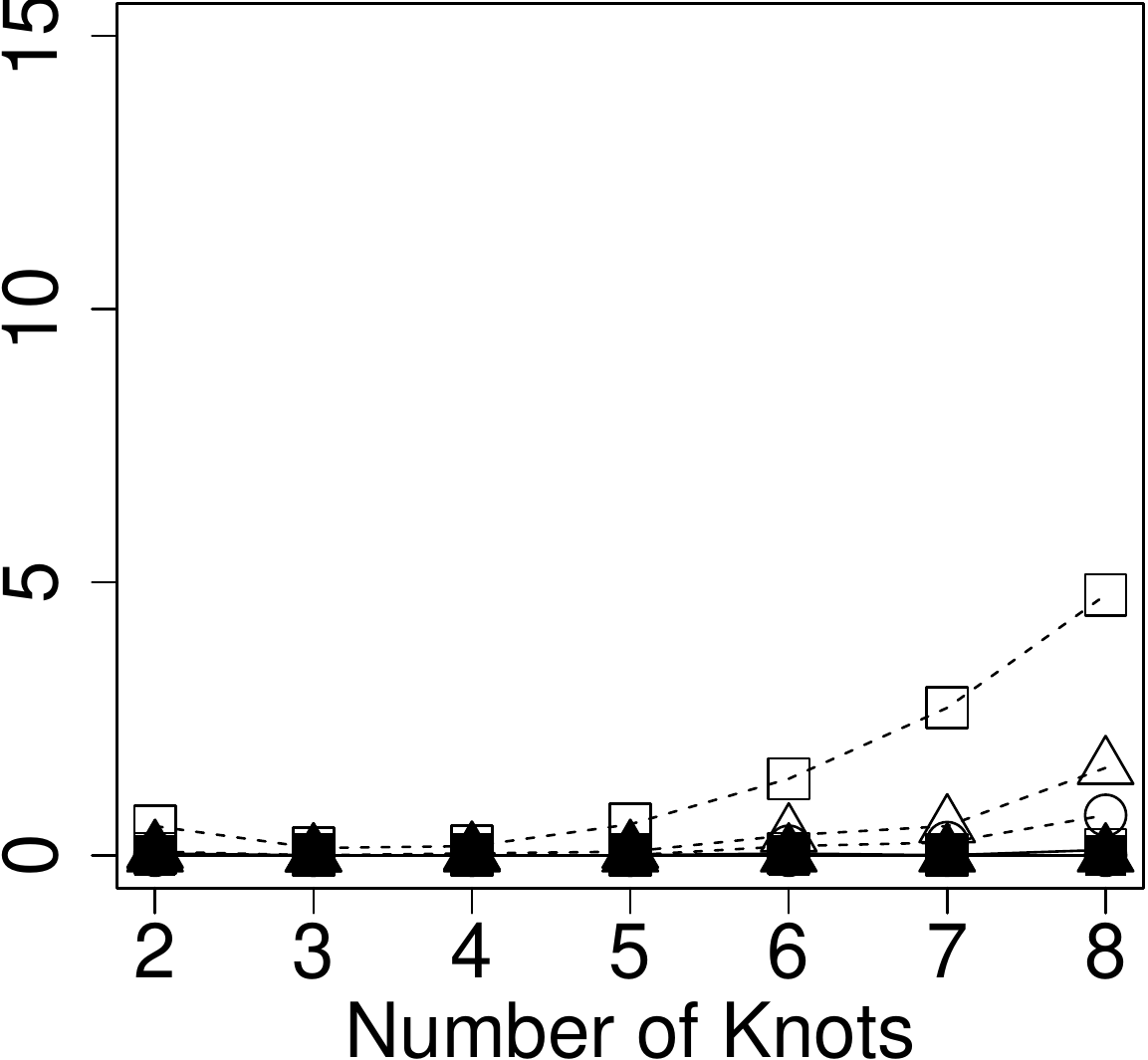} &
	\includegraphics[scale=0.39,trim=-.3in -0.55in 0 0]{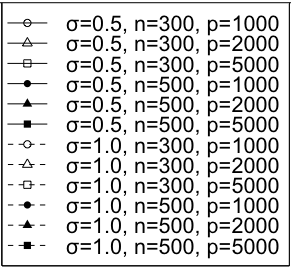} \\[6pt]
\end{tabular}
\caption{First stage selection results using different number of knots.}
\label{FIG:knots}
\end{center}
\end{figure}

In summary, the values of $N$ often have little effect on the model selection results. Choosing small values of $N$ can also help to reduce computational burden. So we recommend using fewer knots at the model selection stage, especially when the sample size is small compared to the number of predictors. In practice, $N=2\sim 5$ usually would be adequate to identify the model structure.

Next, we study the effect of the smoothing parameters at the refitting stage. For the selected model, we approximate the nonlinear functional components using higher order polynomial splines to obtain more accurate pilot estimators. Then we apply spline backfitted local-linear smoothing to obtain the final SBLL estimators and the corresponding SCBs. According to Assumption (A6$^\prime$), to obtain the SCB with the desired confidence level, the number of interior knots $M_n$ for a refitting spline needs to satisfy: $\{n^{1/(2d)} \vee n^{4/(10d-5)}\} \ll M_n \ll n^{1/3}$, where $d$ is the degree of the polynomial spline basis functions used in the refitting. The widely used quadratic/ cubic splines and any polynomial splines of degree $d\geq 2$ all satisfy this condition. Therefore, in practice we suggest choosing
\begin{equation*}
	M_n=\min\{\lfloor n^{1/(2d) \vee 4/(10d-5)}\ln(n)\rfloor, \lfloor n/(4s)\rfloor\}+1,
\end{equation*}
where $s$ is the number of nonlinear components selected at the first stage and the term $\lfloor n/(4s)\rfloor$ is to guarantee that we have at least four observations in each subinterval between two adjacent knots to avoid getting (near) singular design matrices in the spline smoothing.
A researcher with some knowledge of the shape of the nonlinear component may be able to select a more suitable number of knots. In our simulation studies, we try $4$, $6$ and $8$ interior knots to test the sensitivity of the SBLL estimators and the corresponding SCBs.

For the local-linear smoothing in the backfitting, Condition (B2) requires that the bandwidths are of order $n^{-1/5}$. Any bandwidths with this rate lead to the same limiting distribution for $\widehat{\phi}_{\ell}^{\mathrm{SBLL}}$, so the user can consider any standard routine for bandwidth selection. There have been many proposals for bandwidth selection in the literature. In our simulation, we consider three popular bandwidth selectors described in \cite{Fan:Gijbels:96} and \cite{Wand:Jones:95}: rule-of-thumb bandwidth (``thumbBw''), plug-in bandwidth selector (``pluginBw'') and leave-one-out cross-validation bandwidth selector (``regCVBwSelC''). Below we present simulation results to compare the performance of three bandwidth selectors. The kernel that we use here is the Epanechnikov kernel: $K(u)=3/4(1-u^{2})I(|u|\leq 1)$.

To see how the refitting smoothing parameters affect estimation accuracy, we report the average mean square errors (AMSEs) of the SBLL estimators based on $4$, $6$ and $8$ interior knots in the spline refitting and three different bandwidth selectors in the kernel backfitting. Figure \ref{FIG:AMSE-knots-BW} presents the AMSEs of the resulting SBLL estimators based on different combinations of the refitting smoothing parameters. For both $\phi_1$ and $\phi_2$, the AMSEs are very similar across the different combinations of knots and bandwidth selectors.

\begin{figure}[htbp]
\begin{center}
\begin{tabular}{cc}
	Refitting $\phi_2$ with $4$ interior knots & Refitting $\phi_3$ with $4$ interior knots \\
	\includegraphics[scale=0.4]{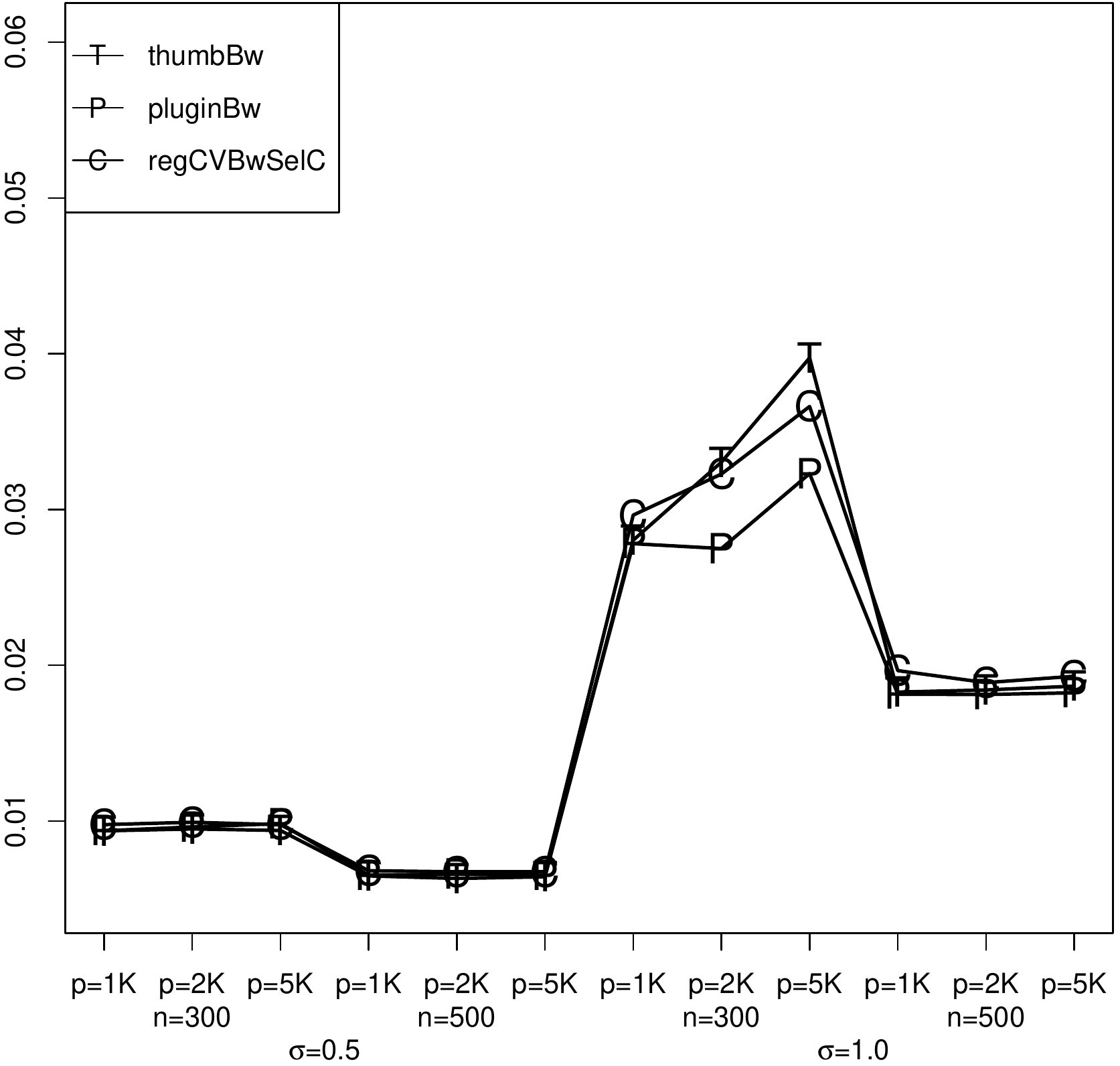}  &
	\includegraphics[scale=0.4]{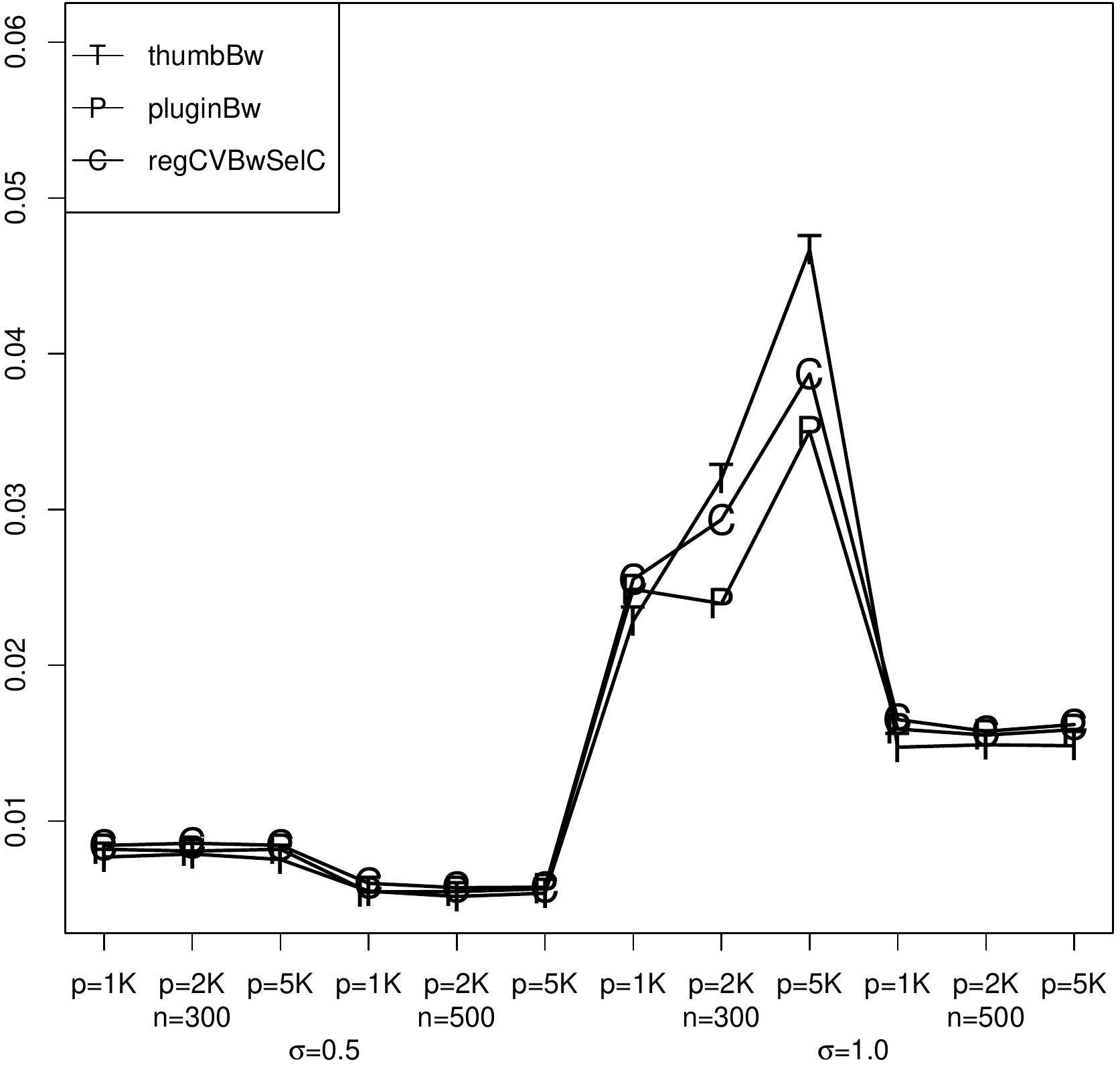} \\
	Refitting $\phi_2$ with $6$ interior knots & Refitting $\phi_3$ with $6$ interior knots \\
	\includegraphics[scale=0.4]{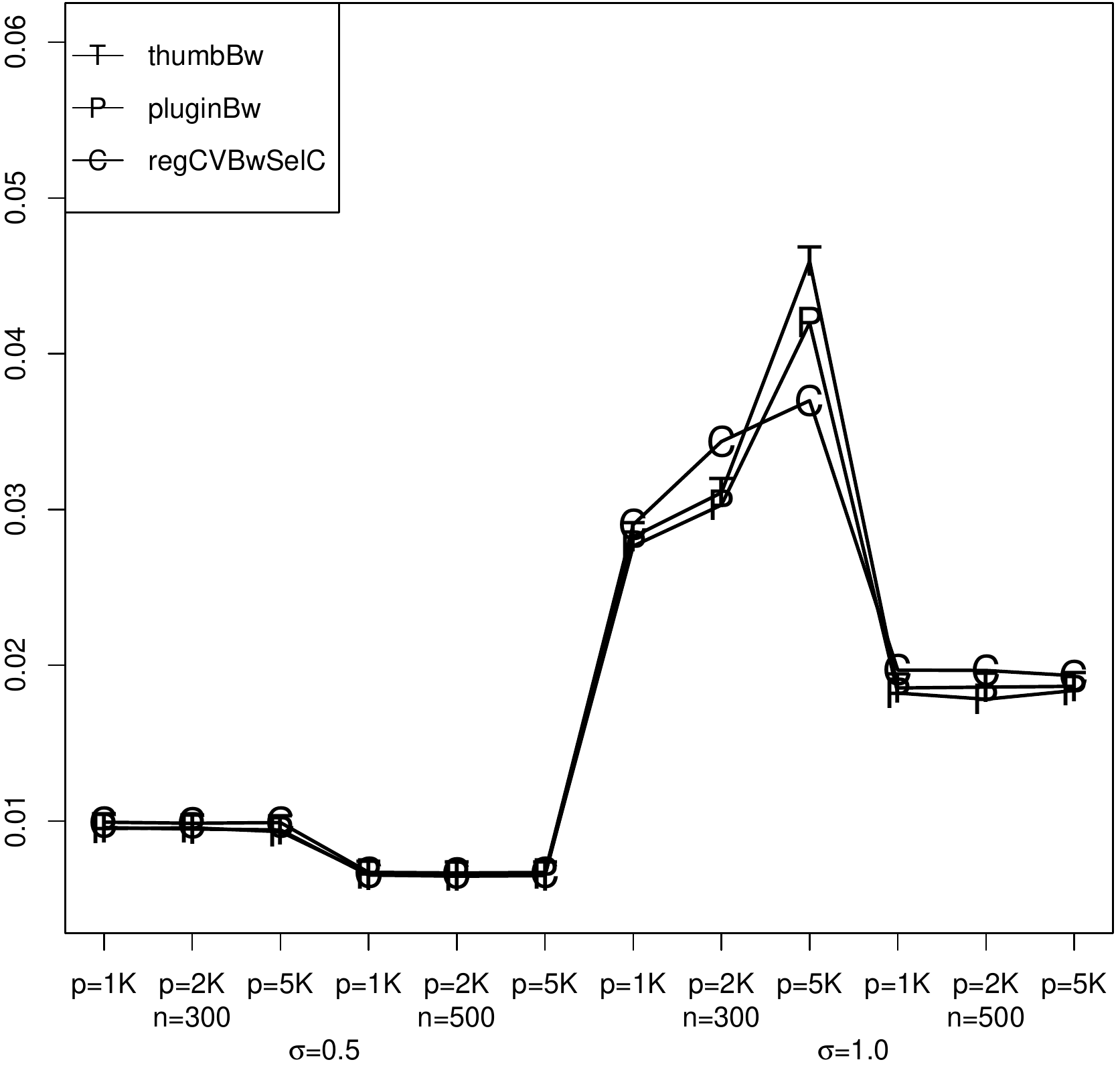}  &
	\includegraphics[scale=0.4]{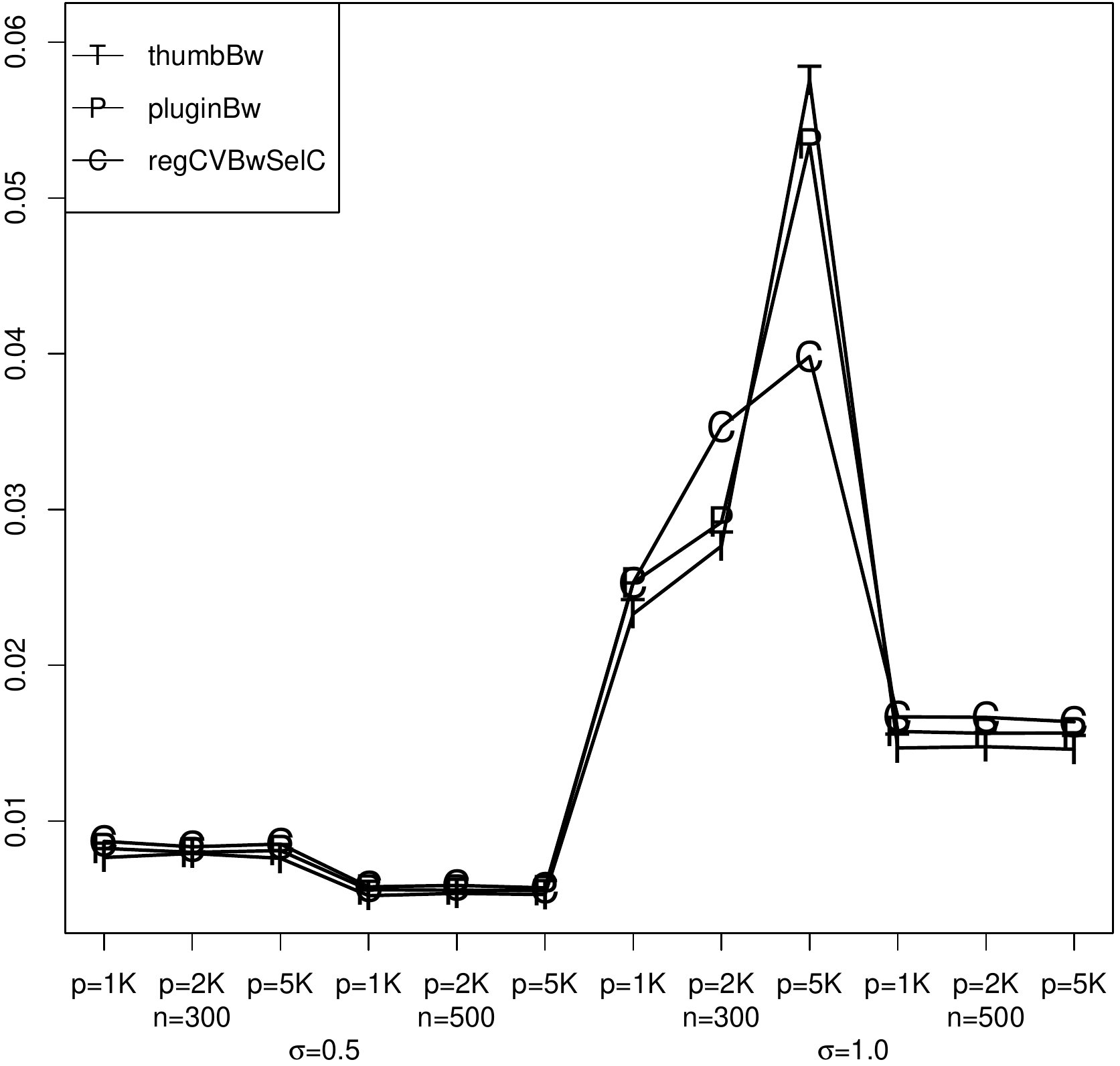} \\
	Refitting $\phi_2$ with $8$ interior knots & Refitting $\phi_3$ with $8$ interior knots \\
	\includegraphics[scale=0.4]{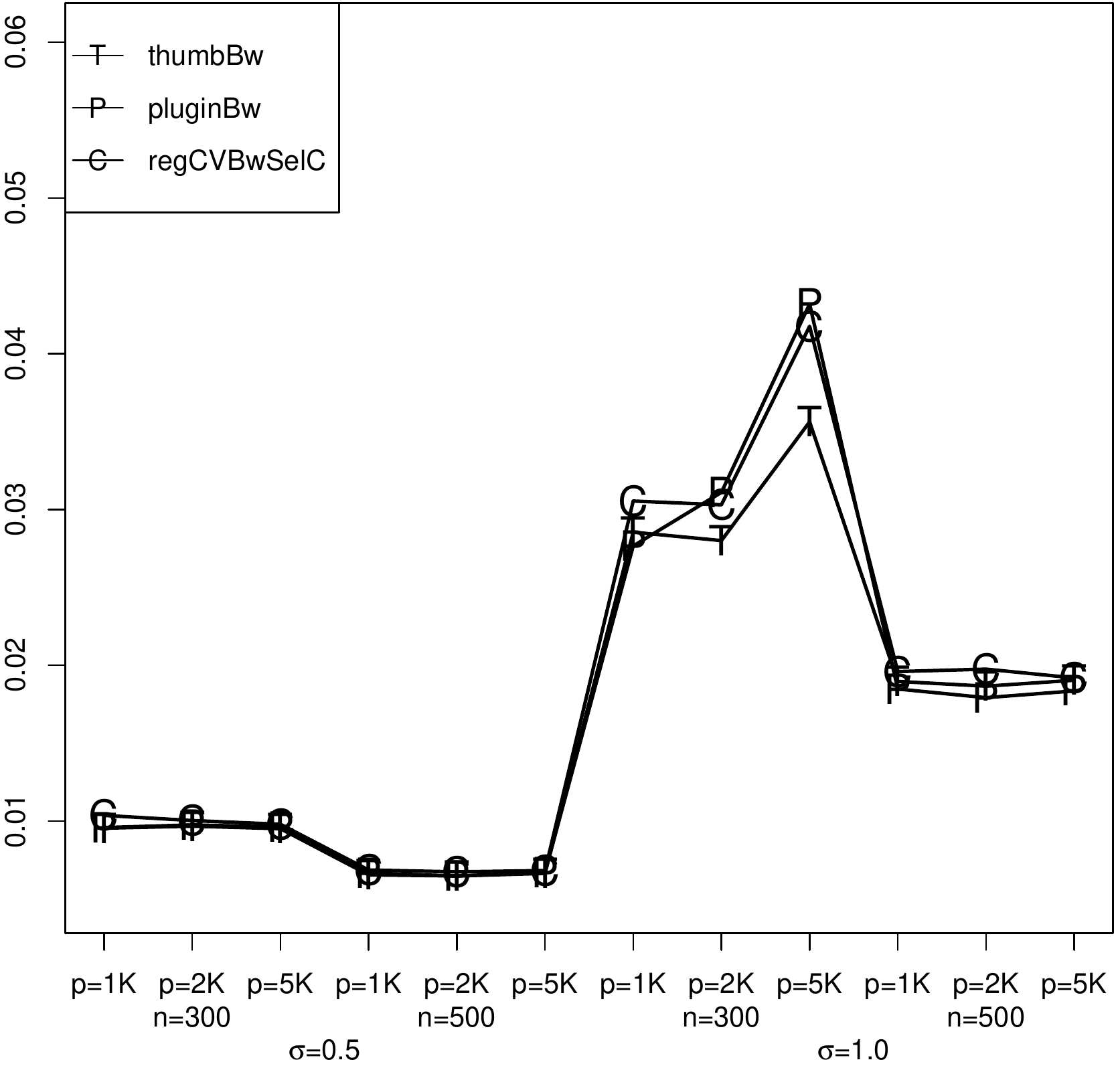}  &
	\includegraphics[scale=0.4]{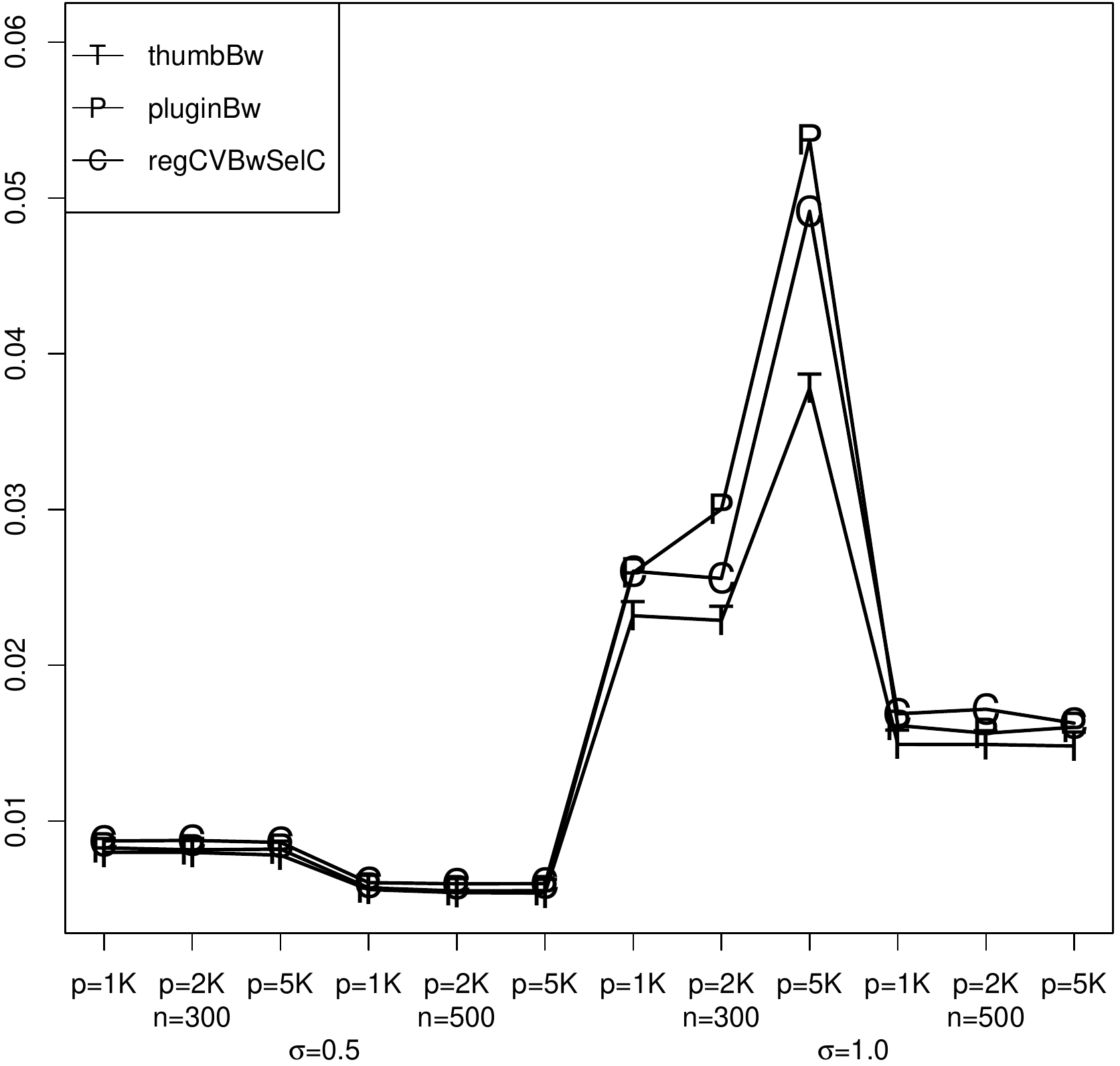} \\
\end{tabular}
\caption{Average mean squared errors (AMSEs) of the SBLL estimators of $\phi_2$ and $\phi_3$.}
\label{FIG:AMSE-knots-BW}
\end{center}
\end{figure}

Figure \ref{FIG:SCB-knots-BW} shows the coverage rates of the SCBs based on different combinations of knots and bandwidth selectors. From Figure \ref{FIG:SCB-knots-BW}, it is clear that the number of knots for a spline in the refitting has very little effect on the coverage of the SCBs. One also observes that the performances of the SCBs based on different smoothing parameters become more similar with increasing sample size, whereas the coverage rates of the SCBs using the ``thumbBw'' selector are the closest to the nominal level in all the simulation settings.
Thus we recommend the ``thumbBw'' selector, especially when the sample size is small.

\begin{figure}[htbp]
\begin{center}
\begin{tabular}{cc}
	Refitting $\phi_2$ with $4$ interior knots & Refitting $\phi_3$ with $4$ interior knots \\
	\includegraphics[scale=0.4]{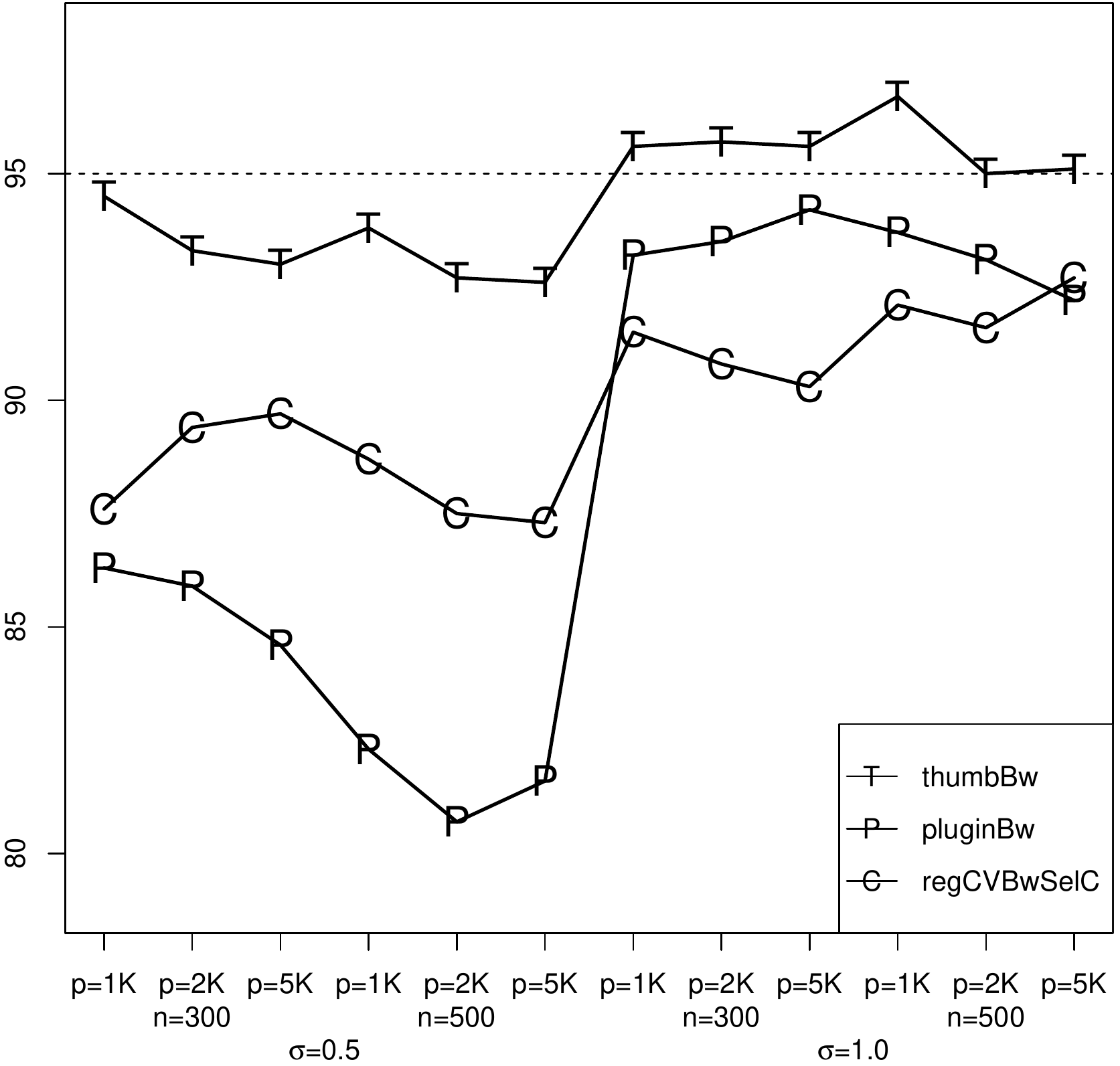}  &
	\includegraphics[scale=0.4]{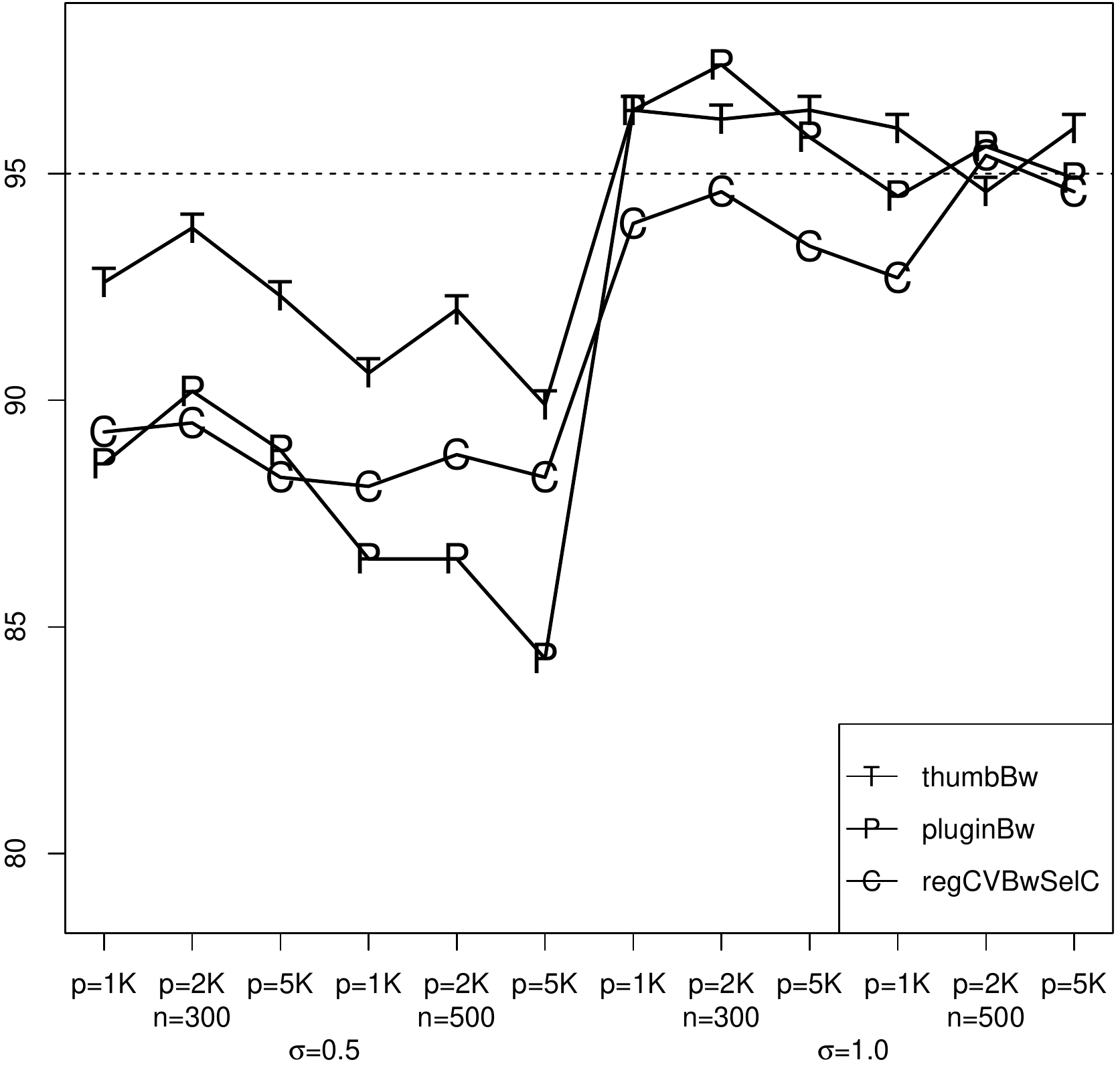} \\
	Refitting $\phi_2$ with $6$ interior knots & Refitting $\phi_3$ with $6$ interior knots \\
	\includegraphics[scale=0.4]{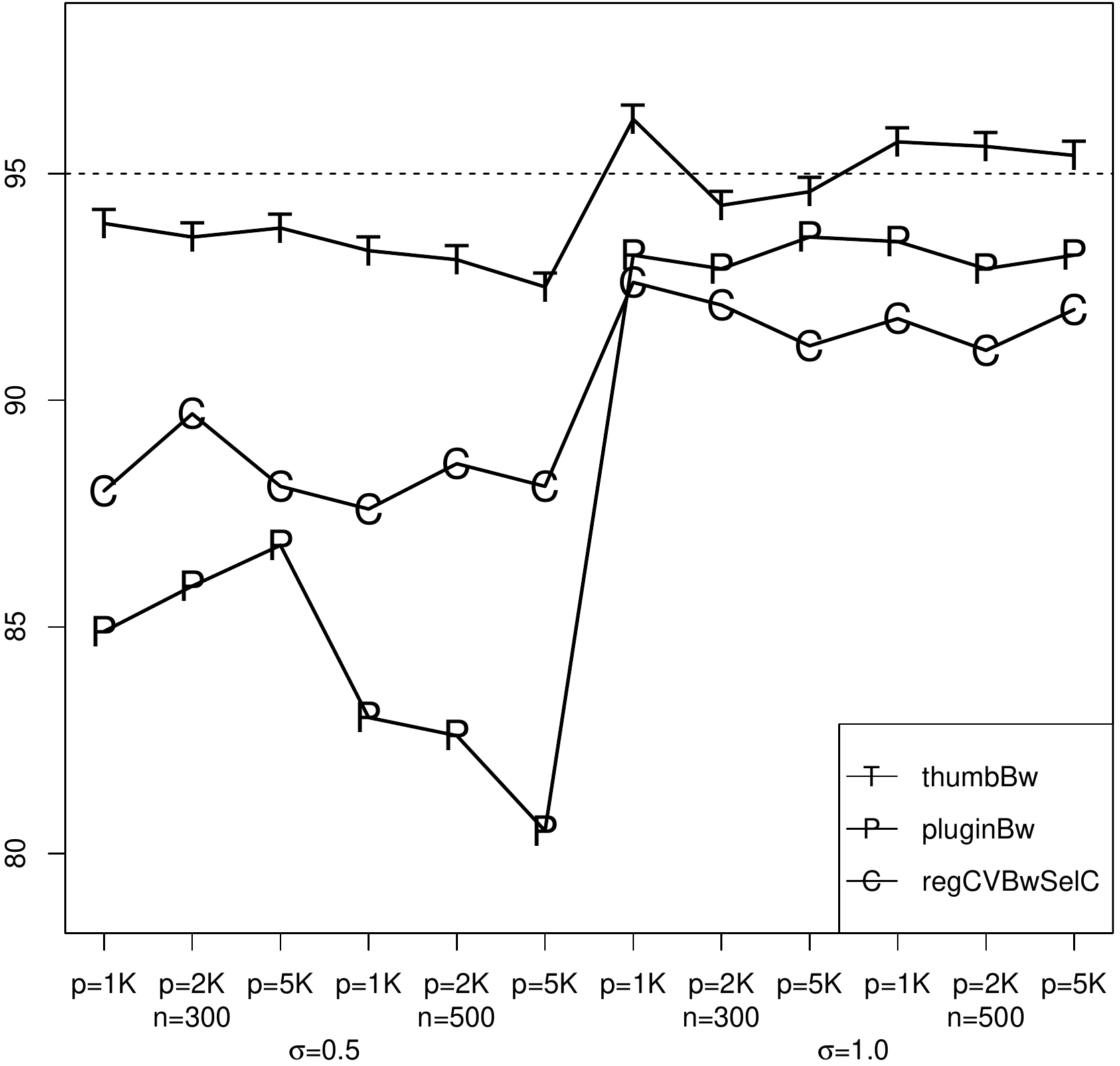}  &
	\includegraphics[scale=0.4]{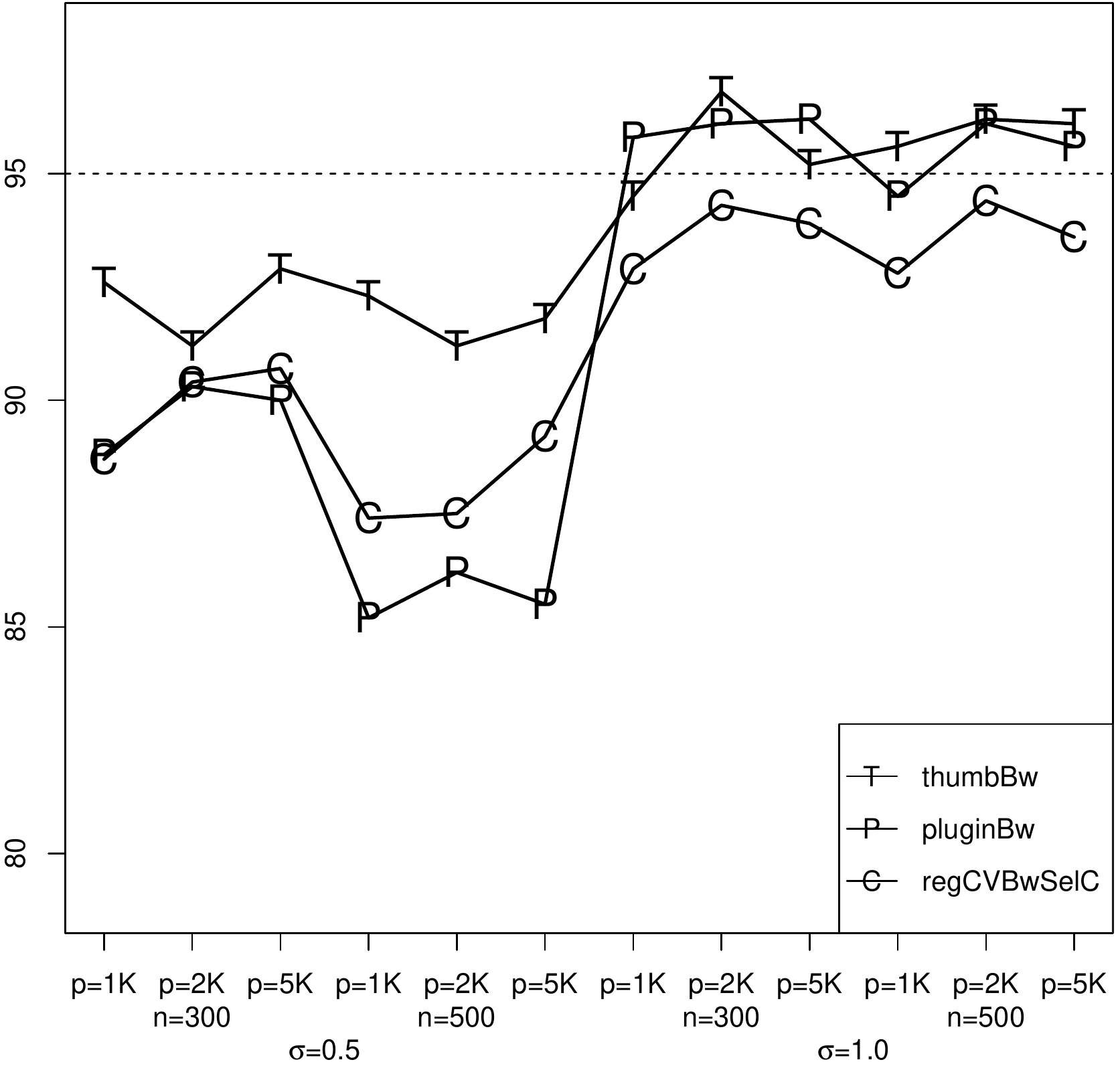} \\
	Refitting $\phi_2$ with $8$ interior knots & Refitting $\phi_3$ with $8$ interior knots \\
	\includegraphics[scale=0.4]{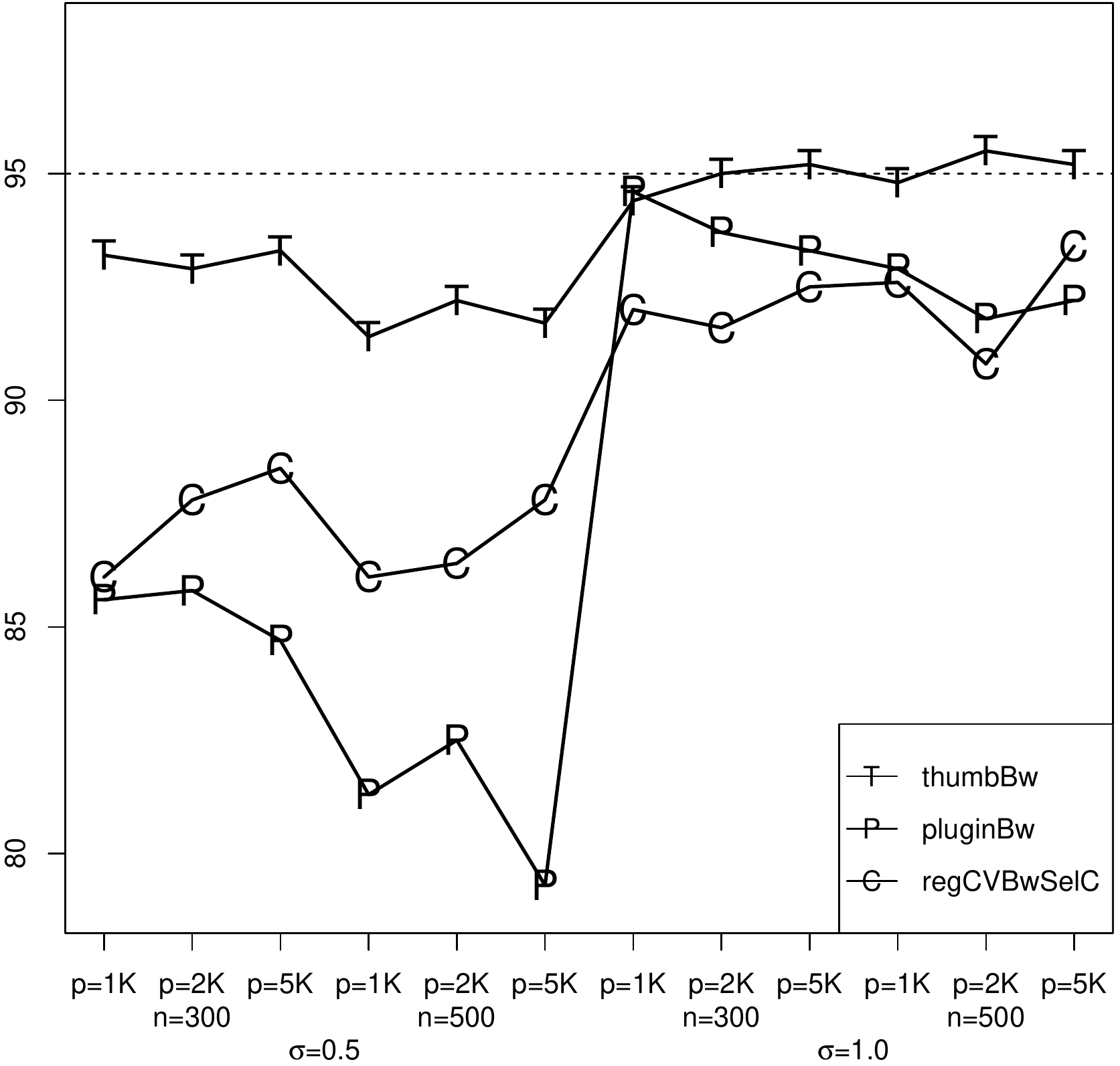}  &
	\includegraphics[scale=0.4]{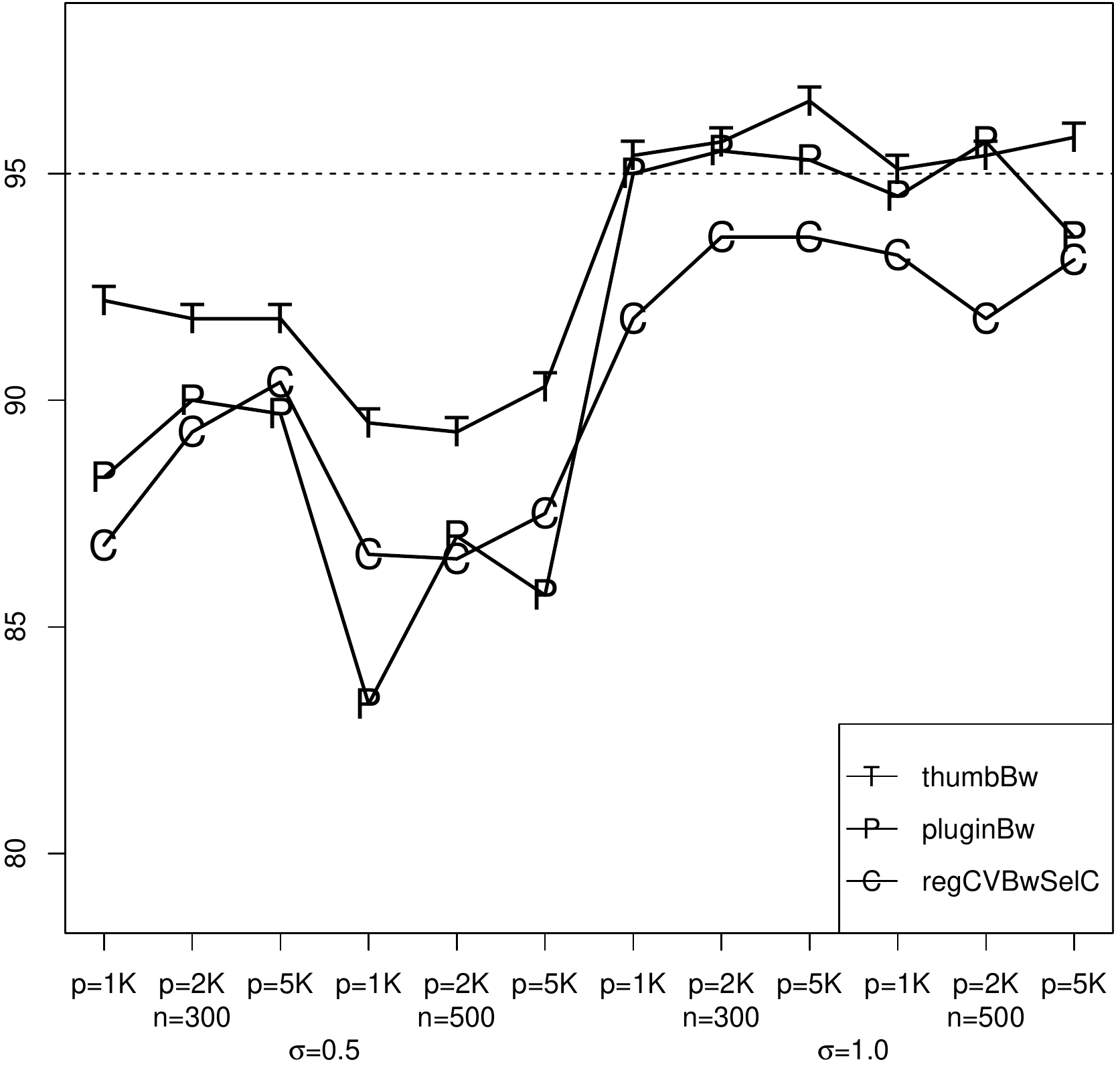} \\
\end{tabular}
\caption{Coverage rates of the SCBs for $\phi_2$ and $\phi_3$.}
\label{FIG:SCB-knots-BW}
\end{center}
\end{figure}

\setcounter{chapter}{9}
\renewcommand{\theequation}{B.\arabic{equation}}
\renewcommand{\thesection}{B.\arabic{section}}
\renewcommand{\thesubsection}{B.\arabic{subsection}}
\renewcommand{\thetheorem}{B.\arabic{theorem}}
\renewcommand{\thelemma}{B.\arabic{lemma}}
\renewcommand{\theproposition}{B.\arabic{proposition}}
\renewcommand{\thecorollary}{B.\arabic{corollary}}
\renewcommand{\thefigure}{B.\arabic{figure}}
\renewcommand{\thetable}{B.\arabic{table}}
\setcounter{table}{0}
\setcounter{figure}{0}
\setcounter{equation}{0}
\setcounter{theorem}{0}
\setcounter{lemma}{0}
\setcounter{proposition}{0}
\setcounter{section}{0}
\setcounter{subsection}{0}

\section*{B. Simulation Studies Using Purely Additive Models or Purely Linear Models}

In this section, we examine the performance the proposed method when the underlying model is either a purely additive model (AM) or a purely linear model (LM). We evaluated the selection, estimation and prediction accuracy, and inference performance of the proposed SMILE method. We also compared the performance of SMILE with the sparse APLM estimator with adaptive group LASSO penalty (SAPLM), the ordinary linear least squares estimator with the adaptive LASSO penalty (SLM), and the oracle estimator (ORACLE), which uses the same estimation techniques as the SMILE except that no penalization or data-driven variable selection is used because all active and inactive index sets are treated as known. All the performance measures were computed based on 200 replicates.

\vskip .1in
\noindent{\bf Case I. A Purely Additive Model.} We generate simulated datasets using the AM structure
\begin{align*}
	Y_i = \sum_{\ell=1}^{p} \phi_{\ell}(X_{i\ell}) + \varepsilon_i,
\end{align*}
		where
\begin{align*}
	\phi_1(x) &= \frac{8\sin(2\pi x)}{2-\sin(2\pi x)} - \mathrm{E}\left\{\frac{8\sin(2\pi X_1)}{2-\sin(2\pi X_1)}\right\}, \\
	\phi_2(x) & = -3 \cos^2(\pi x) + 6 \sin^2(\pi x) - \mathrm{E}\{-3 \cos^2(\pi X_2 + 6 \sin^2(\pi X_2)\}, \\
	\phi_3(x) & = 6x + 18x^{2} - \mathrm{E}(6X_3 + 18X_3^{2} ),
\end{align*}
and $\phi_4(x) = \ldots = \phi_{p}(x) = 0$.

\vskip .1in
\noindent{\bf Case II. A Purely Linear Model.} We generate simulated datasets using the LM structure:
\begin{align*}
	Y_i = \sum_{\ell=1}^{p} \beta_{\ell} X_{i\ell} + \varepsilon_i,
\end{align*}
where $\beta_1=3$, $\beta_2=4$, $\beta_3=-2$, and $\beta_4 = \ldots = \beta_{p} = 0$.

We use  the criteria mentioned in Section \ref{sec:simulation} to evaluate the methods on the accuracy of variable selection and prediction. The model selection results are provided in Tables \ref{TAB:AM1} and \ref{TAB:LM1}, respectively. The SMILE can correctly discover the linear or nonlinear structure in covariates $\bs{X}$, while the SAPLM neglects linear structure in $\bs{X}$ and SLM fails in presenting the nonlinear part of covariates $\bs{X}$. For the SMILE, regardless of the underlying models, the percents of nonzero covariates correctly selected are very close to ORACLE (100 for corrX and corrX0, respectively), as shown in Table \ref{TAB:AM1}; and the percents of components incorrectly identified approach to 0 as the sample size $n$ increases, as shown in Table \ref{TAB:AM1}. The SMILE is close in the selection of nonlinear covariates to the SAPLM estimator, and it overwhelms the SAPLM in identifying the linear structure of covariates $\bs{X}$. SMILE is close in the selection of linear covariates to the SLM estimator, and it overwhelms the SLM in identifying the nonlinear structure of covariates $\bs{X}$. From the results in Tables \ref{TAB:AM1} and \ref{TAB:LM1}, it is also evident that model misspecification leads to poor variable selection performance for the SLM, as the SLM fails to select the right nonlinear components in each simulation.

The estimation and prediction results are displayed in Tables \ref{TAB:AM2} and \ref{TAB:LM2}. Specifically, we present the AMSEs for functions $\phi_1$, $\phi_2$ and $\phi_3$ and the CV-MSPEs for predicting $Y$. The case with known active covariates (ORACLE) is also reported in each setting and serves as a benchmark. The SMILE performs well in predicting $Y$ regardless of the model structure for the underlying model, as indicated by results closest to ORACLE in CV-MSPE for base cases. The SLM is around 18$\sim$36 times higher than the SMILE in the AM case, and the SAPLM is around 1$\sim$3 times higher than the SMILE in the LM case. The estimation of functions $\phi_1$, $\phi_2$ and $\phi_3$ is also good for the SMILE, and matches with ORACLE as sample size $n$ increases. The inferior performance of the SAPLM in the LM case and the poor performance of SLM in the AM case, in both estimation and prediction, illustrates the importance and necessity of identifying correct model structure.

\begin{table}[htbp]
\caption{AM Case: Selection statistics comparing the SMILE, SAPLM and SLM.}
\label{TAB:AM1}
\begin{center}
\renewcommand{\arraystretch}{0.9}
\begin{tabular}{ccclrrrr}
	\hline\hline
	Size  & Noise &	   &	   & \multicolumn{2}{c}{True Selection} & \multicolumn{2}{c}{False Selection} \\
	\cmidrule(r){5-6} \cmidrule(l){7-8}
	$n$ & $\sigma$ & $p$ & Method & corrX & corrX0 & NtoL & Nto0 \\
	\hline
	300   & 0.5   & 1000  & SMILE & 100   & 99.9995 & 2.3333 & 0 \\
          &       &       & SAPLM & 100   & 100   & 0     & 0 \\
          &       &       & SLM   & 65.1667 & 99.9985 & 65.1667 & 34.8333 \\
          &       & 2000  & SMILE & 100   & 99.9998 & 3.1667 & 0 \\
          &       &       & SAPLM & 100   & 100   & 0     & 0 \\
          &       &       & SLM   & 64    & 99.9995 & 64.0003 & 36.0000 \\
          &       & 5000  & SMILE & 99.6667 & 99.9994 & 7     & 0.3333 \\
          &       &       & SAPLM & 100   & 100   & 0     & 0 \\
          &       &       & SLM   & 63.6667 & 99.9998 & 63.6667 & 36.3333 \\
          & 1.0   & 1000  & SMILE & 100   & 100   & 6.5000 & 0 \\
          &       &       & SAPLM & 100   & 100   & 0     & 0 \\
          &       &       & SLM   & 64.6667 & 99.9980 & 64.6667 & 35.3333 \\
          &       & 2000  & SMILE & 99.8333 & 99.9995 & 7.8333 & 0.16667 \\
          &       &       & SAPLM & 100   & 100   & 0     & 0 \\
          &       &       & SLM   & 63.5  & 99.9995 & 63.5003 & 36.5000 \\
          &       & 5000  & SMILE & 99.1667 & 99.9994 & 13.8333 & 0.8333 \\
          &       &       & SAPLM & 100   & 100   & 0     & 0 \\
          &       &       & SLM   & 62.1667 & 99.9999 & 62.1667 & 37.8333 \\
	500   & 0.5   & 1000  & SMILE & 100   & 99.9995 & 0     & 0 \\
          &       &       & SAPLM & 100  & 100   & 0     & 0 \\
          &       &       & SLM   & 66.66667 & 99.999 & 66.6667 & 33.3333 \\
          &       & 2000  & SMILE & 100   & 99.99975 & 0     & 0 \\
          &       &       & SAPLM & 100   & 100   & 0     & 0 \\
          &       &       & SLM   & 66.6667 & 99.99975 & 66.6667 & 33.3333 \\
          &       & 5000  & SMILE & 100   & 99.9999 & 0     & 0 \\
          &       &       & SAPLM & 100   & 100   & 0     & 0 \\
          &       &       & SLM   & 66.6667 & 99.9999 & 66.6667 & 33.3333 \\
          & 1.0   & 1000  & SMILE & 100   & 100   & 0     & 0 \\
          &       &       & SAPLM & 100   & 100   & 0     & 0 \\
          &       &       & SLM   & 66.6667 & 99.9980 & 66.6667 & 33.3333 \\
          &       & 2000  & SMILE & 100   & 100   & 0.16667 & 0 \\
          &       &       & SAPLM & 100   & 100   & 0     & 0 \\
          &       &       & SLM   & 66.6667 & 99.9995 & 66.6667 & 33.3333 \\
          &       & 5000  & SMILE & 100   & 99.9998 & 0     & 0 \\
          &       &       & SAPLM & 100   & 100   & 0     & 0 \\
          &       &       & SLM   & 66.6667 & 99.9999 & 66.6667 & 33.3333 \\
	\hline \hline
\end{tabular}
\end{center}
\end{table}

\begin{table}[htbp]
\caption{AM Case: Estimation statistics comparing the SMILE, SAPLM and SLM.}
\label{TAB:AM2}
\begin{center}
\renewcommand{\arraystretch}{0.72}
\begin{tabular}{ccclrrrr}
	\hline\hline
	\multicolumn{1}{l}{Size} & \multicolumn{1}{l}{Noise} &	   &	   & \multicolumn{3}{c}{AMSE} & \multicolumn{1}{c}{CV-} \\
	\cmidrule(lr){5-7}
	$n$ & $\sigma$ & $p$ & Method & $\phi_1(\cdot)$ & $\phi_2(\cdot)$ & $\phi_3(\cdot)$ &  MSPE \\
	\hline
	300   & 0.5   & 1000  & ORACLE & 0.0151 & 0.0136 & 0.0094 & 0.2924 \\
          &       &       & SMILE & 0.0236 & 0.0196 & 0.2484 & 0.5400 \\
          &       &       & SAPLM & 0.0152 & 0.0136 & 0.0094 & 0.2924 \\
          &       &       & SLM   & 6.0801 & 10.1935 & 1.9317 & 18.7746 \\
          &       & 2000  & ORACLE & 0.0156 & 0.0136 & 0.0093 & 0.2928 \\
          &       &       & SMILE & 0.0329 & 0.0282 & 0.5159 & 0.5893 \\
          &       &       & SAPLM & 0.0156 & 0.0137 & 0.0092 & 0.2928 \\
          &       &       & SLM   & 6.0919 & 10.1303 & 2.0038 & 19.1599 \\
          &       & 5000  & ORACLE & 0.0154 & 0.0131 & 0.0087 & 0.2922 \\
          &       &       & SMILE & 0.1042 & 0.0841 & 0.7177 & 1.1427 \\
          &       &       & SAPLM & 0.0154 & 0.0131 & 0.0086 & 0.2921 \\
          &       &       & SLM   & 6.1067 & 10.0619 & 2.1242 & 19.1425 \\
          & 1.0   & 1000  & ORACLE & 0.0446 & 0.0373 & 0.0256 & 1.1294 \\
          &       &       & SMILE & 0.0719 & 0.0535 & 0.7567 & 1.6505 \\
          &       &       & SAPLM & 0.0445 & 0.0373 & 0.0244 & 1.1272 \\
          &       &       & SLM   & 6.0704 & 10.1433 & 1.9960 & 19.5683 \\
          &       & 2000  & ORACLE & 0.0453 & 0.0357 & 0.0272 & 1.1313 \\
          &       &       & SMILE & 0.1605 & 0.0717 & 1.2918 & 2.0046 \\
          &       &       & SAPLM & 0.0452 & 0.0358 & 0.0259 & 1.1291 \\
          &       &       & SLM   & 6.2667 & 10.1503 & 2.1672 & 20.0240 \\
          &       & 5000  & ORACLE & 0.0459 & 0.0367 & 0.0250 & 1.1275 \\
          &       &       & SMILE & 0.4953 & 0.2275 & 1.5863 & 3.6949 \\
          &       &       & SAPLM & 0.0456 & 0.0367 & 0.0235 & 1.1255 \\
          &       &       & SLM   & 6.1453 & 10.0619 & 2.2305 & 20.0080 \\
	500   & 0.5   & 1000  & ORACLE & 0.0103 & 0.0087 & 0.0064 & 0.2762 \\
          &       &       & SMILE & 0.0103 & 0.0087 & 0.0064 & 0.2762 \\
          &       &       & SAPLM & 0.0103 & 0.0087 & 0.0063 & 0.2762 \\
          &       &       & SLM   & 6.0317 & 10.1296 & 1.8271 & 18.3818 \\
          &       & 2000  & ORACLE & 0.0103 & 0.0088 & 0.0060 & 0.2755 \\
          &       &       & SMILE & 0.0103 & 0.0088 & 0.0060 & 0.2755 \\
          &       &       & SAPLM & 0.0103 & 0.0088 & 0.0061 & 0.2756 \\
          &       &       & SLM   & 6.0997 & 10.1132 & 1.8425 & 18.2931 \\
          &       & 5000  & ORACLE & 0.0105 & 0.0092 & 0.0065 & 0.2775 \\
          &       &       & SMILE & 0.0105 & 0.0092 & 0.0064 & 0.2845 \\
          &       &       & SAPLM & 0.0105 & 0.0092 & 0.0065 & 0.2775 \\
          &       &       & SLM   & 6.0894 & 10.1655 & 1.8348 & 18.5435 \\
          & 1.0   & 1000  & ORACLE & 0.0297 & 0.0240 & 0.0179 & 1.0784 \\
          &       &       & SMILE & 0.0296 & 0.0240 & 0.0174 & 1.0778 \\
          &       &       & SAPLM & 0.0296 & 0.0240 & 0.0174 & 1.0778 \\
          &       &       & SLM   & 6.0328 & 10.1296 & 1.8276 & 19.1402 \\
          &       & 2000  & ORACLE & 0.0305 & 0.0245 & 0.0169 & 1.0770 \\
          &       &       & SMILE & 0.0306 & 0.0247 & 0.0396 & 1.0846 \\
          &       &       & SAPLM & 0.0304 & 0.0246 & 0.0164 & 1.0765 \\
          &       &       & SLM   & 6.1016 & 10.1132 & 1.8431 & 19.0421 \\
          &       & 5000  & ORACLE & 0.0297 & 0.0255 & 0.0173 & 1.0833 \\
          &       &       & SMILE & 0.0297 & 0.0255 & 0.0170 & 1.1096 \\
          &       &       & SAPLM & 0.0296 & 0.0255 & 0.0170 & 1.0826 \\
          &       &       & SLM   & 6.0649 & 10.1493 & 1.8395 & 19.3058 \\
	\hline\hline
\end{tabular}%
\end{center}
\end{table}

\begin{table}[htbp]
\caption{LM Case: Selection Statistics comparing the SMILE, SAPLM and SLM.}
\label{TAB:LM1}
\begin{center}
\renewcommand{\arraystretch}{0.9}
\begin{tabular}{ccclrrrr}
	\hline\hline
	Size  & Noise &       &       & \multicolumn{2}{c}{True Selection} & \multicolumn{2}{c}{False Selection} \\
	\cmidrule(r){5-6} \cmidrule(l){7-8}
	$n$ & $\sigma$ & $p$ & Method & corrX & corrX0 & NtoL & Nto0 \\
	\hline
	300   & 0.5   & 1000  & SMILE & 100   & 99.9975 & 0     & 0 \\
          &       &       & SAPLM & 100   & 100   & 100   & 0 \\
          &       &       & SLM   & 100   & 100   & 0     & 0 \\
          &       & 2000  & SMILE & 100   & 99.9998 & 0     & 0 \\
          &       &       & SAPLM & 100   & 100   & 100   & 0 \\
          &       &       & SLM   & 100   & 100   & 0     & 0 \\
          &       & 5000  & SMILE & 100   & 99.9998 & 0     & 0 \\
          &       &       & SAPLM & 100   & 100   & 100   & 0 \\
          &       &       & SLM   & 100   & 100   & 0     & 0 \\
          & 1.0   & 1000  & SMILE & 100   & 99.9890 & 0     & 0 \\
          &       &       & SAPLM & 32.8333 & 100   & 32.8450 & 67.1667 \\
          &       &       & SLM   & 100   & 99.9995 & 0     & 0 \\
          &       & 2000  & SMILE & 100   & 99.9940 & 0     & 0 \\
          &       &       & SAPLM & 17.3333 & 100   & 17.3550 & 82.6667 \\
          &       &       & SLM   & 100   & 99.9998 & 0     & 0 \\
          &       & 5000  & SMILE & 100   & 99.9982 & 0     & 0 \\
          &       &       & SAPLM & 4.5000 & 100   & 4.5050 & 95.5000 \\
          &       &       & SLM   & 100   & 100   & 0     & 0 \\
	500   & 0.5   & 1000  & SMILE & 100   & 100   & 0     & 0 \\
          &       &       & SAPLM & 100   & 100   & 100   & 0 \\
          &       &       & SLM   & 100   & 100   & 0     & 0 \\
          &       & 2000  & SMILE & 100   & 100   & 0     & 0 \\
          &       &       & SAPLM & 100   & 100   & 100   & 0 \\
          &       &       & SLM   & 100   & 100   & 0     & 0 \\
          &       & 5000  & SMILE & 100   & 100   & 0     & 0 \\
          &       &       & SAPLM & 100   & 100   & 100   & 0 \\
          &       &       & SLM   & 100   & 100   & 0     & 0 \\
          & 1.0   & 1000  & SMILE & 100   & 99.9945 & 0     & 0 \\
          &       &       & SAPLM & 100   & 100   & 100   & 0 \\
          &       &       & SLM   & 100   & 100   & 0     & 0 \\
          &       & 2000  & SMILE & 100   & 99.9983 & 0     & 0 \\
          &       &       & SAPLM & 100   & 100   & 100   & 0 \\
          &       &       & SLM   & 100   & 100   & 0     & 0 \\
          &       & 5000  & SMILE & 100   & 99.9993 & 0     & 0 \\
          &       &       & SAPLM & 99.8333 & 100   & 99.8350 & 0.1667 \\
          &       &       & SLM   & 100   & 100   & 0     & 0 \\
	\hline\hline
\end{tabular}
\end{center}
\end{table}

\begin{table}[htbp]
\caption{LM Case: Estimation Statistics comparing the SMILE, SAPLM and SLM.}
\label{TAB:LM2}
\begin{center}
\renewcommand{\arraystretch}{0.725}
\begin{tabular}{ccclrrrr}
	\hline\hline
	Size & Noise &	   &	   & \multicolumn{3}{c}{AMSE $\times 10^{-3}$} & \multicolumn{1}{c}{CV-} \\
	\cmidrule(lr){5-7}
	$n$ & $\sigma$ & $p$ & Method & $\phi_1(\cdot)$ & $\phi_2(\cdot)$ & $\phi_3(\cdot)$ &  MSPE \\
	\hline
	300   & 0.5   & 1000  & ORACLE & 0.7675 & 0.9103 & 0.9152 & 0.2548 \\
          &       &       & SMILE & 0.7737 & 0.9264 & 0.9243 & 0.2551 \\
          &       &       & SAPLM & 3.9953 & 4.1778 & 3.9635 & 0.2636 \\
          &       &       & SLM & 1.4344 & 1.6371 & 1.4934 & 0.2549 \\
          &       & 2000  & ORACLE & 0.8341 & 0.6486 & 0.8838 & 0.2532 \\
          &       &       & SMILE & 0.8335 & 0.6481 & 0.8837 & 0.2534 \\
          &       &       & SAPLM & 3.9852 & 3.9137 & 3.9935 & 0.2619 \\
          &       &       & SLM & 1.4163 & 1.3344 & 1.5203 & 0.2533 \\
          &       & 5000  & ORACLE & 0.8666 & 0.8955 & 0.6963 & 0.2526 \\
          &       &       & SMILE & 0.8685 & 0.9001 & 0.6955 & 0.2529 \\
          &       &       & SAPLM & 4.2351 & 4.1994 & 3.8121 & 0.2609 \\
          &       &       & SLM & 1.6461 & 1.5967 & 1.3410 & 0.2528 \\
          & 1.0   & 1000  & ORACLE & 3.1740 & 3.5473 & 3.5171 & 1.0206 \\
          &       &       & SMILE & 3.1932 & 3.5094 & 3.5415 & 1.0299 \\
          &       &       & SAPLM & 496.6754 & 881.6171 & 234.2496 & 3.2926 \\
          &       &       & SLM & 3.8947 & 4.2478 & 4.1987 & 1.0215 \\
          &       & 2000  & ORACLE & 2.9228 & 2.7934 & 3.5526 & 1.0079 \\
          &       &       & SMILE & 2.9376 & 2.7424 & 3.5271 & 1.0186 \\
          &       &       & SAPLM & 606.3631 & 1071.5918 & 291.8006 & 3.3790 \\
          &       &       & SLM & 3.5838 & 3.4188 & 4.4221 & 1.0090 \\
          &       & 5000  & ORACLE & 3.0392 & 3.4082 & 2.9891 & 1.0067 \\
          &       &       & SMILE & 3.0248 & 3.4140 & 3.0178 & 1.0137 \\
          &       &       & SAPLM & 711.5465 & 1265.1620 & 323.5512 & 3.4189 \\
          &       &       & SLM & 3.8804 & 4.1732 & 4.0380 & 1.0076 \\
	500   & 0.5   & 1000  & ORACLE & 0.5143 & 0.6210 & 0.4891 & 0.2532 \\
          &       &       & SMILE & 0.5143 & 0.6210 & 0.4892 & 0.2532 \\
          &       &       & SAPLM & 2.5860 & 2.7880 & 2.3860 & 0.2582 \\
          &       &       & SLM & 1.2596 & 1.2759 & 1.0586 & 0.2533 \\
          &       & 2000  & ORACLE & 0.4745 & 0.5541 & 0.5393 & 0.2518 \\
          &       &       & SMILE & 0.4745 & 0.5541 & 0.5393 & 0.2518 \\
          &       &       & SAPLM & 2.7168 & 2.5863 & 2.3419 & 0.2569 \\
          &       &       & SLM & 1.1549 & 1.2357 & 1.0850 & 0.2519 \\
          &       & 5000  & ORACLE & 0.5280 & 0.7104 & 0.6061 & 0.2514 \\
          &       &       & SMILE & 0.5280 & 0.7104 & 0.6061 & 0.2514 \\
          &       &       & SAPLM & 2.6505 & 2.8623 & 2.4300 & 0.2563 \\
          &       &       & SLM & 1.2083 & 1.3043 & 1.0216 & 0.2514 \\
          & 1.0   & 1000  & ORACLE & 1.9255 & 2.2321 & 2.0733 & 1.0141 \\
          &       &       & SMILE & 1.9323 & 2.2494 & 2.0626 & 1.0170 \\
          &       &       & SAPLM & 8.2207 & 7.9738 & 8.1074 & 1.0418 \\
          &       &       & SLM & 2.7106 & 2.8192 & 2.7374 & 1.0145 \\
          &       & 2000  & ORACLE & 1.8588 & 2.2051 & 1.9679 & 1.0110 \\
          &       &       & SMILE & 1.8588 & 2.2136 & 1.9783 & 1.0138 \\
          &       &       & SAPLM & 8.6430 & 8.1768 & 7.9037 & 1.0356 \\
          &       &       & SLM & 2.5655 & 2.7929 & 2.5935 & 1.0116 \\
          &       & 5000  & ORACLE & 2.0815 & 2.1319 & 1.9007 & 1.0055 \\
          &       &       & SMILE & 2.0905 & 2.1545 & 1.9073 & 1.0076 \\
          &       &       & SAPLM & 8.6299 & 8.9538 & 9.5719 & 1.0561 \\
          &       &       & SLM & 2.8374 & 2.7330 & 2.4121 & 1.0058 \\
	\hline\hline
\end{tabular}
\end{center}
\end{table}

Next we investigated the coverage rates of the proposed SCB. For each replication, we tested if the true functions can be covered by the SCB at the simulated values of the covariate in $[-0.5+h, 0.5-h]$, where $h$ is the bandwidth. Table \ref{TAB:AM3} shows the empirical coverage probabilities for a nominal 95\% confidence level out of 200 replications. For comparison, we also provided the SCBs from the SAPLM and ORACLE estimator. From Table \ref{TAB:AM3}, one observes that coverage probabilities for the SMILE, SAPLM and ORACLE SCBs all approach the nominal levels as the sample size $n$ increases. In most cases, the SMILE performs as well as or better than the SAPLM and arrives at about the nominal coverage when $n = 500$ and $\sigma = 1.0$.

\begin{table}[htbp]
\caption{AM Case: Coverage rates comparing the SMILE, SAPLM and SLM.}
\label{TAB:AM3}
\renewcommand{\arraystretch}{0.9}
\begin{center}
\begin{tabular}{cccccccccccc}
	\hline\hline
	Size  & Noise &	   & \multicolumn{3}{c}{ORACLE (\%)} & \multicolumn{3}{c}{SMILE (\%)} & \multicolumn{3}{c}{SAPLM (\%)} \\
	\cmidrule(lr){4-6} \cmidrule(lr){7-9}\cmidrule(lr){10-12}
	$n$ & $\sigma$ & $p$ & $\phi_1$ & $\phi_2$ & $\phi_3$ & $\phi_1$ & $\phi_2$ & $\phi_3$ & $\phi_1$ & $\phi_2$ & $\phi_3$  \\
	\hline
	300   & 0.5   & 1000  & 81.0  & 86.0  & 91.0  & 84.5  & 91.5  & 90.5  & 84.5  & 92.0  & 90.5 \\
          &       & 2000  & 80.5  & 84.5  & 88.0  & 85.5  & 92.0  & 92.2  & 83.0  & 92.0  & 91.5 \\
          &       & 5000  & 81.5  & 86.5  & 94.5  & 80.9  & 94.0  & 97.1  & 80.0  & 92.5  & 96.0 \\
          & 1     & 1000  & 92.0  & 94.0  & 96.0  & 91.0  & 95.5  & 95.3  & 92.5  & 95.5  & 95.5 \\
          &       & 2000  & 93.0  & 96.0  & 95.5  & 91.0  & 95.0  & 95.9  & 93.5  & 97.0  & 95.0 \\
          &       & 5000  & 89.0  & 93.5  & 99.0  & 92.3  & 94.4  & 97.8  & 89.0  & 94.5  & 98.0 \\
	500   & 0.5   & 1000  & 84.5  & 92.5  & 89.0  & 85.5  & 96.5  & 93.0  & 83.5  & 96.0  & 93.0 \\
          &       & 2000  & 85.0  & 86.5  & 85.5  & 86.0  & 93.5  & 88.5  & 85.5  & 93.5  & 88.5 \\
          &       & 5000  & 75.0  & 88.5  & 87.0  & 76.5  & 91.5  & 92.5  & 77.0  & 91.5  & 92.5 \\
          & 1     & 1000  & 88.5  & 96.5  & 97.0  & 88.5  & 97.0  & 96.0  & 88.0  & 98.0  & 96.0 \\
          &       & 2000  & 92.0  & 97.5  & 96.0  & 91.5  & 97.5  & 95.5  & 91.5  & 97.5  & 95.5 \\
          &       & 5000  & 91.0  & 90.5  & 94.5  & 94.5  & 94.0  & 95.0  & 94.0  & 94.0  & 95.0 \\
	\hline\hline
\end{tabular}%
\end{center}
\end{table}

\setcounter{chapter}{10}
\renewcommand{\theequation}{C.\arabic{equation}}
\renewcommand{\thesection}{C.\arabic{section}}
\renewcommand{\thesubsection}{C.\arabic{subsection}}
\renewcommand{\thetheorem}{C.\arabic{theorem}}
\renewcommand{\thelemma}{C.\arabic{lemma}}
\renewcommand{\theproposition}{C.\arabic{proposition}}
\renewcommand{\thecorollary}{C.\arabic{corollary}}
\renewcommand{\thefigure}{C.\arabic{figure}}
\renewcommand{\thetable}{C.\arabic{table}}
\setcounter{table}{0}
\setcounter{figure}{0}
\setcounter{equation}{0}
\setcounter{theorem}{0}
\setcounter{lemma}{0}
\setcounter{proposition}{0}
\setcounter{section}{0}
\setcounter{subsection}{0}

\section*{C. A Simulation Study to Explore the Impacts of Covariate Interactions}

Our model considers the APLM, which focuses on variable selection, estimation, model identification and inference for main effects. There might be scenarios where the responses (measurement of SAM tissues, or phenotypes) are affected by interactions between SNP genotypes and RNA sequences.

To explore the robustness of our method in the behavior of selection and model identification for main terms, we conduct a simulation study under an underlying model that includes interaction terms. To be specific, we simulate datasets using the model:
\[
	Y_i = \sum_{k=1}^{p_1}Z_{ik}\alpha_k + \sum_{\ell=1}^{p_2} \phi_{\ell}(X_{i\ell}) + \sum_{m=4}^5 Z_{im}\psi_m(X_{im}) + \varepsilon_i,
\]
where
\begin{align*}
	&\alpha_1 = 3, ~ \alpha_2 = 4, ~ \alpha_3 = -2, ~ \alpha_4 = \ldots = \alpha_{p_1} = 0; \\
	&\phi_1(x) = 9x,
	~~~\phi_2(x) = -1.5 \cos^2(\pi x) + 3 \sin^2(\pi x) - \mathrm{E}\{-1.5 \cos^2(\pi X_2 + 3 \sin^2(\pi X_2)\}, \\
	&\phi_3(x) = 6x + 18x^{2} - \mathrm{E}(6X_3 + 18X_3^{2}), 
	~~~\phi_4(x) = \ldots = \phi_{p_2}(x) = 0; \\
	&\psi_4(x)= 6x, ~~\psi_5(x)=\frac{6}{1+\exp(-20x)}.
\end{align*}

We adopt the following similar criteria used in Section \ref{sec:simulation}:
\begin{enumerate}
	\item[(B-i')] Percent of covariates in $\bs{Z}$ with nonzero linear coefficients (i.e., $Z_1$, $Z_2$ and $Z_3$) that are correctly identified  (``CorrZ'');
	\item[(B-ii')] Percent of covariates in $\bs{Z}$ with zero linear coefficients (all except $Z_1,\ldots,Z_5$) that are correctly identified (``CorrZ0'');
	\item[(B-iii')] Percent of covariates in $\bs{X}$ with nonzero purely linear functions (i.e., $X_1$) that are correctly identified (``CorrL'');
	\item[(B-iv')] Percent of covariates in $\bs{X}$ with nonzero purely nonlinear functions (i.e., $X_2$) that are correctly identified (``CorrN'');
	\item[(B-v')] Percent of covariates in $\bs{X}$ with nonzero linear and nonlinear functions (i.e., $X_3$) that are correctly identified (``CorrLN');
	\item[(B-vi')] Percent of in $\bs{X}$ with zero functions (all except $X_1,\ldots,X_5$) that are correctly identified covariates (``CorrX0'');
	\item[(C-i')] Percent of covariates in $\bs{Z}$ with nonzero linear coefficients (i.e., $Z_1$, $Z_2$ and $Z_3$) incorrectly identified as having zero linear coefficients (``Zto0'');
	\item[(C-ii')] Percent of covariates in $\bs{X}$ with nonzero purely linear functions (i.e., $X_1$) incorrectly identified as having nonlinear functions (``LtoN'');
	\item[(C-iii')] Percent of covariates in $\bs{X}$ with nonzero purely nonlinear functions (i.e., $X_2$) incorrectly identified as having linear functions (``NtoL'');
	\item[(C-iv')] Percent of covariates in $\bs{X}$ with nonzero linear or nonzero nonlinear functions (i.e., $X_1$, $X_2$ and $X_3$) incorrectly identified as having both zero linear and zero nonlinear functions (``Xto0'').
\end{enumerate} \medskip
Note that Criteria (B-i')--(B-vi') measure the frequency of getting the correct model structure; Criteria (C-i')--(C-v') measure the frequency of getting an incorrect model structure. All the above performance measures were computed based on 200 replicates.

The model selection results are provided in Tables \ref{TAB:Interaction1} and \ref{TAB:Interaction2}, respectively. The SMILE can effectively identify informative linear and nonlinear components as well as correctly discover the linear and nonlinear structure in covariates $\bs{X}$, while the SAPLM neglects the linear structure in $\bs{X}$ and SLM fails in presenting the nonlinear part of covariates $\bs{X}$. For the SMILE, the numbers of correctly selected nonzero covariates in $\bs{Z}$, linear, nonlinear, linear-and-nonlinear components in $\bs{X}$, nonzero covariates are very close to ORACLE (100\% for corrZ, corrL, corrN, corrLN, corrZ0 and corrX0, respectively); and the numbers of incorrectly identified components approach 0 as the sample size $n$ increases, as shown in Table \ref{TAB:Interaction2}. From the results in Tables \ref{TAB:Interaction1} and \ref{TAB:Interaction2}, it is evident that our method is robust in the sense that main effects are correctly identified in the presence of interaction effects; in contrast, neither SAPLM nor SLM performs well in this scenario. Especially for the selection of nonlinear and linear-nonlinear covariates in $\bs{X}$, which is our main focus for real data analysis, both SAPLM and SLM fail to select the right nonlinear and linear-nonlinear components in each simulation. Because SMILE, SAPLM, and SLM are based on additive models, none of these approaches are appropriate for detecting interactions.

Table \ref{TAB:Interaction3} reports the percentage of those covariates involved in the interaction ($Z_4$, $Z_5$, $X_4$ and $X_5$) selected out of 200 replications.
As shown in Table \ref{TAB:Interaction3}, SMILE can detect $Z_5$ and $X_5$ in most cases; and the percentages of selection approach 100 as the sample size $n$ increases. In contrast, SAPLM completely fails to select $Z_5$ and $X_5$;  while SLM is only slightly worse than SMILE in the detection of $Z_5$, it has poor performance in the detection of $X_5$. For the interaction terms with smaller main-effect signal, i.e., $Z_4$ and $X_5$, SMILE outperforms in the detection of $X_4$ compared to SAPLM and SLM, and the detection power increases when the sample size $n$ increases. In addition, all three methods fail to detect the relevance of $Z_4$, due to the weak main-effect signal and interaction with $X_4$.

\begin{table}[htbp]
\caption{Statistics of true selection comparing the SMILE, SAPLM and SLM.}
\label{TAB:Interaction1}
\begin{center}
\renewcommand{\arraystretch}{0.92}
\begin{tabular}{ccclrrrrrr}
	\hline \hline
	Size & Noise & & & \multicolumn{2}{c}{Z Part} & \multicolumn{ 4}{c}{X Part} \\
	\cmidrule(r){5-6} \cmidrule(l){7-10}
	$n$ & sig & $p$ & Method & corrZ & corrZ0 & corrL & corrN & corrLN & corrX0 \\
	\hline
	300   & 0.5   & 1000  & SMILE & 97.8333 & 99.9985 & 100 & 82.0000 & 82.0000 & 99.9970 \\
          &       &       & SAPLM & 10.5000 & 100 & 0 & 0 & 0 & 100 \\
          &       &       & SLM   & 90.0000 & 99.9980 & 100 & 0 & 0 & 99.9965 \\
          &       & 2000  & SMILE & 95.6667 & 99.9995 & 100 & 55.5000 & 55.5000 & 99.9992 \\
          &       &       & SAPLM & 9.1667 & 100 & 0 & 0 & 0 & 100 \\
          &       &       & SLM   & 85.5000 & 99.9995 & 100 & 0 & 0 & 99.9990 \\
          &       & 5000  & SMILE & 88.6667 & 99.9994 & 100 & 29.0000 & 28.5000 & 99.9997 \\
          &       &       & SAPLM & 6.6667 & 100 & 0 & 0 & 0 & 100 \\
          &       &       & SLM   & 80.1667 & 99.9998 & 100 & 0 & 0 & 99.9995 \\
          & 1.0   & 1000  & SMILE & 93.8333 & 99.9950 & 100 & 49.0000 & 49.0000 & 99.9990 \\
          &       &       & SAPLM & 7.3333 & 100 & 0 & 0 & 0 & 100 \\
          &       &       & SLM   & 86.0000 & 99.9990 & 99.5000 & 0 & 0 & 99.9975 \\
          &       & 2000  & SMILE & 90.3333 & 99.9967 & 100 & 19.5000 & 19.0000 & 99.9987 \\
          &       &       & SAPLM & 8.5000 & 100 & 0 & 0 & 0 & 100 \\
          &       &       & SLM   & 80.5000 & 99.9990 & 99.0000 & 0 & 0 & 99.9992 \\
          &       & 5000  & SMILE & 86.0000 & 99.9977 & 100 & 7.5000 & 7.0000 & 99.9999 \\
          &       &       & SAPLM & 7.1667 & 100 & 0 & 0 & 0 & 100 \\
          &       &       & SLM   & 74.1667 & 99.9999 & 99.0000 & 0 & 0 & 100 \\
	500   & 0.5   & 1000  & SMILE & 100 & 100 & 100 & 100 & 100 & 99.9995 \\
          &       &       & SAPLM & 64.0000 & 99.9980 & 0 & 56.0000 & 0 & 100 \\
          &       &       & SLM   & 99.6667 & 100 & 100 & 0 & 0 & 99.9985 \\
          &       & 2000  & SMILE & 100 & 100 & 100 & 100 & 100 & 99.9992 \\
          &       &       & SAPLM & 22.1667 & 100 & 0 & 11.0000 & 0 & 100 \\
          &       &       & SLM   & 98.3333 & 100 & 100 & 0 & 0 & 99.9992 \\
          &       & 5000  & SMILE & 100 & 100 & 100 & 100 & 100 & 99.9999 \\
          &       &       & SAPLM & 14.1667 & 100 & 0 & 0 & 0 & 100 \\
          &       &       & SLM   & 98.5000 & 99.9998 & 100 & 0 & 0 & 99.9999 \\
          & 1.0   & 1000  & SMILE & 100 & 100 & 100 & 100 & 100 & 99.9990 \\
          &       &       & SAPLM & 24.1667 & 99.9995 & 0 & 9.5000 & 0 & 100 \\
          &       &       & SLM   & 98.8333 & 100 & 100 & 0 & 0 & 99.9985 \\
          &       & 2000  & SMILE & 100 & 99.9992 & 100 & 100 & 100 & 99.9990 \\
          &       &       & SAPLM & 11.8333 & 100 & 0 & 0.5000 & 0 & 100 \\
          &       &       & SLM   & 97.1667 & 99.9992 & 100 & 0 & 0 & 99.9987 \\
          &       & 5000  & SMILE & 100 & 100 & 100 & 99.0000 & 99.0000 & 100 \\
          &       &       & SAPLM & 13.5000 & 100 & 0 & 0 & 0 & 100 \\
          &       &       & SLM   & 97.1667 & 100 & 100 & 0 & 0 & 99.9998 \\
	\hline \hline
\end{tabular}
\end{center}
\end{table}

\begin{table}[htbp]
\caption{Statistics of false selection comparing the SMILE, SAPLM and SLM.}
\label{TAB:Interaction2}
\begin{center}
\renewcommand{\arraystretch}{0.92}
\begin{tabular}{ccclrrrr}
	\hline \hline
	Size & Noise & & & Z Part & \multicolumn{3}{c}{X Part} \\
	\cmidrule(r){5-5} \cmidrule(l){6-8}
	  $n$ & sig & $p$ & Method & Zto0 & LtoN & NtoL & Xto0 \\
	\hline
	300   & 0.5   & 1000  & SMILE & 2.1667 & 0 & 0 & 6.0000 \\
          &       &       & SAPLM & 89.5000 & 0 & 0 & 100 \\
          &       &       & SLM   & 10.0000 & 0 & 0 & 33.3333 \\
          &       & 2000  & SMILE & 4.3333 & 0 & 0 & 14.8333 \\
          &       &       & SAPLM & 90.8333 & 0 & 0 & 100 \\
          &       &       & SLM   & 14.5000 & 0 & 0 & 33.6667 \\
          &       & 5000  & SMILE & 11.3333 & 0 & 0 & 23.6667 \\
          &       &       & SAPLM & 93.3333 & 0 & 0 & 100 \\
          &       &       & SLM   & 19.8333 & 0 & 0 & 33.8333 \\
          & 1.0   & 1000  & SMILE & 6.1667 & 0 & 0 & 17.0000 \\
          &       &       & SAPLM & 92.6667 & 0 & 0 & 100 \\
          &       &       & SLM   & 14.0000 & 0 & 0 & 33.6667 \\
          &       & 2000  & SMILE & 9.6667 & 0 & 0 & 26.8333 \\
          &       &       & SAPLM & 91.5000 & 0 & 0 & 100 \\
          &       &       & SLM   & 19.5000 & 0 & 0 & 34.8333 \\
          &       & 5000  & SMILE & 14.0000 & 0 & 0 & 30.8333 \\
          &       &       & SAPLM & 92.8333 & 0 & 0 & 100 \\
          &       &       & SLM   & 25.8333 & 0 & 0 & 35.0000 \\
	500   & 0.5   & 1000  & SMILE & 0 & 0 & 0 & 0 \\
          &       &       & SAPLM & 36.0000 & 56.0000 & 0 & 44.0000 \\
          &       &       & SLM   & 0.3333 & 0 & 0 & 33.3333 \\
          &       & 2000  & SMILE & 0 & 0 & 0 & 0 \\
          &       &       & SAPLM & 77.8333 & 11.0000 & 0 & 89.0000 \\
          &       &       & SLM   & 1.6667 & 0 & 0 & 33.3333 \\
          &       & 5000  & SMILE & 0 & 0 & 0 & 0 \\
          &       &       & SAPLM & 85.8333 & 0 & 0 & 100 \\
          &       &       & SLM   & 1.5000 & 0 & 0 & 33.3333 \\
          & 1.0   & 1000  & SMILE & 0 & 0 & 0 & 0 \\
          &       &       & SAPLM & 75.8333 & 9.5000 & 0 & 90.5000 \\
          &       &       & SLM   & 1.1667 & 0 & 0 & 33.3333 \\
          &       & 2000  & SMILE & 0 & 0 & 0 & 0 \\
          &       &       & SAPLM & 88.1667 & 0.5000 & 0 & 99.5000 \\
          &       &       & SLM   & 2.8333 & 0 & 0 & 33.3333 \\
          &       & 5000  & SMILE & 0 & 0 & 0 & 0.3333 \\
          &       &       & SAPLM & 86.5000 & 0 & 0 & 100 \\
          &       &       & SLM   & 2.8333 & 0 & 0 & 33.3333 \\
	\hline \hline
\end{tabular}
\end{center}
\end{table}

\begin{table}[htbp]
\caption{Percents of $Z_4$, $Z_5$, $X_4$ and $X_5$ are selected comparing the SMILE, SAPLM and SLM.}
\label{TAB:Interaction3}
\begin{center}
\renewcommand{\arraystretch}{0.92}
\begin{tabular}{ccclrrrr}
	\hline \hline
	Size  & Noise &       &       & \multicolumn{2}{c}{Z Part} & \multicolumn{2}{c}{X Part} \\
	\cmidrule(r){5-6} \cmidrule(l){7-8}
	  $n$ & sig & $p$ & Method & \multicolumn{1}{c}{Z4} & \multicolumn{1}{c}{Z5} & \multicolumn{1}{c}{X4} & \multicolumn{1}{c}{X5} \\
	\hline
	300   & 0.5   & 1000  & SMILE & 0     & 100   & 33    & 79.5 \\
          &       &       & SAPLM & 0     & 8.5   & 0     & 0 \\
          &       &       & SLM   & 0     & 98.5  & 5.5   & 21 \\
          &       & 2000  & SMILE & 0     & 99    & 15.5  & 52 \\
          &       &       & SAPLM & 0     & 6     & 0     & 0 \\
          &       &       & SLM   & 0     & 98    & 3     & 20 \\
          &       & 5000  & SMILE & 0     & 96.5  & 4.5   & 26.5 \\
          &       &       & SAPLM & 0     & 3     & 0     & 0 \\
          &       &       & SLM   & 0     & 92    & 2     & 13 \\
          & 1.0   & 1000  & SMILE & 0     & 99.5  & 11    & 46 \\
          &       &       & SAPLM & 0     & 7     & 0     & 0 \\
          &       &       & SLM   & 0     & 94.5  & 4     & 15 \\
          &       & 2000  & SMILE & 0     & 98.5  & 4     & 20.5 \\
          &       &       & SAPLM & 0     & 2.5   & 0     & 0 \\
          &       &       & SLM   & 0     & 95.5  & 1.5   & 13 \\
          &       & 5000  & SMILE & 0     & 97    & 1     & 13 \\
          &       &       & SAPLM & 0     & 2.5   & 0     & 0 \\
          &       &       & SLM   & 0     & 86.5  & 1     & 7 \\
	500   & 0.5   & 1000  & SMILE & 0     & 100   & 85    & 100 \\
          &       &       & SAPLM & 0     & 61.5  & 0     & 0 \\
          &       &       & SLM   & 0     & 100   & 21    & 70 \\
          &       & 2000  & SMILE & 0     & 100   & 75.5  & 100 \\
          &       &       & SAPLM & 0     & 21    & 0     & 0 \\
          &       &       & SLM   & 0     & 100   & 12.5  & 60 \\
          &       & 5000  & SMILE & 0     & 100   & 66    & 99 \\
          &       &       & SAPLM & 0     & 12.5  & 0     & 0 \\
          &       &       & SLM   & 0     & 100   & 6.5   & 50 \\
          & 1.0   & 1000  & SMILE & 0     & 100   & 73.5  & 98 \\
          &       &       & SAPLM & 0     & 22    & 0     & 0 \\
          &       &       & SLM   & 0     & 100   & 14.5  & 58 \\
          &       & 2000  & SMILE & 0     & 100   & 55    & 94 \\
          &       &       & SAPLM & 0     & 11    & 0     & 0 \\
          &       &       & SLM   & 0     & 100   & 7     & 49 \\
          &       & 5000  & SMILE & 0     & 100   & 41.5  & 91.5 \\
          &       &       & SAPLM & 0     & 11    & 0     & 0 \\
          &       &       & SLM   & 0     & 99.5  & 4     & 42 \\
	\hline \hline
\end{tabular}
\end{center}
\end{table}

\setcounter{chapter}{11}
\renewcommand{\theequation}{D.\arabic{equation}}
\renewcommand{\thesection}{D.\arabic{section}}
\renewcommand{\thesubsection}{D.\arabic{subsection}}
\renewcommand{\thetheorem}{D.\arabic{theorem}}
\renewcommand{\thelemma}{D.\arabic{lemma}}
\renewcommand{\theproposition}{D.\arabic{proposition}}
\renewcommand{\thecorollary}{D.\arabic{corollary}}
\renewcommand{\thefigure}{D.\arabic{figure}}
\renewcommand{\thetable}{D.\arabic{table}}
\setcounter{table}{0}
\setcounter{figure}{0}
\setcounter{equation}{0}
\setcounter{theorem}{0}
\setcounter{lemma}{0}
\setcounter{proposition}{0}
\setcounter{section}{0}
\setcounter{subsection}{0}

\section*{D. A Simulation Study Based on the SAM Data}

In this section, we conduct a simulation study using the SNPs and RNA transcripts selected in real data analysis as the active covariates in our data-generating model. This  demonstrates the performance of our method when there are many true nonzero components in both linear and nonlinear parts.

With the true linear coefficients and nonlinear functions set to be the same as the estimates obtained in real data analysis in Section 6, we choose the noise level $\sigma$ as 0.01, 0.02, 0.04 and 0.05, in accordance with the errors in real data analysis ($\widehat{\sigma}\approx 0.04$). We compare SMILE with SAPLM and SLM. We still summarize the simulation results by using the statistics described in Section \ref{sec:simulation}.  All the performance measures were computed based on 200 replicates.

The model selection results are provided in Tables \ref{SAM:SIM:1} and \ref{SAM:SIM:2}, respectively. The SMILE can effectively identify informative linear and nonlinear components as well as correctly discover the linear and nonlinear structure in covariate $\bs{X}$, while the SAPLM neglects linear structure in $\bs{X}$ and SLM fails in presenting the nonlinear part of covariate $\bs{X}$. For the SMILE, the selection performance in $\bs{Z}$, including the numbers of correctly selected nonzero and zero covariates in $\bs{Z}$ (corrZ and corrZ0) and the numbers of incorrectly identified components in  $\bs{Z}$ (Zto0), is very close to the performance of SLM; while SMILE outperforms SLM in the selection of important components of $\bs{X}$, indicated by the much higher percents of correctly selected components in $\bs{X}$ (corrL, corrN, and corrX0) and much lower percents in the incorrectly selected components in $\bs{X}$ (Xto0). The SMILE outperforms SAPLM in almost all statistics. From the results in Tables \ref{SAM:SIM:1} and \ref{SAM:SIM:2}, it is evident that model misspecification leads to poor variable selection performance for the SAPLM and SLM.

\begin{table}[htbp]
\caption{True selection statistics comparing the SMILE, SAPLM and SLM.}
\label{SAM:SIM:1}
\begin{center}
\begin{tabular}{clrrrrr}
	\hline\hline
	Noise &       & \multicolumn{2}{c}{Z Part} & \multicolumn{3}{c}{X Part} \\
	\cmidrule(r){3-4}\cmidrule(l){5-7}
	$\sigma$   & Method & \multicolumn{1}{l}{corrZ} & \multicolumn{1}{l}{corrZ0} & \multicolumn{1}{l}{corrL} & \multicolumn{1}{l}{corrN} & \multicolumn{1}{l}{corrX0} \\
	\hline
	0.01  & SMILE & 71.62 & 99.81 & 81.9  & 99.5  & 98.61 \\
          & SAPLM & 38.57 & 99.66 & 0     & 0     & 99.43 \\
          & SLM   & 78.83 & 99.72 & 25.8  & 0     & 99.59 \\
	0.02  & SMILE & 69.31 & 99.80 & 75.55 & 96.75 & 98.50 \\
          & SAPLM & 38.36 & 99.62 & 0     & 0     & 99.37 \\
          & SLM   & 77.60 & 99.69 & 26.3  & 0     & 99.37 \\
	0.04  & SMILE & 65.80 & 99.74 & 65.45 & 87.5  & 98.07 \\
          & SAPLM & 36.95 & 99.56 & 0     & 0     & 99.25 \\
          & SLM   & 74.19 & 99.57 & 23.35 & 0     & 98.89 \\
	0.05  & SMILE & 64.35 & 99.72 & 57.75 & 85.5  & 97.79 \\
          & SAPLM & 36.06 & 99.53 & 0     & 0     & 99.20 \\
          & SLM   & 72.49 & 99.49 & 22.35 & 0     & 98.65 \\
	\hline\hline
\end{tabular}
\end{center}
\end{table}

\begin{table}[htbp]
\caption{False selection statistics comparing the SMILE, SAPLM and SLM.}
\label{SAM:SIM:2}
\begin{center}
\begin{tabular}{clrrrr}
	\hline\hline
	Noise &       & \multicolumn{1}{c}{Z Part} & \multicolumn{3}{c}{X Part} \\
	\cmidrule(l){4-6}
	$\sigma$   & Method & Zto0  & LtoN  & NtoL  & Xto0 \\
	\hline
	0.01  & SMILE & 28.38 & 3.6   & 0     & 0.08 \\
          & SAPLM & 61.43 & 25.8  & 0     & 78.50 \\
          & SLM   & 21.17 & 0     & 0     & 16.67 \\
	0.02  & SMILE & 30.69 & 7.45  & 0     & 0.71 \\
          & SAPLM & 61.64 & 26.3  & 0     & 78.08 \\
          & SLM   & 22.40 & 0     & 0     & 16.67 \\
	0.04  & SMILE & 34.20 & 13.95 & 0     & 3.25 \\
          & SAPLM & 63.05 & 23.35 & 0     & 80.54 \\
          & SLM   & 25.81 & 0     & 0     & 16.79 \\
	0.05  & SMILE & 35.65 & 17.9  & 0     & 4.58 \\
          & SAPLM & 63.94 & 22.35 & 0     & 81.37 \\
          & SLM   & 27.51 & 0     & 0     & 17.38 \\
	\hline\hline
\end{tabular}
\end{center}
\end{table}

\setcounter{chapter}{12}
\renewcommand{\theequation}{E.\arabic{equation}}
\renewcommand{\thesection}{E.\arabic{section}}
\renewcommand{\thesubsection}{E.\arabic{subsection}}
\renewcommand{\thetheorem}{E.\arabic{theorem}}
\renewcommand{\thelemma}{E.\arabic{lemma}}
\renewcommand{\theproposition}{E.\arabic{proposition}}
\renewcommand{\thecorollary}{E.\arabic{corollary}}
\setcounter{equation}{0}
\setcounter{theorem}{0}
\setcounter{lemma}{0}
\setcounter{proposition}{0}
\setcounter{section}{0}
\setcounter{subsection}{0}

\section*{E. Technical Details}
This section contains some technical assumptions, lemmas and proofs. For any real numbers $a$ and $b$, let $a \vee b$ and $a \wedge b$ denote the maximum and minimum of $a$ and $b$, respectively. For any two sequences $\{a_n\}$, $\{b_n\}$, $n = 1, 2, \ldots$, we use $a_n \asymp b_n$ if there are constants $0 < c_1 < c_2 < \infty$ such that $c_1 < a_n / b_n < c_2$ for all $n$ sufficiently large. On any fixed interval $[a,b]$, we denote the space of the second order smooth functions as $C^{(d)}[a,b]=\left\{f\left|f^{(d)}\in C[a,b]\right\}\right.$ and the class of Lipschitz continuous functions for any fixed constant $C>0$ as Lip $([a,b],C)=\{f| \left|f(x)-f(x^\prime)\right|\leq C\left|x-x^\prime\right|,\forall x, x^\prime \in [a,b]\}$.

Furthermore, let $\mathbf{Y} = (Y_1, \ldots, Y_n)^{\top}$ be an $n$-dimensional vector, $\mathbf{Z} =(\mathbf{Z}_1, \ldots, \mathbf{Z}_{p_1})$ be an $n \times p_1$ matrix, where $\mathbf{Z}_k = (Z_{1k}, \ldots, Z_{nk})^{\top}$, $k = 1, \ldots, p_1$, and $\mathbf{X} =(\mathbf{X}_1, \ldots, \mathbf{X}_{p_2})$ be an $n \times p_2$ matrix, where $\mathbf{X}_{\ell} = (X_{1\ell}, \ldots, X_{n\ell})^{\top}$, $\ell=1, \ldots, p_2$.
Let $\mathbf{B}^{(d)} = (\mathbf{B}_1^{(d)}, \ldots, \mathbf{B}_{p_2}^{(d)})$ be a dimension $n \times (p_2 M_n)$ matrix, where $\mathbf{B}_{\ell}^{(d)} = (\mathbf{B}_{\ell}^{(d)}(X_{1\ell}), \ldots, \mathbf{B}_{\ell}^{(d)}(X_{n\ell}))^{\top}$ is a dimension $n \times M_n$ matrix of spline basis functions of order $d$, for $\ell= 1,\ldots, p_2$. Let $\mathcal{A} \subseteq \{1, \ldots, p_1 + 2p_2\}$ be an index set, and let $|\mathcal{A}|$ denote the cardinality of set $\mathcal{A}$.

\subsection{Technical Assumptions} \label{SUBSEC:assump}

In addition to the sparsity condition (A1) stated in Section \ref{sec:method}, we need the following additional regularity conditions to establish the theoretical results in this paper.

\begin{enumerate}

\item[(A2)] (\textit{Conditions on errors}) The errors $\varepsilon_1, \ldots, \varepsilon_n$ are independent and identically distributed with $\mathrm{E}(\varepsilon_i )= 0$, $\mathrm{Var}(\varepsilon_i) = \sigma^2$, $\mathrm{E} |\varepsilon_i|^{2 + \delta} \leq M_{\delta}$ for some positive constant $M_{\delta}$  ($\delta > 0.5$), and have $b$-sub-gaussian tails, i.e., $\mathrm{E} \{\exp({t\varepsilon})\} \leq \exp({b^2t^2 / 2})$, for any $t \geq 0$ and some $b > 0$.

\item[(A3)] (\textit{Conditions on nonlinear functions}) The additive component function $g_{\ell}(\cdot) \in C^{(2)}[a,$ $b]$, $\ell=1, \ldots, p_2$.

\item[(A4)] (\textit{Conditions on covariates}) Each covariate in the parametric part of the model
is bounded,
that is, there is a positive constant $C_3$ such that $
|Z_{k}|
\leq C_3, \, 1 \leq k \leq p_1$;
also, $\mathrm{E}(X_{\ell})=0$, and there is a positive constant $C_4$ such that $
|X_{\ell}|
\leq C_4, \, 1 \leq \ell \leq p_2$. The joint density function of active pure linear $\bs{X}$ is continuous and bounded below and above.
Each covariate in the nonparametric part of the model has a continuous density and there exist constants $C_1$ and $C_2$ such that the marginal density function $f_{\ell}$ of $X_{\ell}$ has continuous derivatives on its support, and satisfies $0 < C_1 \leq f_{\ell}(x_{\ell}) \leq C_2 < \infty$ on its support for every $1 \leq \ell \leq p_{2}$. In addition, the eigenvalues of $\mathrm{E} \{(\bs{Z} \bs{Z}^{\top}) | \bs{X}\}$ are bounded away from 0.


    \item[(A5)] (\textit{Conditions on the initial estimators}) The initial estimators satisfy $r_{n1}\max\limits_{k \in \mathcal{N}_z} |\widetilde{\alpha}_k| = O_P(1)$, $r_{n2}\max\limits_{\ell \in \mathcal{N}_x} |\widetilde{\beta}_{\ell}| = O_P(1)$, $r_{n3}\max\limits_{\ell \in \mathcal{N}_x} \Vert\tbgamma_{\ell} \Vert_2 = O_P(1)$, $r_{n1}$, $r_{n2}$, $r_{n3} \to \infty$, and there exist positive constants $c_{b1}$, $c_{b2}$ and $c_{b3}$ such that $\Pr \left(\min\limits_{k \in \mathcal{S}_z} |\widetilde{\alpha}_{k}| \geq c_{b1}b_{n1}\right) \to 1$, $\Pr \left(\min\limits_{\ell \in \mathcal{S}_{x, L}} |\widetilde{\beta}_{\ell}| \geq c_{b2}b_{n2}\right) \to 1$, $\Pr \left(\min\limits_{\ell \in \mathcal{S}_{x, N}} \Vert \tbgamma_{\ell}\Vert_2 \geq c_{b3}b_{n3}\right) \to 1$, where $b_{n1} = \min_{k \in \mathcal{S}_z} |\alpha_{0k}|$, $b_{n2} = \min_{\ell \in \mathcal{S}_{x, L}} |\beta_{0\ell}|$, and $b_{n3} = \min_{\ell \in \mathcal{S}_{x, N}} \Vert g_{0\ell} \Vert_2$.

    \item[(A6)] (\textit{Conditions on parameters and spline basis functions}) Let $p_1$ and $p_2$ be the number of linear and nonlinear components, respectively. Suppose that $n^{-1}N_n + n^{-2}\sum_{j=1}^{3}\lambda_{nj}^2 =  o(1)$, and
         \[
            \frac{\sqrt{n\ln(p_1)}}{\lambda_{n1}r_{n1}} + \frac{\sqrt{n\ln(p_2)}}{\lambda_{n2}r_{n2}} + \frac{\sqrt{n N_n\ln(p_2 N_n)}}{\lambda_{n3}r_{n3}} + \sum_{j=1}^{3}\frac{n}{\lambda_{nj}r_{nj}N_n}=  o(1).
       \]

\end{enumerate}

Assumptions (A1)--(A4) are regularity conditions that are commonly used in the APLM literature. To obtain the selection consistency of the SBLL-AGLASSO, we need an order requirement for a general initial estimator; see Assumption (A5).
Theorem \ref{THM:LASSO-1} below demonstrates that the group LASSO estimator defined in (\ref{EQ:loss_glasso_spline}) satisfies Assumption (A5) under some weak conditions, specifically if $\sum_{j=1}^{3}\widetilde{\lambda}_{nj}^2 \asymp {n \{\ln(p_1) \vee N_n \ln(p_2N_n)\}}$ and
$N_n \asymp n^{1/3}$,
then the consistent rates for the group LASSO estimator in (A5) have order {$r_{n1} \asymp r_{n2} \asymp r_{n3} = O\{n^{1/2} / \sqrt{\ln (p_1) \vee N_n \ln(p_2 N_n)} \}$}. Consequently, Assumption (A6) is equivalent to:
\begin{equation}
    \frac{\sum_{j=1}^{3}\lambda_{nj}^2}{n^2}
    + \frac{\ln (p_1) \vee N_n \ln(p_2 N_n)}{(\lambda_{n1} \wedge \lambda_{n2} \wedge \lambda_{n3})}
    + \frac{n^{1/6}\sqrt{\ln (p_1) \vee N_n \ln(p_2 N_n)}}{(\lambda_{n1} \wedge \lambda_{n2} \wedge \lambda_{n3})} = o(1), \label{EQN:order}
\end{equation}
If we take $\lambda_{n1} \asymp \lambda_{n2} \asymp \lambda_{n3} = O(n^{1/2})$, then (\ref{EQN:order}) indicates $p_1 = \exp\{o(n^{1/2})\}$ and $p_2 = \exp\{o(n^{1/6})\}$.


We need the following additional assumptions in order to develop the asymptotic SCBs for the nonparametric components.

\begin{enumerate}
{\renewcommand{\itemsep}{-1pt}
%
    \item[(A3$^{\prime}$)] (\textit{Conditions on nonlinear functions}) For any $\ell \in \mathcal{S}_{x,N}$, $\phi_{0\ell} \in {C}^{(d)}[a,b]$, for some integer $d\geq 2$. In addition, $\psi_{\ell}^x$ defined in (\ref{DEF:psi}) satisfies $\psi_{\ell}^x \in C^{(d)}[a,b]$.

    \item[(A6$^{\prime}$)](\textit{Conditions on spline basis functions})  The order of the spline basis functions is at least $d$, and the number of interior knots $M_n$ satisfies: $\left\{ n^{1/\left( 2d\right)} \vee n^{4/(10d-5)} \right\} \ll M_{n}\ll n^{1/3}$.

    \item[(B1)] (\textit{Conditions on the kernel function}) The kernel function $K \in \mbox{Lip\ }([-1,1],C_{K})$ for some constant $C_{K}>0$, and is bounded, nonnegative, symmetric, and supported on $\left[
-1,1\right] $ with the second moment $\mu_{2}(K)=\int u^{2}K\left( u\right) du$.

    \item[(B2)] (\textit{Conditions on bandwidth}) For each $\ell\in \mathcal{S}_{x,N}$, the bandwidth of the kernel $K$ is $h_{\ell}^{-1}=O(n^{1/5}\ln^{\delta}n)$ for some constant $\delta>1/5$.}
\end{enumerate}

Assumptions (A3$^{\prime}$), (B1) and (B2) are typical in the local polynomial smoothing literature; see, for instance, \cite{zheng2016statistical}. Assumption (A6$^{\prime}$) imposes the condition of the number of knots for spline smoothing. For example, if $d=2$, we can
take $M_{n}\sim n^{4/15}\ln n$.

\subsection{Selection and estimation properties of the group LASSO estimators}

In this section, we consider the selection and estimation properties of the group LASSO estimator $\bs{\tbtheta} = (\widetilde{\bs{\alpha}}^{\top}, \widetilde{\bs{\beta}}^{\top}, \widetilde{\bs{\gamma}}^{\top})^{\top}$ in (\ref{EQ:loss_glasso_spline}). In the following, denote $\bs{\alpha} = (\alpha_1, \ldots, \alpha_{p_1})^{\top}$ with length $p_1$, $\bs{\beta} = (\beta_1, \ldots, \beta_{p_2})^{\top}$ with length $p_2$, and $\bs{\gamma} = (\bs{\gamma}_1^{\top}, \ldots, \bs{\gamma}_{p_2}^{\top})^{\top} $ with length ($p_2 N_n$).
Let
\[
\bs{\theta}^{\top}\!\! = \!(\bs{\alpha}^{\top}\!, \bs{\beta}^{\top}\!, \bs{\gamma}^{\top}) \!= \!\left(\alpha_1, \ldots, \alpha_{p_1}, \beta_1, \ldots, \beta_{p_2}, \bs{\gamma}_1^{\top}\!, \ldots, \bs{\gamma}_{p_2}^{\top} \right)\!=\! \left(\bs{\theta}_1^{\top}\!, \ldots, \bs{\theta}_m^{\top}, \ldots, \bs{\theta}_{p_1 + 2p_2}^{\top}\right),
\]
where $\bs{\theta}_m = \alpha_m I\{1 \leq m \leq p_1\} + \beta_{m-p_1} I\{p_1+1 \leq m \leq p_1+p_2\} + \bs{\gamma}_{m - p_1-p_2} I\{p_1+p_2+1 \leq m \leq p_1 + 2p_2 \}$, with $I(\cdot)$ being an indicator function.
Let
\[
\mathbf{D} = (\mathbf{Z}_1, \ldots, \mathbf{Z}_{p_1}, \mathbf{X}_1, \ldots, \mathbf{X}_{p_2}, \mathbf{B}_1^{(1)},\ldots, \mathbf{B}_{p_2}^{(1)}) \equiv \left(\mathbf{D}_1, \ldots, \mathbf{D}_m, \ldots, \mathbf{D}_{p_1 + 2p_2}\right)
\]
be an $n \times (p_1 + p_2  + p_2N_n)$ matrix, where
\[
\mathbf{D}_m = \mathbf{Z}_m I\{1 \leq m \leq p_1\} + \mathbf{X}_m I\{1 \leq m \leq p_1\} + \mathbf{B}_{m - p_1-p_2}^{(1)} I\{p_1+p_2 +1 \leq m \leq p_1 + 2p_2\},
\]
an $n \times d_m$ submatrix of $\mathbf{D}$ with $d_m= I(1 \leq m \leq p_1) + I(p_1 + 1 \leq m \leq p_1+p_2) + N_n I(p_1+p_2 + 1 \leq m \leq p_1 + 2p_2)$. Define
\begin{equation}
\widetilde{\mathcal{S}} = \{m: \Vert \bs{\tbtheta}_m\Vert \neq 0, 1 \leq m \leq p_1 + 2p_2\}.
\label{DEF:S_tilde}
\end{equation}

Next we define the active linear index set for $\bs{X}$ as $\mathcal{S}_{x,L}=\mathcal{S}_{x, PL}\cup\mathcal{S}_{x, LN}$, the inactive linear index set for $\bs{X}$ as $\mathcal{N}_{x, L}$, and the inactive nonlinear index set for $\bs{X}$ as $\mathcal{N}_{x, N}$. Note that $\mathcal{N}_{x}=\mathcal{N}_{x, L}\cap\mathcal{N}_{x, N}$. Further, let
\begin{align}
    &\mathcal{S} =\mathcal{S}_z \, \cup \, \{\ell + p_1: \ell \in\mathcal{S}_{x,L}\} \, \cup \, \{\ell + p_1 + p_2: \ell \in\mathcal{S}_{x,N}\}, \nonumber \\
    &\mathcal{N} =\mathcal{N}_z \, \cup \, \{\ell + p_1: \ell \in\mathcal{N}_{x, L}\} \, \cup \, \{\ell + p_1 + p_2: \ell \in\mathcal{N}_{x, N}\}.
\label{DEF:SN}
\end{align}

For any index set $\mathcal{A} \subseteq \{1, \ldots, p_1 + 2p_2\}$, define $\mathbf{D}_{\mathcal{A}} = \{\mathbf{D}_{m}: m \in \mathcal{A}\}$.
Next denote $\mathbf{C}_\mathcal{A} = n^{-1}\mathbf{D}_\mathcal{A}^{\top}\mathbf{D}_\mathcal{A}$, and let $\pi_{\min}(\mathbf{C}_\mathcal{A})$ and $\pi_{\max}(\mathbf{C}_\mathcal{A})$ represent the minimum and maximum eigenvalues of $\mathbf{C}_\mathcal{A}$, respectively.

\begin{lemma}\label{LEM:eigen} 
Let $N_n = O(n^{\gamma})$, where $0 < \gamma < 0.5$. Suppose that $|\mathcal{A}|$ is bounded by a fixed constant independent of n, $p_{1}$ and $p_{2}$. Then under Assumption (A4), with probability approaching one as $n \to \infty$, $c_1 \leq \pi_{\min}(\mathbf{C}_\mathcal{A}) \leq \pi_{\max}(\mathbf{C}_\mathcal{A}) \leq c_2$, where $c_1$ and $c_2$ are two positive constants.
\end{lemma}
\begin{proof}
Similar to the proof of Lemma A.1 in \cite{li2017ultra}.
\end{proof}

\begin{lemma}
\label{LEM:spline-approx} 
Under Assumption (A3), there exists a vector $\bs{\gamma}_0=(\bs{\gamma}_{0 1}^{\top}, \ldots, \bs{\gamma}_{0 p_2}^{\top})^{\top}$, such that
$\|\bs{\gamma}_{0\ell}\|\neq 0$, for $\ell \in \mathcal{S}_{x, N}$, $\|\bs{\gamma}_{0\ell}\|=0$, $\ell \in \mathcal{N}_{x,N}$
and $\Vert g_{0\ell}-\mathbf{B}_{\ell}^{(d)\top} \bs{\gamma}_{0\ell}\Vert_2=O(M_n^{-d})$.
\end{lemma}
\begin{proof}
Similar to the proof of Lemma A.2 in \cite{li2017ultra}.
\end{proof}


In the following, we denote $g_{n\ell}(\cdot) = \sum_{J = 1}^{N_n} \gamma_{0\ell J} {B}_{J,\ell}^{(1)}\left(\cdot\right)$ the best constant spline approximation of $g_{0\ell}(\cdot)$ such that $\Vert g_{0\ell} - g_{n\ell} \Vert_{\infty} = \sup_{x \in [a, b]} |g_{0\ell}(x) - g_{n\ell}(x)| = O(N_n^{-1})$. Let $\bs{\gamma}_{0\ell}=(\gamma_{0\ell J}, J=1,\ldots,N_n)^{\top}$ be the vector of the coefficients of the best spline approximation in Lemma \ref{LEM:spline-approx}. Denote $\bs{\theta}_{0}^{\top} = (\bs{\theta}_{01}^{\top}, \ldots, \bs{\theta}_{0m}^{\top}, \ldots, \bs{\theta}_{0,p_1 + 2p_2}^{\top}) = (\bs{\alpha}_{0}^{\top}, \bs{\beta}_{0}^{\top}, \bs{\gamma}_{0}^{\top})
= (\alpha_{01}, \ldots, \alpha_{0p_1},$ $\beta_{01},\ldots, \beta_{0p_2}, \bs{\gamma}_{01}^{\top}, \ldots, \bs{\gamma}_{0p_2}^{\top} )$. Define $\bs{\theta}_{\mathcal{A}} = (\bs{\theta}_{m}^{\top}: m \in \mathcal{A})^{\top}$, $\bs{\theta}_{0,\mathcal{A}} = (\bs{\theta}_{0m}^{\top}: m \in \mathcal{A})^{\top}$ and $\widetilde{\bs{\theta}}_{\mathcal{A}} = (\widetilde{\bs{\theta}}_m^{\top}: m \in \mathcal{A})^{\top}$.



\begin{theorem}
\label{THM:LASSO-1} 
Suppose that Assumptions (A1)--(A4) hold.
  \begin{enumerate} [{\normalfont (i)}]
    \item
    If $\{\ln(p_1) \vee N_n \ln(p_2N_n)\} / {n} \rightarrow 0$ and $n^{-2}\sum_{j=1}^{3} \widetilde{\lambda}_{nj}^2 \rightarrow 0$ as $n \to \infty$, then with probability converging to one, all the nonzero linear parameters $\alpha_{0k}$ and $\beta_{0\ell}$, $k \in \mathcal{S}_z, \ell \in \mathcal{S}_{x,L}$, and nonzero additive components $g_{0\ell}$, $\ell \in \mathcal{S}_{x, N}$, are selected.
    \item
    In addition,
    \begin{align*}
            \sum_{k= 1}^{p_1} |\widetilde{\alpha}_k - {\alpha}_{0k}|_2^2 &= O_P\left\{\frac{\ln(p_1) \vee N_n \ln(p_2N_n)}{n}\right\} + O\left(N_n^{-2}\right) + O\left(n^{-2}\sum_{j=1}^{3} \widetilde{\lambda}_{nj}^2\right), \\
            \sum_{\ell = 1}^{p_2}| \widetilde{\beta}_{\ell} - {\beta}_{0\ell} |^2 &= O_P\left\{\frac{\ln(p_1) \vee N_n \ln(p_2N_n)}{n}\right\} + O\left(N_n^{-2}\right) + O\left(n^{-2}\sum_{j=1}^{3} \widetilde{\lambda}_{nj}^2\right), \\
           \sum_{\ell = 1}^{p_2} \Vert \widetilde{g}_{\ell} - g_{0\ell} \Vert_2^2 &= O_P\left\{\frac{\ln(p_1) \vee N_n \ln(p_2N_n)}{n}\right\} + O\left(N_n^{-2}\right) + O\left(n^{-2}\sum_{j=1}^{3} \widetilde{\lambda}_{nj}^2\right).
        \end{align*}
  \end{enumerate}
\end{theorem}

\begin{proof}
We prove part (ii) first. Let $\tbtheta^{\top} \!\!\!\equiv ({\tbtheta}_1^{\top}, \ldots, \tbtheta_{p_1 + 2p_2}^{\top}) \!=\! (\widetilde{\alpha}_1, \!\ldots, \!\widetilde{\alpha}_{p_1}, \widetilde{\beta}_1, \ldots, \widetilde{\beta}_{p_2},$ $ \tbgamma_{1}^{\top} \ldots, \tbgamma_{p_2}^{\top})$. For $\mathcal{S}$ defined in (\ref{DEF:SN}) and $\widetilde{\mathcal{S}}$ defined in (\ref{DEF:S_tilde}), denote $\mathcal{S}^{\prime} = \mathcal{S} \bigcup \widetilde{\mathcal{S}}  = \{m: \Vert \bs{{\theta}}_{0m}\Vert_2 \neq 0 \, \text{or} \, \Vert \bs{\tbtheta}_m\Vert_2 \neq 0 \}$ and $d^{\prime} = |\mathcal{S}^{\prime}|$. By Lemma \ref{LEM:gLASSO_select},
$d^{\prime} = O(|\mathcal{S}|)$. Notice that $\mathbf{D}\widetilde{\bs{\theta}} = \mathbf{D}_{\mathcal{S}^{\prime}}\sigsy{\tbtheta}$ and $\mathbf{D}\bs{\theta}_0 = \mathbf{D}_{\mathcal{S}^{\prime}}\bs{\theta}_{0,\mathcal{S}^{\prime}}$, by the definition of $\tbtheta$ and $\mathcal{S}^{\prime}$,
\begin{align*}
\Vert \mathbf{Y} - \mathbf{D}_{\mathcal{S}^{\prime}}\sigsy{\tbtheta} \Vert^2
& - \Vert \mathbf{Y} - \mathbf{D}_{\mathcal{S}^{\prime}}\bs{\theta}_{0,\mathcal{S}^{\prime}} \Vert^2\\
& \leq  \sum_{m\in \mathcal{S}^{\prime}}\{\widetilde{\lambda}_{n1} I(m \leq p_1) + \widetilde{\lambda}_{n2} I(p_1 < m \leq p_1 + p_2) + \widetilde{\lambda}_{n3} I(m > p_1 + p_2)\}  \Vert \bs{\theta}_{0m} \Vert \\
& \quad - \sum_{m\in \mathcal{S}^{\prime}}\{\widetilde{\lambda}_{n1} I(m \leq p_1) + \widetilde{\lambda}_{n2} I(p_1 < m \leq p_1 + p_2) + \widetilde{\lambda}_{n3} I(m > p_1 + p_2)\} \Vert \tbtheta_m \Vert.
\end{align*}
Let $\bs{\eta} = \mathbf{Y} - \mathbf{D}\bs{\theta}_0$
and $\bs{\nu} = \mathbf{D}_{\mathcal{S}^{\prime}}(\sigsy{\tbtheta} - {\bs{\theta}_{0, \mathcal{S}^{\prime}}})$, so $\bs{\eta} - \bs{\nu} = \mathbf{Y} - \mathbf{D}_{\mathcal{S}^{\prime}}\sigsy{\tbtheta}$,
and we have $\Vert \mathbf{Y} - \mathbf{D}_{\mathcal{S}^{\prime}}\sigsy{\tbtheta} \Vert^2 - \Vert \mathbf{Y} - \mathbf{D}_{\mathcal{S}^{\prime}}\bs{\theta}_{0,\mathcal{S}^{\prime}} \Vert^2  = \bs{\nu}^{\top}\bs{\nu} - 2\bs{\eta}^{\top}\bs{\nu}$. Thus, from the triangle inequality and the Cauchy-Schwartz inequality,
\begin{align}
\Vert \bs{\nu} \Vert^2 \!-\! 2 \bs{\eta}^{\top}\bs{\nu}
	&\!\leq\! \sum_{m\in \mathcal{S}^{\prime}}\{\widetilde{\lambda}_{n1} I(m \!\leq\! p_1) \!+\! \widetilde{\lambda}_{n2} I(p_1 \!<\! m \leq p_1 + p_2) \!+\! \widetilde{\lambda}_{n3} I(m \!>\! p_1 + p_2)\} (\Vert \bs{\theta}_{0m} \Vert \!-\! \Vert \tbtheta_m \Vert)  \nonumber  \\
    &\leq \sqrt{d^{\prime}\sum_{j=1}^3 \widetilde\lambda_{nj}^2} \Vert \sigsy{\tbtheta} - \bs{\theta}_{0,\mathcal{S}^{\prime}}  \Vert
    \leq \frac{ d^{\prime} \sum_{j=1}^3 \widetilde\lambda_{nj}^2}{nc_{*}} + \frac{1}{4} nc_{*}\Vert \sigsy{\tbtheta} - \bs{\theta}_{0,\mathcal{S}^{\prime}}  \Vert^2,
\label{EQ:nu2_etanu}
\end{align}
where $c_{*}$ is the lower bound of eigenvalues of $n^{-1}\mathbf{D}_{\mathcal{S}^{\prime}}^{\top}\mathbf{D}_{\mathcal{S}^{\prime}}$. By Lemma\ref{LEM:eigen} and Lemma \ref{LEM:gLASSO_select}, $c_{*}\asymp 1$ with probability approaching one. Apparently,
\begin{equation}
\label{EQ:vu2_geq}
    \Vert \bs{\nu} \Vert^2 \geq nc_{*} \Vert \sigsy{\tbtheta} - \bs{\theta}_{0,\mathcal{S}^{\prime}}  \Vert^2.
\end{equation}
Define $\bs{\eta}^{\ast} \equiv
\mathbf{D}_{\mathcal{S}^{\prime}}(\mathbf{D}_{\mathcal{S}^{\prime}}^{\top}\mathbf{D}_{\mathcal{S}^{\prime}})^{-1}\mathbf{D}_{\mathcal{S}^{\prime}}^{\top}\bs{\eta}$ to be the projection of $\bs{\eta}$ onto the column space of $\mathbf{D}_{\mathcal{S}^{\prime}}$. Obviously, $\bs{\eta}^{\top}\bs{\nu} = \bs{\eta}^{\ast\top}\bs{\nu}$. By the Cauchy-Schwartz inequality, we have
\begin{equation}
\label{EQ:eta_nu}
    2|\bs{\eta}^{\top}\bs{\nu}| \leq 2 \Vert \bs{\eta}^{\ast} \Vert \Vert \bs{\nu} \Vert \leq 2 \Vert \bs{\eta}^{\ast} \Vert^2 + \frac{1}{2} \Vert \bs{\nu} \Vert^2.
\end{equation}
Combining (\ref{EQ:nu2_etanu}), (\ref{EQ:vu2_geq}) and (\ref{EQ:eta_nu}), we obtain
    \begin{eqnarray}
        \Vert \sigsy{\tbtheta} - \bs{\theta}_{0,\mathcal{S}^{\prime}}  \Vert^2 &\leq& \frac{8\Vert \bs{\eta}^{\ast} \Vert^2}{nc_{*}} + \frac{4 d^{\prime}\sum_{j=1}^3 \widetilde\lambda_{nj}^2}{n^2c_{*}^2} .
        \label{EQN:the_diffs1}
    \end{eqnarray}
With $\eta_i$ defined to be the $i$th element of $\bs{\eta}$, we have the following decomposition:
\begin{align}
	\eta_i = Y_i &- \sum_{k = 1}^{p_1} \, Z_{ik}\alpha_{0k} - \sum_{\ell = 1}^{p_2} \, X_{i\ell}\beta_{0\ell} - \sum_{\ell = 1}^{p_2}\sum_{J = 1}^{N_n} \, \gamma_{0J,l}B_{J,\ell}^{(1)}(X_{i\ell}) \nonumber \\
	= Y_i &\!-\! \sum_{k \in \mathcal{S}_z} \, Z_{ik}\alpha_{0k} \!-\! \sum_{\ell \in \mathcal{S}_{x,L}} \, X_{i\ell}\beta_{0\ell} \!-\! \sum_{\ell \in \mathcal{S}_{x,N}}\, g_{0\ell} (X_{i\ell})
	\!-\! \sum_{\ell \in \mathcal{S}_{x,N}}\, g_{0\ell} (X_{i\ell})  \!-\! \sum_{\ell \in \mathcal{S}_{x,N}}\sum_{J = 1}^{N_n} \gamma_{0J,l}B_{J,\ell}^{(1)} (X_{i\ell})  \nonumber \\
	= \varepsilon_i &+ \sum_{\ell \in \mathcal{S}_{x,N}} \, \delta_{i\ell}, \label{EQN:eta_form}
\end{align}
where $\delta_{i\ell} = g_{0\ell} (X_{i\ell}) - \sum_{J = 1}^{N_n} \, \gamma_{0J,l}B_{J,\ell}^{(1)} (X_{i\ell})$.
Let $\delta_i = \sum_{\ell \in \mathcal{S}_{x,N}} \delta_{i\ell}$,
$\bs{\delta} =  (\delta_1, \ldots, \delta_n)^{\top}$, and $\bs{\varepsilon} =  (\varepsilon_1, \ldots, \varepsilon_n)^{\top}$. Then $\bs{\eta = \varepsilon + \delta}$. Define $\sigsy{\delta} = (\sum_{\ell + p_1 + p_2 \in \mathcal{S}^{\prime}, 1\leq \ell \leq p_2} \delta_{i\ell}, i = 1, \dots, n)^{\top}$. By (\ref{EQN:eta_form}) and the fact that ${|\delta_{i\ell}| = O_P(N_n^{-1})}$,
\begin{eqnarray}
    \Vert \bs{\eta}^{\ast} \Vert^2 = \Vert \bs{\varepsilon}^{\ast} + \sigsy{\delta}^{\ast} \Vert^2 \leq 2 \Vert \bs{\varepsilon}^{\ast} \Vert^2 + 2 \Vert \sigsy{\delta} \Vert^2 \leq 2 \Vert \bs{\varepsilon}^{\ast} \Vert^2 + O_P (nd^{\prime 2}N_n^{-2}),
\label{EQN:eta_ord}
\end{eqnarray}
where $\bs{\varepsilon}^{\ast} \equiv \mathbf{P}_{\mathbf{D}_{\mathcal{S}^{\prime}}}\bs{\varepsilon}$ and $\sigsy{\delta}^{\ast} \equiv \mathbf{P}_{\mathbf{D}_{\mathcal{S}^{\prime}}}\sigsy{\delta}$ are the projections of $\bs{\varepsilon}$ and $\sigsy{\delta}$ onto the column space of $\mathbf{D}_{\mathcal{S}^{\prime}}$, respectively.
Define $T_1 = \max_{1 \leq k \leq p_1}$ $|n^{-1/2} \sum_{i = 1}^{n} Z_{ik} \varepsilon_i|$, $T_2 = \max_{1 \leq \ell \leq p_2}$ $|n^{-1/2} \sum_{i = 1}^{n} X_{i\ell} \varepsilon_i|$, and $T_3 = \max_{1 \leq \ell \leq p_2,  1 \leq J \leq N_n}$ $|n^{-1/2} \sum_{i = 1}^{n} B_{J,\ell}^{(1)}(X_{i\ell}) \varepsilon_i|$.
Then,
$\Vert \bs{\varepsilon}^{\ast} \Vert^2 = \Vert \left(\mathbf{D}_{\mathcal{S}^{\prime}}^{\top}\mathbf{D}_{\mathcal{S}^{\prime}}\right)^{-1/2} \mathbf{D}_{\mathcal{S}^{\prime}}^{\top}\bs{\varepsilon} \Vert^2 \leq (nc_{*})^{-1}\Vert \mathbf{D}_{\mathcal{S}^{\prime}}^{\top}\bs{\varepsilon} \Vert^2$,
and
\[
\max_{\mathcal{A}: |\mathcal{A}| \leq d^{\prime}} \Vert \mathbf{D}_\mathcal{A}^{\top}\bs{\varepsilon} \Vert^2  =  \max_{\mathcal{A}: |\mathcal{A}| \leq d^{\prime}} \sum_{m \in \mathcal{A}} \Vert \mathbf{D}_m^{\top}\bs{\varepsilon} \Vert^2 \leq nd^{\prime}(T_1^2 \vee T_2^2 \vee N_nT_3^2).
\]
By Lemma \ref{LEM:2}, $\max_{\mathcal{A}: |\mathcal{A}| \leq d^{\prime}} \Vert \mathbf{D}_\mathcal{A}^{\top}\bs{\varepsilon} \Vert^2 = O_P[n d^{\prime}\{\ln(p_1) \vee N_n \ln(p_2N_n)\}]$. Therefore,
    \begin{eqnarray}
        \Vert \bs{\varepsilon}^{\ast} \Vert^2 = O_P[d^{\prime}c_{*}^{-1}\{\ln(p_1) \vee N_n \ln(p_2N_n)\}].
    \label{EQN:eps_ord}
    \end{eqnarray}
Combing (\ref{EQN:the_diffs1}), (\ref{EQN:eta_ord}) and (\ref{EQN:eps_ord}), we conclude that
    \begin{align*}
        \Vert \sigsy{\tbtheta} - \bs{\theta}_{0,\mathcal{S}^{\prime}}  \Vert^2
 & =  O_P\left[\frac{d^{\prime}\left\{\ln(p_1) \vee N_n \ln(p_2N_n)\right\}}{nc_{*}}\right]  + O\left(\frac{d^{\prime}}{N^2 c_{*}}\right) + \frac{4 d^{\prime} \sum_{j=1}^3 \widetilde\lambda_{nj}^2}{n^{2}c_{*}^{2}} \\
 & =   O_P\left[n^{-1}\{\ln(p_1) \vee N_n \ln(p_2N_n)\}\right] + O\left(N_n^{-2}\right) + O\left(n^{-2}\sum_{j=1}^{3}\widetilde{\lambda}_{nj}^2\right),
    \end{align*}
where the last inequality follows by $d^{\prime} = O(|\mathcal{S}_z| + |\mathcal{S}_{x,L}| + |\mathcal{S}_{x,N}|)$ and $c_{*} \asymp 1$ with probability approaching one.
By the properties of splines \citep{de2001practical},
  $\Vert \widetilde{g}_{\ell} - g_{n\ell} \Vert_2^2 \asymp \Vert \tbgamma_{\ell} - \bs{\gamma}_{0\ell}  \Vert^2$,
where $g_{n\ell}, \, \ell=1, \ldots, p_2$, is the best approximation for function $g_{\ell}$.
Hence, part (ii) follows from $\sum_{k = 1}^{p_1}| \widetilde{\alpha}_k - {\alpha}_{0k} |^2 = O(\Vert \sigsy{\tbtheta} - \bs{\theta}_{0,\mathcal{S}^{\prime}} \Vert^2)$, $\sum_{\ell = 1}^{p_2}| \widetilde{{\beta}}_{\ell} - {\beta}_{0\ell}|^2 = O(\Vert \sigsy{\tbtheta} - \bs{\theta}_{0,\mathcal{S}^{\prime}} \Vert^2)$ and $\sum_{\ell = 1}^{p_2}\| \widetilde{\bs{\gamma}}_{\ell} - {\bs{\gamma}}_{0\ell} \|^2 = O(\Vert \sigsy{\tbtheta} - \bs{\theta}_{0,\mathcal{S}^{\prime}}  \Vert^2)$.

We now prove part (i). Under Assumption (A1), if $\Vert \bs{\theta}_{0m} \Vert \neq 0$ but $\Vert \tbtheta_{m} \Vert = 0$, then $\Vert \bs{\theta}_{0m} - \tbtheta_{m}  \Vert \geq c_{\alpha} \vee c_{\beta} \vee c_g$, which contradicts part (ii) when ${\ln(p_1) \vee N_n \ln(p_2N_n)} / {n} \rightarrow 0$, ${\widetilde{\lambda}_{n1}^2} / {n^2} \rightarrow 0$, ${\widetilde{\lambda}_{n2}^2} / {n^2} \rightarrow 0$ and ${\widetilde{\lambda}_{n3}^2} / {n^2} \rightarrow 0$.
The results follow by
\begin{align*}
\widetilde{\bs{\alpha}} - \bs{\alpha}_{0} &= \left(\mathbf{I}_{|\mathcal{S}_z|}  ~~ \bs{0}_{|\mathcal{S}_z| \times |\mathcal{S}_{x,L}|}  ~~ \bs{0}_{|\mathcal{S}_z| \times (|\mathcal{S}_{x,N}|N_n)}\right)(\tbtheta - \bs{\theta}_{0}), \\
\tbbeta - \bs{\beta}_{0} &= \left(\bs{0}_{|\mathcal{S}_{x,L}| \times |\mathcal{S}_z|} ~~ \mathbf{I}_{|\mathcal{S}_{x,L}|}  ~~ \bs{0}_{|\mathcal{S}_{x,L}|\times (|\mathcal{S}_{x,N}|N_n)}\right)(\tbtheta - \bs{\theta}_{0}),\\
\tbgamma - \bs{\gamma}_{0} &= \left(\bs{0}_{(|\mathcal{S}_{x,N}| N_n) \times |\mathcal{S}_z|} ~~ \bs{0}_{(|\mathcal{S}_{x,N}| N_n) \times |\mathcal{S}_{x,L}|} ~~ \mathbf{I}_{|\mathcal{S}_{x,N}| N_n}\right)(\tbtheta - \bs{\theta}_{0})
\end{align*}
and the definition of $\widetilde{g}_{\ell}$, $1 \leq \ell \leq p_2$.
\end{proof}

\subsection{Selection and estimation properties of the adaptive group LASSO estimators} \label{SUBSEC:KKT}
In this section, we establish the selection and estimation properties of the adaptive group LASSO estimators as stated in Theorems \ref{THM:selection} and \ref{THM:consistency}.

\begin{proof}[Proof of Theorem \ref{THM:selection}]
By the Karush-Kuhn-Tucker (KKT) condition \citep{boyd2004convex}, if $(\widehat{\bs{\alpha}}, \widehat{\bs{\beta}}, \widehat{\bs{\gamma}})$ is the unique minimizer of $L(\bs{\alpha}, \bs{\beta}, \bs{\gamma}; \lambda_1, \lambda_2, \lambda_3)$, it is equivalent to satisfy
\begin{enumerate}
{\renewcommand{\itemsep}{-1pt}
    \item[(C1-1)] $\mathbf{Z}_k^{\top} \left(\mathbf{Y}-\mathbf{Z}\bs{\alpha} - \mathbf{X}\bs{\beta}-\sum_{{\ell^{\prime}} = 1}^{p_2}\mathbf{B}_{\ell^{\prime}}^{(1)}\bs{\gamma}_{\ell^{\prime}}\right) = \lambda_{n1} w_k^{\alpha}\alpha_k / |\alpha_k|$, for any $k\in \mathcal{S}_z$,

    \item[(C1-2)] $\mathbf{X}_{\ell}^{\top} \left(\mathbf{Y}-\mathbf{Z}\bs{\alpha} - \mathbf{X}\bs{\beta}-\sum_{{\ell^{\prime}} = 1}^{p_2}\mathbf{B}_{\ell^{\prime}}^{(1)}\bs{\gamma}_{\ell^{\prime}}\right) = \lambda_{n2} w_{\ell}^{\beta}  \beta_{\ell} / |\beta_{\ell}|$, for any $\ell \in \mathcal{S}_{x, L}$,

    \item[(C1-3)] $\mathbf{B}_{\ell}^{(1)\top} \left(\mathbf{Y}-\mathbf{Z}\bs{\alpha} - \mathbf{X}\bs{\beta}-\sum_{{\ell^{\prime}} = 1}^{p_2}\mathbf{B}_{\ell^{\prime}}^{(1)}\bs{\gamma}_{\ell^{\prime}}\right) = \lambda_{n3} w_{\ell}^{\gamma}  \bs{\gamma}_{\ell} / \Vert\bs{\gamma}_{\ell}\Vert$, for any $\ell \in \mathcal{S}_{x,N}$,

    \item[(C2)] $\left|\mathbf{Z}_k^{\top} \left(\mathbf{Y}-\mathbf{Z}\bs{\alpha} - \mathbf{X}\bs{\beta}-\sum_{{\ell^{\prime}} = 1}^{p_2}\mathbf{B}_{\ell^{\prime}}^{(1)}\bs{\gamma}_{\ell^{\prime}}\right)\right| \leq\lambda_{n1} w_k^{\alpha}$, for any $k\in \mathcal{N}_z$,

    \item[(C3)] $\left|\mathbf{X}_{\ell}^{\top} \left(\mathbf{Y}-\mathbf{Z}\bs{\alpha} - \mathbf{X}\bs{\beta}-\sum_{{\ell^{\prime}} = 1}^{p_2}\mathbf{B}_{\ell^{\prime}}^{(1)}\bs{\gamma}_{\ell^{\prime}}\right)\right|\leq \lambda_{n2} w_{\ell}^{\beta}$, for any $\ell \in \mathcal{N}_x \cup \mathcal{S}_{x, PN}$,

    \item[(C4)] $\Vert\mathbf{B}_{\ell}^{(1)\top} \left(\mathbf{Y}-\mathbf{Z}\bs{\alpha} - \mathbf{X}\bs{\beta}-\sum_{{\ell^{\prime}} = 1}^{p_2}\mathbf{B}_{\ell^{\prime}}^{(1)}\bs{\gamma}_{\ell^{\prime}}\right)\Vert \leq \lambda_{n3} w_{\ell}^{\gamma}$, for any $\ell \in \mathcal{N}_x\cup \mathcal{S}_{x, PL}$,
}
\end{enumerate}

Define $ \overline{\bs{\theta}}^o = \left(\mathbf{D}_{\mathcal{S}}^{\top} \, \mathbf{D}_{\mathcal{S}}\right)^{-1} \mathbf{D}_{\mathcal{S}}^{\top}\mathbf{Y}$, a vector with length $|\mathcal{S}_z| + |\mathcal{S}_{x,L}| + |\mathcal{S}_{x,N}|N_n$.
Denote three vectors, $\bs{v}_1$, $\bs{v}_2$ and $\bs{v}_3$, whose elements are in the form:
\begin{align}
 & \bs{v}_{1m} = \frac{\omega_m^{\alpha}\overline{\theta}_{0m}}{|\overline{\theta}_{0m}|} I\{m \in \mathcal{S}_z\} + \bs{0}_{N} I\{m - |\mathcal{S}_z| - |\mathcal{S}_{x,LN}| \in \mathcal{S}_{x,N}\}, \label{DEF:vj}\\
 & \bs{v}_{2m} = \frac{\omega_{m - |\mathcal{S}_z|}^{\beta}\overline{\theta}_{0m}}{|\overline{\theta}_{0m}|} I\{m - |\mathcal{S}_z| \in \mathcal{S}_{x, L}\} + \bs{0}_{N} I\{m - |\mathcal{S}_z| - |\mathcal{S}_{x,LN}| \in \mathcal{S}_{x,N}\}, \notag \\
 & \bs{v}_{3m} = \frac{\omega_{m - |\mathcal{S}_z| - |\mathcal{S}_{x,L}|}^{\gamma}\overline{\bs{\theta}}_{0m}}{\Vert \overline{\bs{\theta}}_{0m} \Vert} I\{m - |\mathcal{S}_z| - |\mathcal{S}_{x,LN}| \in \mathcal{S}_{x,N}\},
    ~ \forall m \in \mathcal{S}. \notag
\end{align}

Next  define $\hbtheta^o = (\widehat{\bs{\theta}}_m^o, 1 \leq m \leq p_1 + 2p_2)^{\top}$,
where $\hbtheta_{\mathcal{S}}^o \equiv (\widehat{\bs{\theta}}_m^o,  m \in \mathcal{S})^{\top} = \left(\mathbf{D}_{\mathcal{S}}^{\top} \, \mathbf{D}_{\mathcal{S}}\right)^{-1} (\mathbf{D}_{\mathcal{S}}^{\top}\mathbf{Y} - \sum_{j=1}^{3}\lambda_{nj} \bs{v}_j)$,
$\widehat{\theta}_m^o = 0$ for $m \in \mathcal{N}_z$ and $m - p_1 \in \mathcal{N}_{x, L}$, and $\hbtheta_m^o = \bs{0}_N$ for $m - p_1 - p_2 \in \mathcal{N}_{x,N}$.
So we can represent $\hbtheta^o \equiv (\widehat{\bs{\theta}}_{\mathcal{S}_z}^{o\top}, \widehat{\bs{\theta}}_{\mathcal{N}_z}^{o\top}, \widehat{\bs{\theta}}_{\mathcal{S}_{x, L}}^{o\top}, \widehat{\bs{\theta}}_{\mathcal{N}_{x, L}}^{o\top}, \widehat{\bs{\theta}}_{\mathcal{S}_{x, N}}^{o\top}, \widehat{\bs{\theta}}_{\mathcal{N}_{x, N}}^{o\top})^{\top}$,
and $\hbtheta_{\mathcal{S}}^o \equiv (\widehat{\bs{\theta}}_{\mathcal{S}_z}^{o\top}, \widehat{\bs{\theta}}_{\mathcal{S}_{x, L}}^{o\top}, \widehat{\bs{\theta}}_{\mathcal{S}_{x, N}}^{o\top})^{\top}$.
Denote $\widehat{\mathcal{S}}^o =\{1 \leq m \leq p_1 + 2p_2: \|\hbtheta_m^o\| > 0\}$. Apparently, $\widehat{\mathcal{S}}^o \subseteq \mathcal{S}$.
Notice that $\mathbf{D}\hbtheta^o = \mathbf{D}_{\mathcal{S}}\hbtheta_{\mathcal{S}}^o$ and $\{\mathbf{D}_m, \, m \in \mathcal{S} \}$ are linearly independent,
so by the definition of $\hbtheta^o$, (C1-1), (C1-2) and (C1-3) hold for $\hbtheta^{o}$ if $\widehat{\mathcal{S}}^o \supseteq \mathcal{S}$.
Therefore, if $\hbtheta^o$ satisfies
\begin{enumerate}
{\renewcommand{\itemsep}{-1pt}
    \item[(C1$^{\prime}$)] $\widehat{\mathcal{S}}^o \supseteq \mathcal{S}$,

    \item[(C2$^{\prime}$)] $\left|\mathbf{Z}_k^{\top} \left(\mathbf{Y} - \mathbf{D}\widehat{\bs{\theta}}^o\right)\right| \leq\lambda_{n1}\omega_k^{\alpha}$, for any $k\in \mathcal{N}_z$,

    \item[(C3$^{\prime}$)] $\Vert\mathbf{X}_{\ell}^{\top} \left(\mathbf{Y} - \mathbf{D}\widehat{\bs{\theta}}^o\right)\Vert \leq \lambda_{n2} \omega_{\ell}^{\beta}$, for any $\ell \in \mathcal{N}_{x, L}$,

    \item[(C4$^{\prime}$)] $\Vert\mathbf{B}_{\ell}^{(1)\top} \left(\mathbf{Y} - \mathbf{D}\widehat{\bs{\theta}}^o\right)\Vert \leq \lambda_{n3} \omega_{\ell}^{\gamma}$, for any $\ell \in \mathcal{N}_{x, N}$,
}
\end{enumerate}
then $\hbtheta^o$ is the unique minimizer of $L_n(\bs{\theta};\lambda_{n1},\lambda_{n2},\lambda_{n3})$, in other words, $\hbtheta^o = \hbtheta$ with probability approaching one.
Therefore, in order to show $\Pr(\widehat{\mathcal{S}}=\mathcal{S}) \rightarrow 1$, it is equivalent to show $\hbtheta^o$ satisfies (C1$^{\prime}$)--(C3$^{\prime}$) with probability approaching one, as $n \rightarrow \infty$.

Further notice that
\begin{enumerate}
{\renewcommand{\itemsep}{-1pt}
    \item[(C1$^{\prime\prime}$)] $\Vert  \bs{\theta}_{0m} \Vert - \Vert  \hbtheta_m^o \Vert < \Vert  \bs{\theta}_{0m} \Vert$, $\forall \, m \in \mathcal{S}$
}
\end{enumerate}
 implies Condition (C1$^{\prime}$).
 Therefore, to show $\hbtheta^o$ is the unique minimizer of $L_n(\bs{\theta};\lambda_{n1},\lambda_{n2},$ $\lambda_{n3})$, and consequently, $\Pr(\widehat{\mathcal{S}}=\mathcal{S}) \rightarrow 1$, it suffices to show that $\hbtheta^o$ satisfies Conditions (C1$^{\prime\prime}$), (C2$^{\prime}$) and (C3$^{\prime}$) with probability approaching one, as $n \rightarrow \infty$.

According to Lemma \ref{LEM:sel_cons_S} and Lemma \ref{LEM:sel_cons_N} below, we obtain that
\begin{align*}
	\Pr(\widehat{\mathcal{S}} \neq \mathcal{S}) \leq &\Pr(\Vert  \bs{\theta}_{0m} - \hbtheta_m^o  \Vert \geq \Vert  \bs{\theta}_{0m} \Vert, \, \exists  \, m \in \mathcal{S})  + \Pr(|\mathbf{Z}_k^{\top} (\mathbf{Y} - \mathbf{D}\hbtheta^o)| > \lambda_{n1}\omega_k^{\alpha} , \, \exists \, k \in \mathcal{N}_z) \\
	& + \Pr(\Vert\mathbf{X}_{\ell}^{\top} (\mathbf{Y} - \mathbf{D}\hbtheta^o)\Vert >  \lambda_{n2} \omega_{\ell}^{\beta}, \, \exists \, \ell \in \mathcal{N}_{x, L})\\
    & + \Pr(\Vert\mathbf{B}_{\ell}^{(1)\top} (\mathbf{Y} - \mathbf{D}\hbtheta^o)\Vert >  \lambda_{n3} \omega_{\ell}^{\gamma} , \, \exists \, \ell \in \mathcal{N}_{x, N}) \rightarrow 0,
\end{align*}
as $n \rightarrow \infty$. This completes the proof.
\end{proof}

The following Lemma \ref{LEM:sel_cons_S} and Lemma \ref{LEM:sel_cons_N} are used in the proof of Theorem \ref{THM:selection}.
\begin{lemma}
\label{LEM:sel_cons_S} 
Under Assumptions (A3)--(A6), as $n \rightarrow \infty$,
    \[
        \Pr \left(\Vert  \bs{\theta}_{0m} - \hbtheta_m^o  \Vert \geq \Vert  \bs{\theta}_{0m} \Vert, \, \exists  \, m \in \mathcal{S}\right) \to 0.
    \]
\end{lemma}

\begin{proof}
Let $\mathbf{Q}_{m}$ be an $d_m \times (|\mathcal{S}_z|+|\mathcal{S}_{x,L}|+|\mathcal{S}_{x,N}|N_n\}$ matrix, $d_m = 1$ for $m \in \mathcal{S}_z$ or $m\in\mathcal{S}_{x,L}$, and $d_m = N_n$ for $m\in\mathcal{S}_{x,N}$, with the form
\begin{align*}
  \mathbf{Q}_m = \begin{pmatrix}
    \mathbf{Q}_{1m} & \bs{0}_{(|\mathcal{S}_{x,N}|N_n)\times (|\mathcal{S}_{x,N}|N_n)}
  \end{pmatrix} & I(m \in \mathcal{S}_z \cup \mathcal{S}_{x,L})\\
    & + \begin{pmatrix}
    \bs{0}_{N_n\times (|\mathcal{S}_z|+|\mathcal{S}_{x,L}|)} & \mathbf{Q}_{2, m - |\mathcal{S}_z| - |\mathcal{S}_{x,L}|}
  \end{pmatrix} I(m \in \mathcal{S}_{x,N})
\end{align*}
with $\mathbf{Q}_{1m} = (0, \ldots, 0, 1, 0, \ldots, 0)$ and $\mathbf{Q}_{2m} = (\bs{0}_{N_n\times N_n}, \ldots, \bs{0}_{N_n\times N_n}, \mathbf{I}_{N_n}, \bs{0}_{N_n\times N_n}, \ldots,$ $\bs{0}_{N_n\times N_n})$,
where scalar 1 is the $m$-th element of vector $\mathbf{Q}_{1m} $ with length $|\mathcal{S}_z|$, and an $N_n \times N_n$ identity matrix $\mathbf{I}_{N_n}$ is at the $m$-th block of the $N_n \times (|\mathcal{S}_{x,N}|N_n)$ matrix $\mathbf{Q}_{2m}$ with rest $N_n \times N_n$ matrices of zeros $\bs{0}_{N_n \times N_n}$.

Then from (\ref{EQN:theta_dfc}), $\hbtheta_m^o - \bs{\theta}_{0m} = n^{-1} \mathbf{Q}_{m} \, \mathbf{C}_{\mathcal{S}}^{-1} \left(\mathbf{D}_{\mathcal{S}}^{\top} \bs{\varepsilon}  + \mathbf{D}_{\mathcal{S}}^{\top} \bs{\delta} - \sum_{j=1}^3 \lambda_{nj}\bs{v}_j\right )$.
By the triangle inequality,
\[
	\left\Vert \hbtheta_m^o - \bs{\theta}_{0m} \right\Vert \leq
	 n^{-1} \left\Vert \mathbf{Q}_{m} \mathbf{C}_{\mathcal{S}}^{-1} \mathbf{D}_{\mathcal{S}}^{\top} \bs{\varepsilon}  \right\Vert
	 + \, n^{-1} \left\Vert \mathbf{Q}_{m} \mathbf{C}_{\mathcal{S}}^{-1} \mathbf{D}_{\mathcal{S}}^{\top} \bs{\delta} \right\Vert
	 + \, n^{-1} \left\Vert \mathbf{Q}_{m} \mathbf{C}_{\mathcal{S}}^{-1} \left(\sum_{j=1}^3 \lambda_{nj}\bs{v}_j\right)\right\Vert .
\]
Recall that $\pi_1$ and $\pi_2$ are the minimum and maximum eigenvalues of $\mathbf{C}_{\mathcal{S}}$, respectively. By Lemmas \ref{LEM:2} and \ref{LEM:eigen}, the first term on the right-hand side
\begin{align*}
	\max\limits_{m \in \mathcal{S}} \, & n^{-1} \Vert \mathbf{Q}_{m} \mathbf{C}_{\mathcal{S}}^{-1} \mathbf{D}_{\mathcal{S}}^{\top} \bs{\varepsilon}  \Vert_2
	\leq \max\limits_{m \in \mathcal{S}} \,  n^{-1} \pi_1^{-1} \Vert \, \mathbf{Q}_{m} \mathbf{D}_{\mathcal{S}}^{\top} \bs{\varepsilon}  \Vert_2  \\
	& = n^{-1} \pi_1^{-1} \max\limits_{k,l \in \mathcal{S}} \,  \left\{\sum_{i=1}^n |Z_{ik} \varepsilon_i|, \sum_{i=1}^n |X_{i\ell} \varepsilon_i|, \sum_{i=1}^n \|\mathbf{B}_{i\ell} \varepsilon_i\|\right\}
	= {O_P \left(n^{-1/2}N_n^{1/2} \right) }.
\end{align*}
By Lemma \ref{LEM:eigen}, the second term
\begin{align*}
	\max\limits_{m \in \mathcal{S}} \, & n^{-1} \Vert \mathbf{Q}_{m} \mathbf{C}_{\mathcal{S}}^{-1} \mathbf{D}_{\mathcal{S}}^{\top} \bs{\delta} \, \Vert_2
	\leq \max\limits_{m \in \mathcal{S}} \,  n^{-1} \pi_1^{-1} \Vert \, \mathbf{Q}_{m} \mathbf{D}_{\mathcal{S}}^{\top} \bs{\delta}  \Vert_2  \\
	& = n^{-1} \pi_1^{-1} \max\limits_{k,l \in \mathcal{S}} \,  \left\{\sum_{i=1}^n |Z_{ik} \delta_i|, \sum_{i=1}^n |X_{i\ell} \delta_i|, \sum_{i=1}^n \|\mathbf{B}_{i\ell} \delta_i\|\right\}
	 = O_P (N_n^{-1}).
\end{align*}
By Lemma \ref{LEM:eigen} and Lemma \ref{LEM:4}, the third term
\begin{align*}
	\max\limits_{m \in \mathcal{S}} n^{-1} & \Vert \mathbf{Q}_{m} \mathbf{C}_{\mathcal{S}}^{-1} \left(\sum_{j=1}^3 \lambda_{nj}\bs{v}_j\right)\Vert_2
	\leq n^{-1} \pi_1^{-1} \left\Vert \sum_{j=1}^3 \lambda_{nj}\bs{v}_j\right\Vert
	= O_P \left\{n^{-1} \left(\sum_{j=1}^3 \lambda_{nj}h_{nj}\right) \right\}.
\end{align*}
Thus, the claim follows by Assumption (A6).
\end{proof}

\begin{lemma}
\label{LEM:sel_cons_N} 
Under Assumptions (A3)--(A6), as $n \rightarrow \infty$,
    \begin{eqnarray*}
        &&\Pr \left(|\mathbf{Z}_k^{\top} (\mathbf{Y} - \mathbf{D}\hbtheta^o)| > \lambda_{n1}\omega_k^{\alpha} , \, \exists \, k \in \mathcal{N}_z \right) \to 0, \\
        &&\Pr \left(\Vert\mathbf{X}_{\ell}^{\top} (\mathbf{Y} - \mathbf{D}\hbtheta^o)\Vert >  \lambda_{n2} \omega_{\ell}^{\beta}, \, \exists \, \ell \in \mathcal{N}_{x, L} \right)  \to 0, \\
        &&\Pr \left(\Vert\mathbf{B}_{\ell}^{(1)\top} (\mathbf{Y} - \mathbf{D}\hbtheta^o)\Vert >  \lambda_{n3} \omega_{\ell}^{\gamma} , \, \exists \, \ell \in \mathcal{N}_{x, N} \right) \to 0.
    \end{eqnarray*}
\end{lemma}

\begin{proof}
Note that
\begin{align}
	\mathbf{Y} - \mathbf{D}_{\mathcal{S}}\sigs{\hbtheta}^o
	&= \mathbf{Z}_{\mathcal{S}_z} \bs{\alpha}_{0,\mathcal{S}_z} + \mathbf{X}_{\mathcal{S}_{x,L}} \bs{\beta}_{0,\mathcal{S}_{x,L}} \!+\! \sum_{\ell \in \mathcal{S}_{x,N}} \, {\alpha}_{0\ell}\left(\mathbf{X}_{\ell}\right) \!+\! \bs{\varepsilon} - \mathbf{D}_{\mathcal{S}}\sigs{\hbtheta} = \mathbf{D}_{\mathcal{S}} \bs{\theta}_{0,\mathcal{S}} + \bs{\delta} + \bs{\varepsilon}\!-\! \mathbf{D}_{\mathcal{S}}\sigs{\hbtheta} \nonumber \\
	&= \left\{\mathbf{I} - \mathbf{D}_{\mathcal{S}}\left(\mathbf{D}_{\mathcal{S}}^{\top} \mathbf{D}_{\mathcal{S}}\right)^{-1}\mathbf{D}_{\mathcal{S}}^{\top} \right\} \left(\bs{\delta}+\bs{\varepsilon}\right) + \mathbf{D}_{\mathcal{S}}\left(\mathbf{D}_{\mathcal{S}}^{\top} \mathbf{D}_{\mathcal{S}}\right)^{-1}\left(\sum_{j=1}^3 \lambda_{nj}\bs{v}_j\right) \nonumber \\
	&= \mathbf{H}\bs{\varepsilon} + \mathbf{H}\bs{\delta} + n^{-1}\mathbf{D}_{\mathcal{S}}\mathbf{C}_{\mathcal{S}}^{-1}\left(\sum_{j=1}^3 \lambda_{nj}\bs{v}_j\right). \label{EQN:yds}
\end{align}
For $1 \leq m \leq p_{1} + 2p_{2}$, by (\ref{EQN:yds}), we have
\[
	\mathbf{D}_m^{\top} (\mathbf{Y} - \mathbf{D}_{\mathcal{S}}\sigs{\hbtheta}^o) = \mathbf{D}_m^{\top} \mathbf{H}\bs{\varepsilon} + \mathbf{D}_m^{\top} \mathbf{H}\bs{\delta} + n^{-1}\mathbf{D}_m^{\top}  \mathbf{D}_{\mathcal{S}}\mathbf{C}_{\mathcal{S}}^{-1}\left(\sum_{j=1}^3 \lambda_{nj}\bs{v}_j\right).
\]
By Lemma \ref{LEM:2},
\begin{align*}
	\mathrm{E}\left( \max\limits_{m \in \mathcal{N}_z} \Vert n^{-1/2} \mathbf{D}_m^{\top} \mathbf{H}\bs{\varepsilon}\Vert_2 \right)
	&= \mathrm{E}\left( \max\limits_{k \in \mathcal{N}_z} \Vert n^{-1/2} \mathbf{Z}_k^{\top} \mathbf{H}\bs{\varepsilon}\Vert_2 \right)
    = O \{\sqrt{\ln (p_{1})}\}, \\
	\mathrm{E}\left( \max\limits_{m - p_1 \in \mathcal{N}_{x,L}} \Vert n^{-1/2} \mathbf{D}_m^{\top} \mathbf{H}\bs{\varepsilon}\Vert_2 \right)
	&= \mathrm{E}\left( \max\limits_{\ell \in \mathcal{N}_{x,L}} \Vert n^{-1/2} \mathbf{X}_{\ell}^{\top} \mathbf{H}\bs{\varepsilon}\Vert_2 \right)
	= O \{\sqrt{\ln\left(p_{2}\right)} \},, \\
	\mathrm{E}\left( \max\limits_{m - p_1 - p_2 \in \mathcal{N}_{x,N}} \Vert n^{-1/2} \mathbf{D}_m^{\top} \mathbf{H}\bs{\varepsilon}\Vert_2 \right)
	&= \mathrm{E}\left( \max\limits_{\ell \in \mathcal{N}_{x,N}} \Vert n^{-1/2} \mathbf{B}_{\ell}^{(1)\top} \mathbf{H}\bs{\varepsilon}\Vert_2 \right)
	= O \{\sqrt{\ln\left(p_{2} N_n \right)} \},
\end{align*}
then for the $\mathbf{D}_m^{\top} \mathbf{H}\bs{\varepsilon}$ part, from  Condition (A5), for all $k \in \mathcal{N}_z, \, \omega_k^{\alpha} = |\widetilde{\alpha}_k|^{-1} = O_P(r_{n1})$, there exists a positive constant $c_{1}$, such that
    \begin{eqnarray*}
        && \Pr\left(|\mathbf{Z}_k^{\top} \mathbf{H}\bs{\varepsilon} | > \lambda_{n1}\omega_k^{\alpha} / 3, \, \exists \, k \in \mathcal{N}_z \right)
        \leq \Pr\left(|\mathbf{Z}_k^{\top} \mathbf{H}\bs{\varepsilon} | > c_1\lambda_{n1}r_{n1}, \, \exists \, k \in \mathcal{N}_z\right) + o(1) \\
        &=& \Pr \left( \max \limits_{k \in \mathcal{N}_z} | n^{-1/2} \mathbf{Z}_k^{\top} \mathbf{H}\bs{\varepsilon}| >  c_1 n^{-1 / 2} \lambda_{n3}r_{n1} \right) + o(1) \\
        &\leq& n^{1 / 2}(c_1 \lambda_{n3}r_{n1})^{-1}\mathrm{E}\left( \max\limits_{k \in \mathcal{N}_z} | n^{-1/2} \mathbf{Z}_k^{\top} \mathbf{H}\bs{\varepsilon}| \right)  + o(1)
        \leq n^{1 / 2} O\{\ln (p_{1})^{1 / 2} \}{c_1^{-1} \lambda_{n3}^{-1}r_{n1}^{-1} } + o(1) \\
        &=& O\{n^{1 / 2}\sqrt{\ln (p_{1})}\lambda_{n3}^{-1}r_{n1}^{-1}\} + o(1) .
    \end{eqnarray*}
Similarly,
\[
\Pr\left(\Vert \mathbf{X}_{\ell}^{\top} \mathbf{H}\bs{\varepsilon} \Vert_2 > \lambda_{n2}\omega_{\ell}^{\beta} / 3, \, \exists \, \ell \in \mathcal{N}_{x,L} \right)= O\left(n^{1 / 2}\sqrt{\ln(p_2)}{\lambda_{n2}^{-1} r_{n2}^{-1}}\right) + o(1),
\]
and
\[
 \Pr\left(\Vert \mathbf{B}_{\ell}^{(1)\top} \mathbf{H}\bs{\varepsilon} \Vert_2 > \lambda_{n3}\omega_{\ell}^{\gamma} / 3, \, \exists \,l \in \mathcal{N}_{x,N} \right)= O\left(n^{1 / 2}\sqrt{N_n\ln(p_{2} N_n) }{\lambda_{n3}^{-1} r_{n3}^{-1}}\right) + o(1).
\]
Recall the definition of $\mathcal{N}$ in (\ref{DEF:SN}) and $\pi_3 = \max_{m \notin \mathcal{S}} \Vert n^{-1}\mathbf{D}_m^{\top}\mathbf{D}_m \Vert_2$, by the properties of spline \citep{de2001practical}, the $\mathbf{D}_m^{\top} \mathbf{H}\bs{\delta}$ term has
\begin{align*}
\max\limits_{m \in \mathcal{N}}  \Vert \mathbf{D}_m^{\top} \mathbf{H}\bs{\delta} \Vert_2
&\leq n^{1/2}  \max\limits_{m \notin \mathcal{S}} \Vert \frac{1}{n} \mathbf{D}_m^{\top} \mathbf{D}_m \Vert_2^{1 / 2} \Vert \mathbf{H} \Vert_2 \Vert \bs{\delta} \Vert_2 \!
= O_{P}\left(n\pi_3^{1 / 2}N_n^{-1}|\mathcal{S}_{x,N}|^{1 / 2}\right) =O_{P}\left(nN_n^{-1}\right).
\end{align*}
and the last term follows by Lemma \ref{LEM:eigen} and (\ref{EQN:v1_order}) and (\ref{EQN:v2_order}) in Lemma \ref{LEM:4} that
\begin{align*}
        \max\limits_{m \in \mathcal{N}} &\Vert n^{-1}\mathbf{D}_m^{\top}  \mathbf{D}_{\mathcal{S}}\mathbf{C}_{\mathcal{S}}^{-1}\left(\sum_{j=1}^3 \lambda_{nj}\bs{v}_j\right) \Vert_2
\leq \max\limits_{m \notin \mathcal{S}} \Vert n^{-1 / 2} \mathbf{D}_m \Vert_2 \times \Vert n^{- 1 / 2} \mathbf{D}_{\mathcal{S}}\mathbf{C}_{\mathcal{S}}^{-1 / 2} \Vert_2 \times  \Vert \mathbf{C}_{\mathcal{S}}^{-1 / 2} \Vert_2 \\
        &\times \Vert \lambda_{n1}\bs{v}_1 + \lambda_{n2}\bs{v}_2 + \lambda_{n3}\bs{v}_3\Vert_2  \leq \pi_3^{1 / 2} \pi_1^{-1 / 2} {O_P\left(\sum_{j=1}^3 \lambda_{nj}h_{nj}\right)}
        = {O_P\left(\sum_{j=1}^3 \lambda_{nj}h_{nj}\right)}.
\end{align*}
\end{proof}

\begin{proof} [Proof of Theorem \ref{THM:consistency}]
The idea of the proof is similar to the proof of part (ii) in Theorem \ref{THM:LASSO-1}, but we look at index set $\mathcal{S}$ instead of $\mathcal{S}^{\prime}$. Let $\pi_1$ and $\pi_2$ be the minimum and maximum eigenvalues of $\mathbf{C}_{\mathcal{S}}$, respectively, and let $\pi_3 = \max_{m \notin \mathcal{S}} \Vert n^{-1}\mathbf{D}_m^{\top}\mathbf{D}_m \Vert$. By Lemma \ref{LEM:eigen}, $\pi_1 \asymp 1$, $\pi_2 \asymp 1$ and $\pi_3 \asymp 1$.
For any $\ell\in\mathcal{S}_{x,N}$, let $g_{0\ell} (\mathbf{X}_{\ell})=\left(g_{0\ell} (X_{1\ell}),\ldots, g_{0\ell} (X_{n\ell})\right)^{\top}$, $\bs{\delta}_{\ell}=g_{0\ell} (\mathbf{X}_{\ell})- \mathbf{B}_{\ell}^{(1)} \bs{\gamma}_{\ell}$ and $\bs{\delta}=\sum_{\ell \in \mathcal{S}_{x,N}} \bs{\delta}_{\ell}$. According to the proof of Theorem \ref{THM:selection}, with probability approaching one, we have
\begin{align*}
 {\hbtheta}_{\mathcal{S}}   &= \widehat{\bs{\theta}}_{\mathcal{S}}^o
  = \left(\mathbf{D}_{\mathcal{S}}^{\top}\mathbf{D}_{\mathcal{S}}\right)^{-1} \left(\mathbf{D}_{\mathcal{S}}^{\top}\mathbf{Y} - \sum_{j=1}^{3}\lambda_{nj} \bs{v}_j \right)
  = \left(\mathbf{D}_{\mathcal{S}}^{\top} \mathbf{D}_{\mathcal{S}}\right)^{-1} \times\\
  &\left[\mathbf{D}_{\mathcal{S}}^{\top} \left\{\mathbf{Z}_{\mathcal{S}_z}\bs{\alpha}_{0, \mathcal{S}_z} + \mathbf{X}_{\mathcal{S}_{x, L}}\bs{\beta}_{0, \mathcal{S}_{x, L}} + \sum_{\ell \in \mathcal{S}_{x,N}}   \left(\mathbf{B}_{\ell}^{(1)} \bs{\gamma}_{0\ell} + \bs{\delta}_{\ell}\right) + \bs{\varepsilon}\right\} - \sum_{j=1}^{3}\lambda_{nj} \bs{v}_j  \right] \\
  &= \bs{\theta}_{0, \mathcal{S}} + \left(\mathbf{D}_{\mathcal{S}}^{\top} \mathbf{D}_{\mathcal{S}}\right)^{-1} \left\{\mathbf{D}_{\mathcal{S}}^{\top} \left(\bs{\delta}+\bs{\varepsilon}\right) - \sum_{j=1}^{3}\lambda_{nj} \bs{v}_j \right \}.
\end{align*}
Let $\mathbf{C}_{\mathcal{S}} = n^{-1} \mathbf{D}_{\mathcal{S}}^{\top} \mathbf{D}_{\mathcal{S}}$ be an $(|\mathcal{S}_z| + |\mathcal{S}_{x,L}| + |\mathcal{S}_{x,N}|N_n)\times(|\mathcal{S}_z| + |\mathcal{S}_{x,L}| + |\mathcal{S}_{x,N}|N_n)$ matrix, and let $\mathbf{H = I} - \mathbf{D}_{\mathcal{S}} \left(\mathbf{D}_{\mathcal{S}}^{\top} \mathbf{D}_{\mathcal{S}}\right)^{-1}\mathbf{D}_{\mathcal{S}}^{\top}$ be an $n \times n$ matrix, then
\begin{equation}
    \hbtheta_{\mathcal{S}} - \bs{\theta}_{\mathcal{S}}^o = n^{-1} \mathbf{C}_{\mathcal{S}}^{-1} \left\{\mathbf{D}_{\mathcal{S}}^{\top} \left(\bs{\delta} + \bs{\varepsilon}\right) - \sum_{j=1}^{3}\lambda_{nj} \bs{v}_j \right\}.
\label{EQN:theta_dfc}
\end{equation}
For $\bs{\eta} = \mathbf{Y} - \mathbf{D}\bs{\theta}$, define $\bs{\eta}_{*}$ as the projection of $\bs{\eta}$ to the column space of $\mathbf{D}_{\mathcal{S}}$, that is, $\bs{\eta}_{*} \equiv \mathbf{P}_{\mathbf{D}_{\mathcal{S}}}\bs{\eta} = \mathbf{D}_{\mathcal{S}}(\mathbf{D}_{\mathcal{S}}^{\top}\mathbf{D}_{\mathcal{S}})^{-1}\mathbf{D}_{\mathcal{S}}^{\top}\bs{\eta}$. Then for $\bs{\varepsilon}_{*} \equiv \mathbf{P}_{\mathbf{D}_{\mathcal{S}}}\bs{\varepsilon}$, similar to (\ref{EQN:the_diffs1}), (\ref{EQN:eta_ord}) and (\ref{EQN:eps_ord}), and by Lemma \ref{LEM:eigen},
\begin{align*}
    \Vert \bs{\varepsilon}_{*} \Vert^2
    = \Vert \left(\mathbf{D}_{\mathcal{S}}^{\top}\mathbf{D}_{\mathcal{S}}\right)^{-1/2} \mathbf{D}_{\mathcal{S}}^{\top}\bs{\varepsilon} \Vert^2
    \leq (n\pi_1)^{-1} \Vert \mathbf{D}_{\mathcal{S}}^{\top}\bs{\varepsilon} \Vert^2
    = O_P\{\pi_1^{-1}(|\mathcal{S}_z|+|\mathcal{S}_{x,L}|+|\mathcal{S}_{x,N}|N_n)\},
\end{align*}
\begin{align*}
    \Vert \bs{\eta}_{*} \Vert^2
    \leq 2 \Vert \bs{\varepsilon}_{*} \Vert^2 + O_P(n|\mathcal{S}_{x,N}|N_n^{-2})
    = O_P\{\pi_1^{-1}(|\mathcal{S}_z|+|\mathcal{S}_{x,L}|+|\mathcal{S}_{x,N}|N_n)\}
    + O_P(nN_n^{-2}),
\end{align*}
\begin{align*}
    \Vert \hbtheta_{\mathcal{S}} - \bs{\theta}_{0, \mathcal{S}} \Vert^2
& \leq \frac{8\Vert \bs{\eta}_{*} \Vert^2}{n\pi_{1}} + \frac{4 \{\lambda_{n1}^2|\mathcal{S}_z| + \lambda_{n2}^2|\mathcal{S}_{x,L}|+ \lambda_{n3}^2|\mathcal{S}_{x,N}|\}}{n^2\pi_{1}^2} \\
 & = O_P\left\{ \frac{|\mathcal{S}_z|+|\mathcal{S}_{x,L}|+|\mathcal{S}_{x,N}|N_n}{n \pi_1^2}\right\} + O\left(\frac{|\mathcal{S}_{x,N}|}{\pi_1 N_n^2}\right) \\
 & \quad +  O_P\left\{\frac{\lambda_{n1}^2|\mathcal{S}_z| + \lambda_{n2}^2|\mathcal{S}_{x,L}|+ \lambda_{n3}^2|\mathcal{S}_{x,N}|}{n^2 \pi_1^2}\right\}.
\end{align*}
Therefore, the results follow by the facts that
\begin{align}
    \widehat{\bs{\alpha}}_{\mathcal{S}_z} - \bs{\alpha}_{0,\mathcal{S}_z} &= \left(\mathbf{I}_{|\mathcal{S}_z|}  ~~ \bs{0}_{|\mathcal{S}_z| \times |\mathcal{S}_{x,L}|}  ~~ \bs{0}_{|\mathcal{S}_z| \times (|\mathcal{S}_{x,N}|N_n)}\right)(\sigs{\hbtheta} - \bs{\theta}_{0,\mathcal{S}} ), \nonumber \\
    \widehat{\bs{\beta}}_{\mathcal{S}_{x, L}} - \bs{\beta}_{0,\mathcal{S}_{x, L}} &= \left(\bs{0}_{|\mathcal{S}_{x,L}| \times |\mathcal{S}_z|} ~~ \mathbf{I}_{|\mathcal{S}_{x,L}|}  ~~ \bs{0}_{|\mathcal{S}_{x,L}|\times (|\mathcal{S}_{x,N}|N_n)}\right)(\sigs{\hbtheta} - \bs{\theta}_{0,\mathcal{S}} ), \nonumber \\
    \widehat{\bs{\gamma}}_{\mathcal{S}_{x, N}} - \bs{\gamma}_{0,\mathcal{S}_{x, N}} &=\left(\bs{0}_{(|\mathcal{S}_{x,N}| N_n) \times |\mathcal{S}_z|} ~~ \bs{0}_{(|\mathcal{S}_{x,N}| N_n) \times |\mathcal{S}_{x,L}|} ~~ \mathbf{I}_{|\mathcal{S}_{x,N}| N_n}\right)(\sigs{\hbtheta} - \bs{\theta}_{0,\mathcal{S}}),
\label{EQN:diffs_betagamma}
\end{align}
and $\Vert \widehat{g}_{\ell} - g_{n\ell} \Vert_2^2 \asymp \Vert \hbgamma_{\ell} - \bs{\gamma}_{0\ell}  \Vert^2$,
where $\widehat{\bs{\beta}}_{\mathcal{S}_z} = (\widehat{\beta}_k, k \in \mathcal{S}_z)^{\top}$, $\bs{\beta}_{0,\mathcal{S}_z} = ({\beta}_{0k}, k \in \mathcal{S}_z)^{\top}$, $\widehat{\bs{\gamma}}_{\mathcal{S}_{x,N}} = (\widehat{\bs{\gamma}}_{\ell}, \ell \in \mathcal{S}_{x,N})^{\top}$ and $\bs{\gamma}_{0,\mathcal{S}_{x,N}} = (\bs{\gamma}_{\ell}, \ell \in \mathcal{S}_{x,N})^{\top}$.
\end{proof}

\subsection{Proof of Theorem \ref{THM:nonlinear-normality}}

In this section, the spline basis functions considered are of order $d$.
For any index set $\mathcal{A} \subseteq \{1, \ldots, p_1+p_2\}$, denote $\bs{\beta}_{\mathcal{A}} = (\beta_k, 1 \leq k \leq p_1, k \in \mathcal{A})^{\top}$, $\widehat{\bs{\beta}}_{\mathcal{A}} = (\widehat{\beta}_k, 1 \leq k \leq p_1, k \in \mathcal{A})^{\top}$, $\bs{\gamma}_{\mathcal{A}} = (\bs{\gamma}_{\ell}, 1 \leq \ell \leq p_2$, $\ell + p_1 \in \mathcal{A})^{\top}$ and $\widehat{\bs{\gamma}}_{\mathcal{A}} = (\widehat{\bs{\gamma}}_{\ell}, 1 \leq \ell \leq p_2$, $\ell + p_1 \in \mathcal{A})^{\top}$.
Next, denote $\mathbf{Z}_{\mathcal{A}} = (\mathbf{Z}_{i, \mathcal{A}} ^{\top}, i = 1, \ldots, n)^{\top}$, where $\mathbf{Z}_{i, \mathcal{A}} = (Z_{ik}, 1 \leq k \leq p_1, k \in \mathcal{A})^{\top}$, $\mathbf{X}_{i, \mathcal{A}} = (X_{i\ell}, 1 \leq \ell \leq p_2, \ell+p_1 \in \mathcal{A})^{\top}$. Similarly, denote $\mathbf{B}_{\mathcal{A}}^{(d)} = (\mathbf{B}_{i, \mathcal{A}}^{(d)\top}, i = 1, \ldots, n)^{\top}$, where $\mathbf{B}_{i, \mathcal{A}}^{(d)} = (B_{J,\ell}^{(d)}(X_{i\ell}), 1 \leq \ell \leq p_2$, $\ell + p_1 + p_2 \in \mathcal{A}, J = 1, \ldots, N_n)^{\top}$.
Define $\mathbf{T}_{\mathcal{S}} = (\mathbf{Z}_{\mathcal{S}_z}, \mathbf{X}_{\mathcal{S}_{x,PL}})$.
By an abuse of notation, let $\mathbf{D}_{\mathcal{S}} = \left(\mathbf{T}_{\mathcal{S}}, \mathbf{B}_{\mathcal{S}}^{(d)}\right)$,
and we define
\begin{equation}
    \mathbf{C}_{\mathcal{S}} = n^{-1}\mathbf{D}_{\mathcal{S}}^{\top} \mathbf{D}_{\mathcal{S}} =
    \begin{pmatrix}
       n^{-1} \mathbf{T}_{\mathcal{S}}^{\top}\mathbf{T}_{\mathcal{S}} & n^{-1}\mathbf{T}_{\mathcal{S}}^{\top}\mathbf{B}_{\mathcal{S}}^{(d)} \\
        n^{-1}\mathbf{B}_{\mathcal{S}}^{(d)\top}\mathbf{T}_{\mathcal{S}} & n^{-1}\mathbf{B}_{\mathcal{S}}^{(d)\top}\mathbf{B}_{\mathcal{S}}^{(d)}
    \end{pmatrix} =
    \begin{pmatrix}
        \mathbf{C}_{11} & \mathbf{C}_{12} \\
        \mathbf{C}_{21} & \mathbf{C}_{22}
    \end{pmatrix} ,
\label{EQN:cs}
\end{equation}
\begin{equation}
    \mathbf{U}_{\mathcal{S}} = \mathbf{C}_{\mathcal{S}} ^{-1} =
    \begin{pmatrix}
        \mathbf{U}_{11} & - \mathbf{U}_{11} \mathbf{C}_{12} \mathbf{C}_{22}^{-1} \\
        - \mathbf{U}_{22} \mathbf{C}_{21} \mathbf{C}_{11}^{-1} & \mathbf{U}_{22}
    \end{pmatrix} =
    \begin{pmatrix}
        \mathbf{U}_{11} & \mathbf{U}_{12} \\
        \mathbf{U}_{21} & \mathbf{U}_{22}
    \end{pmatrix},
\label{EQN:us}
\end{equation}
where $\mathbf{U}_{11}^{-1} = \mathbf{C}_{11} - \mathbf{C}_{12} \mathbf{C}_{22}^{-1} \mathbf{C}_{21} = n^{-1} \mathbf{T}_{\mathcal{S}}^{\top} \left(\mathbf{I}_n - \mathbf{P}_{\mathbf{B}_{\mathcal{S}}^{(d)}}\right) \mathbf{T}_{\mathcal{S}}$ and $\mathbf{U}_{22}^{-1} = \mathbf{C}_{22} - \mathbf{C}_{21} \mathbf{C}_{11}^{-1} \mathbf{C}_{12}$ $= n^{-1} \mathbf{B}_{\mathcal{S}}^{(d)\top} \left(\mathbf{I}_n - \mathbf{P}_{\mathbf{T}_{\mathcal{S}}}\right) \mathbf{B}_{\mathcal{S}}^{(d)}$,
with $\mathbf{P}_{\mathbf{B}_{\mathcal{S}}^{(d)}}$ and $\mathbf{P}_{\mathbf{T}_{\mathcal{S}}}$ being projection matrices for $\mathbf{B}_{\mathcal{S}}^{(d)}$ and $\mathbf{T}_{\mathcal{S}}$, respectively.

In the following, we give the proof of Theorem \ref{THM:nonlinear-normality}.

\begin{proof}[Proof of Theorem \ref{THM:nonlinear-normality}]
The structure of the proof is consisted of two parts: (i) we show the oracle efficiency of $\widehat{\phi}_{\ell}^{\mathrm{SBLL}}$; (ii) we show the uniform asymptotic normality for the ``oracle'' estimator $\widehat{\phi}_{\ell}^o$.

For part (i), note that for $\ell \in \mathcal{S}_{x,N}$
\[
\widehat{\phi}_{\ell}^{\mathrm{SBLL}}\left( x_{\ell}\right)-\widehat{\phi}_{\ell}^o\left( x_{\ell}\right) = \left(1, \, 0\right) (\mathbf{X}_{\ell}^{\ast\top}\mathbf{W}_{\ell}\mathbf{X}_{\ell}^{\ast})^{-1} \mathbf{X}_{\ell}^{\ast\top}\mathbf{W}_{\ell} (\widehat{\mathbf{Y}}_{\ell} - \mathbf{Y}_{\ell}),  ~~\mathrm{where}
\]
\begin{align*}
    \widehat{\mathbf{Y}}_{\ell} - \mathbf{Y}_{\ell}
    &= \mathbf{Z}_{\mathcal{S}_z} (\bs{\alpha}_{0, \mathcal{S}_z} - \widehat{\bs{\alpha}}_{\mathcal{S}_z}^{\ast})
    \!\!\!\!+ \mathbf{X}_{\mathcal{S}_{x,PL}} (\bs{\beta}_{0, \mathcal{S}_{x,PL}} - \widehat{\bs{\beta}}_{\mathcal{S}_{x,PL}}^{\ast})
+ \sum_{\ell^{\prime} \in \mathcal{S}_{x,N}\setminus \{\ell\} } \left\{\phi_{0\ell^{\prime}} (\mathbf{X}_{\ell^{\prime}}) - \widehat{\phi}_{\ell^{\prime}}^{\ast}(\mathbf{X}_{\ell^{\prime}}) \right\} \\
    &=\mathbf{Z}_{\mathcal{S}_z} (\bs{\alpha}_{0, \mathcal{S}_z} - \widehat{\bs{\alpha}}_{\mathcal{S}_z}^{\ast})
    + \mathbf{X}_{\mathcal{S}_{x,PL}} (\bs{\beta}_{0, \mathcal{S}_{x,PL}} - \widehat{\bs{\beta}}_{\mathcal{S}_{x,PL}}^{\ast}) \\
 & \quad \quad + \mathbf{B}_{\mathcal{S}_{x,N} \setminus \{\ell\}}^{(d)} (\bs{\gamma}_{0, \mathcal{S} _{x,N}\setminus \{\ell\}} - {\hbgamma}_{\mathcal{S}_{x,N} \setminus \{\ell\}}^{\ast})
    + \sum_{\ell^{\prime} \in \mathcal{S}_{x,N}\setminus \{\ell\} } \left\{\phi_{0\ell^{\prime}} (\mathbf{X}_{\ell^{\prime}}) - \phi_{nl^{\prime}} (\mathbf{X}_{\ell^{\prime}}) \right\} ,
\end{align*}
and
\begin{align*}
    \!\!\!\text{diag} (1, h_{\ell}^{-1}) \, &\mathbf{X}_{\ell}^{\ast\top}\mathbf{W}_{\ell}\mathbf{X}_{\ell}^{\ast} \, \text{diag} (1, h_{\ell}^{-1}) \\
    = &
        \begin{pmatrix}
            n^{-1}\sum_{i = 1}^n K_{h_{\ell}}(X_{i\ell} - x_{\ell}) & n^{-1} \sum_{i = 1}^n \left(\frac{X_{i\ell} - x_{\ell}}{h_{\ell}}\right)K_{h_{\ell}}(X_{i\ell} - x_{\ell}) \\
            n^{-1} \sum_{i = 1}^n \left(\frac{X_{i\ell} - x_{\ell}}{h_{\ell}}\right) K_{h_{\ell}}(X_{i\ell} - x_{\ell}) & n^{-1}\sum_{i = 1}^n \left(\frac{X_{i\ell} - x_{\ell}}{h_{\ell}}\right)^{2} K_{h_{\ell}}(X_{i\ell} - x_{\ell})
        \end{pmatrix} \\
    = &  f_{\ell}(x_{\ell})
        \begin{pmatrix}
            1 & 0 \\
            0 & \mu_2(K)
        \end{pmatrix}
     + u_P(1),
\end{align*}
with $u_P(\cdot) = o_P(\cdot)$ uniformly for all $x_{\ell} \in [a, b]$. So
\[
    \left(\mathbf{X}_{\ell}^{\ast\top}\mathbf{W}_{\ell}\mathbf{X}_{\ell}^{\ast}\right)^{-1} =  \text{diag} (1, h_{\ell}^{-1}) \, f_{\ell}^{-1}(x_{\ell}) \left\{
        \begin{pmatrix}
            1 & 0 \\
            0 & \mu_2(K)
        \end{pmatrix}
     + u_P(1) \right\} \, \text{diag} (1, h_{\ell}^{-1}),
\]
\[
    \text{diag} (1, h_{\ell}^{-1}) \, \mathbf{X}_{\ell}^{\ast\top}\mathbf{W}_{\ell} =  \frac{1}{n} \times
        \begin{pmatrix}
            K_{h_{\ell}} (X_{1\ell} - x_{\ell}) & ,\ldots, & K_{h_{\ell}} (X_{n\ell} - x_{\ell}) \\
            \left(\frac{X_{1\ell} - x_{\ell}}{h_{\ell}}\right) K_{h_{\ell}} (X_{1\ell} - x_{\ell}) & ,\ldots, & \left(\frac{X_{n\ell} - x_{\ell}}{h_{\ell}}\right) K_{h_{\ell}} (X_{n\ell} - x_{\ell})
        \end{pmatrix}.
\]
Thus,
\begin{align}
    \widehat{\phi}_{\ell}^{\mathrm{SBLL}}\left( x_{\ell}\right) & -\widehat{\phi}_{\ell}^o\left( x_{\ell}\right)  =   f_{\ell}^{-1} (x_{\ell}) \left[
    \frac{1}{n} \sum_{i = 1}^n K_{h_{\ell}}(X_{i\ell} - x_{\ell}) \mathbf{Z}_{i, \mathcal{S}_z}^{\top} (\bs{\alpha}_{0, \mathcal{S}_z} - \widehat{\bs{\alpha}}_{\mathcal{S}_z}^{\ast}) \right . \nonumber \\
 & + \frac{1}{n} \sum_{i = 1}^n K_{h_{\ell}}(X_{i\ell} - x_{\ell}) \mathbf{X}_{i, \mathcal{S}_{x,PL}}^{\top} (\bs{\beta}_{0, \mathcal{S}_{x,PL}} - \widehat{\bs{\beta}}_{\mathcal{S}_{x,PL}}^{\ast}) \nonumber \\
 & +  \frac{1}{n} \sum_{i = 1}^n K_{h_{\ell}}(X_{i\ell} - x_{\ell}) \mathbf{B}_{i, \mathcal{S}_{x,N} \setminus \{\ell\}}^{(d)\top} (\bs{\gamma}_{0, \mathcal{S}_{x,N} \setminus \{\ell\} } - {\hbgamma}_{\mathcal{S}_{x,N} \setminus \{\ell\} }^{\ast})  \nonumber  \\
 & \left . + \, \frac{1}{n} \sum_{i = 1}^n \sum_{\ell^{\prime} \in \mathcal{S}_{x,N} \setminus \{\ell\} } K_{h_{\ell}}(X_{i\ell} - x_{\ell}) \left\{\phi_{0\ell^{\prime}} ({X}_{i\ell^{\prime}}) - \phi_{nl^{\prime}} ({X}_{i\ell^{\prime}}) \right\} + u_P(1) \right].
\label{eqn:alpha_diff}
\end{align}

For the first and second summation terms in the right hand side of (\ref{eqn:alpha_diff}), by Theorem \ref{THM:normality}, we have $ n^{-1}\sum_{i = 1}^n K_{h_{\ell}}(X_{i\ell} - x_{\ell}) \mathbf{Z}_{i, \mathcal{S}_z}^{\top} (\bs{\alpha}_{0, \mathcal{S}_z} - \widehat{\bs{\alpha}}_{\mathcal{S}_z}^{\ast}) = u_P\left( n^{-1/2} \right) $ and $ n^{-1}\sum_{i = 1}^n K_{h_{\ell}}(X_{i\ell} - x_{\ell}) \mathbf{X}_{i, \mathcal{S}_{x,PL}}^{\top} (\bs{\beta}_{0, \mathcal{S}_{x,PL}} \!- \widehat{\bs{\beta}}_{\mathcal{S}_{x,PL}}^{\ast}) \!=\! u_P\left( n^{-1/2} \right)$;
and by Lemma \ref{LEM:spline-approx}, $n^{-1} \!\sum_{i = 1}^n \!\sum_{\ell^{\prime} \in \mathcal{S}_{x,N} \setminus \{\ell\}}K_{h_{\ell}}(X_{i\ell} - x_{\ell}) \left\{\phi_{0\ell^{\prime}} ({X}_{i\ell^{\prime}}) - \phi_{nl^{\prime}} ({X}_{i\ell^{\prime}}) \right\}  = u_P\left\{|\mathcal{S}_{x,N}|M_n^{-d}\right\}$.
As for the third terms, define $\zeta_{i\ell} = \phi_{0\ell}(X_{i\ell}) - \sum_{J=1}^{M_n} \gamma_{0J,l}^{\ast} B_{J,\ell}^{(d)}(X_{i\ell})$, $\zeta_{i} = \sum_{\ell \in \mathcal{S}_{x,N}}\zeta_{i\ell}$, and $\bs{\zeta} = (\zeta_1, \ldots, \zeta_n)^{\top}$, similar to the induction with (\ref{EQN:theta_dfc}), we have
\begin{equation}
    \hbtheta_{\mathcal{S}}^{\ast} - \bs{\theta}_{\mathcal{S}}^o
    = n^{-1} \mathbf{C}_{\mathcal{S}}^{-1} \{\mathbf{D}_{\mathcal{S}}^{\top} \left(\bs{\zeta} + \bs{\epsilon}\right)\},
\label{EQN:theta_dfc_star}
\end{equation}
then $\widehat{\bs{\gamma}}_{\mathcal{S}_{x,N}}^{\ast} - \bs{\gamma}_{0,\mathcal{S}_{x,N}}
    = \left(\bs{0}_{\{|\mathcal{S}_{x,N}|M_n\} \times (|\mathcal{S}_z|+|\mathcal{S}_{x,L}|)} ~ \mathbf{I}_{(|\mathcal{S}_{x,N}|M_n)}\right) \mathbf{C}_{\mathcal{S}}^{-1} n^{-1} \left\{\mathbf{D}_{\mathcal{S}}^{\top} \left(\bs{\zeta} +\bs{\epsilon}\right)\right \}$.
Define a diagonal matrix $\mathbf{I}_{\ell}^0 = \mathrm{diag}\{\bs{1}_{(l-1)M_n}, \bs{0}_{M_n}, \bs{1}_{(|\mathcal{S}_{x,N}| - l)M_n}\}$, $\ell \in \mathcal{S}_{x,N}$. Then
\begin{align*}
    \mathbf{B}_{i, \mathcal{S}_{x,N} \setminus \{\ell\}}^{(d)\top} (\bs{\gamma}_{0, \mathcal{S}_{x,N} \setminus \{\ell\} } - {\hbgamma}_{\mathcal{S}_{x,N}\setminus \{\ell\} }^{\ast} )
    = \, \mathbf{B}_{i, \mathcal{S}}^{(d)\top} \mathbf{I}_{\ell}^0
    \left(\widehat{\bs{\gamma}}_{\mathcal{S}_{x,N}}^{\ast} - \bs{\gamma}_{0,\mathcal{S}_{x,N}}\right).
\end{align*}
Next by Lemma \ref{LEM:spline-approx},  (\ref{EQN:cs}) and (\ref{EQN:us}), for any $\ell \in \mathcal{S}_{x,N}$, we have
\begin{align*}
   \frac{1}{n} \sum_{i = 1}^{n} & K_{h_{\ell}}(X_{i\ell} - x_{\ell}) \, \mathbf{B}_{i, \mathcal{S}}^{(d)\top} \mathbf{I}_{\ell}^0
    \left(\widehat{\bs{\gamma}}_{\mathcal{S}_{x,N}}^{\ast} - \bs{\gamma}_{0,\mathcal{S}_{x,N}}\right)\\
    & =  \frac{1}{n} \sum_{i = 1}^{n} K_{h_{\ell}}(X_{i\ell} - x_{\ell}) \mathbf{B}_{i, \mathcal{S}}^{(d)\top} \mathbf{I}_{\ell}^0 \mathbf{U}_{22} \left(- \mathbf{C}_{21} \mathbf{C}_{11}^{-1} ~~ \mathbf{I}_{(|\mathcal{S}_{x,N}|M_n)} \right)
    \frac{1}{n}\mathbf{B}_{\mathcal{S}}^{(d)\top} \left(\bs{\zeta}+\bs{\epsilon}\right) \\
    = & \frac{1}{n} \sum_{i = 1}^{n} K_{h_{\ell}}(X_{i\ell} - x_{\ell})  \mathbf{B}_{i, \mathcal{S}}^{(d)\top} \mathbf{I}_{\ell}^0 \mathbf{U}_{22} \frac{1}{n}\mathbf{B}_{\mathcal{S}}^{(d)\top} \left(\mathbf{I}_n - \mathbf{P}_{\mathbf{T}_{\mathcal{S}}}\right) \left(\bs{\zeta}+\bs{\epsilon}\right) .
\end{align*}
Following the same idea in the proof of Lemma \ref{LEM:eigen}, we have that
there exist constants $0<c_{U_{2}}<C_{U_{2}}<\infty $, such that with
probability approaching one, $c_{U_{2}}\mathbf{I}%
_{|\mathcal{S}_{x,N}|M_n}\leq \mathbf{U}_{22}\leq C_{U_{2}}\mathbf{I}_{|\mathcal{S}_{x,N}|M_n}$.
Similar to Lemma A.4 in \citet[][]{wang2007spline}, for any $\ell,\ell^{\prime} \in \mathcal{S}_{x,N}$ and $\ell\neq \ell^{\prime}$, we have
\begin{align*}
\sup_{x_{\ell} \in \chi_{h_{\ell}}}\max_{1\leq J\leq M_n}&
\left|\frac{1}{n}\!\sum_{i=1}^{n}\!\left[ K_{h_{\ell}}\left(
X_{i\ell}-x_{\ell}\right) B_{J\ell^{\prime}}^{(d)}\left( X_{i\ell^{\prime}}\right)
\!-\!E\{K_{h_{\ell}}\left(
X_{i\ell}-x_{\ell}\right) B_{J\ell^{\prime}}^{(d)}\left( X_{i\ell^{\prime}}\right)\} \right] \right| \!=\! O_P\!\left(\sqrt{\frac{\ln n}{nh_{\ell}}}\right),
\end{align*}
\begin{equation*}
\sup_{x_{\ell} \in \chi_{h_{\ell}}}\max_{1\leq J\leq M_n}\left|
n^{-1}\sum_{i=1}^{n}\mathrm{E}\{K_{h_{\ell}}\left(
X_{i\ell}-x_{\ell}\right) B_{J\ell}^{(d)}\left( X_{i\ell}\right)\} \right|
=O_{P}\left(M_n^{-1/2}\right).  \label{EQ:supxi}
\end{equation*}
We can show that
\begin{align*}
    \sup_{x_{\ell} \in \chi_{h_{\ell}}}\!\frac{1}{n} \!\sum_{i = 1}^{n} & K_{h_{\ell}}(X_{i\ell} - x_{\ell})   \mathbf{B}_{i, \mathcal{S}}^{(d)\top} \mathbf{I}_{\ell}^0 \mathbf{U}_{22} \frac{1}{n}\mathbf{B}_{\mathcal{S}}^{(d)\top} \left(\mathbf{I}_n - \mathbf{P}_{\mathbf{T}_{\mathcal{S}}}\right) \bs{\delta}
    \!=\!O_{P}\{M_n^{-d + 1} (\sqrt{\ln n / (nh_{\ell})}+M_n^{-1/2})\},
\end{align*}
and by Proposition 2 in \citet[][]{wang2009efficient},
\[
    \sup_{x_{\ell} \in \chi_{h_{\ell}}} \frac{1}{n} \sum_{i = 1}^{n} K_{h_{\ell}}(X_{i\ell} - x_{\ell})  \mathbf{B}_{i, \mathcal{S}}^{(d)\top} \mathbf{I}_{\ell}^0 \mathbf{U}_{22} \frac{1}{n}\mathbf{B}_{\mathcal{S}}^{(d)\top} \left(\mathbf{I}_n - \mathbf{P}_{\mathbf{T}_{\mathcal{S}}}\right) \bs{\epsilon} =O_{P}(\sqrt{\ln(n) / n}).
\]
Therefore,
\[
	\sup_{x_{\ell} \in \chi_{h_{\ell}}}|\widehat{\phi}_{\ell}^{\mathrm{SBLL}}\left( x_{\ell}\right)-\widehat{\phi}_{\ell}^o\left( x_{\ell}\right) |
	= O_{P}\left\{\sqrt{\frac{\ln n} {n}} + M_n^{-d+1} \left(\sqrt{\frac{\ln n} {nh_{\ell}}}+\sqrt{\frac{1}{M_n}}\right)\right\}.
\]

For part (ii), below we show that for any $t$ and $ \ell\in \mathcal{S}_{x,N}$,
\[
	\lim_{n \rightarrow \infty} \Pr\left\{\sqrt{\ln(h_{\ell}^{-2})} \left(\sup_{x_{\ell} \in \chi_{h_{\ell}}} \frac{\sqrt{nh_{\ell}}}{v_{\ell}(x_{\ell})} |\widehat{\phi}_{\ell}^{o}\left( x_{\ell}\right) -\phi_{0\ell}(x_{\ell})| - \tau_n\right) < t\right\}
	= e^{-2 e^{-t}},
\]
where 
$v_{\ell}^{2}(x_{\ell}) = \|K\|_2^2 f_{\ell}^{-1}(x_{\ell}) \sigma^{2}$,
$\tau_n = \sqrt{\ln(h_{\ell}^{-2})} + \ln\{\|K^{\prime}\|_2 / (2\pi \|K\|_2)\} / \sqrt{\ln(h_{\ell}^{-2})}$.

Define $M_{h}(x)=h_{\ell}^{-1/2}\int K\{(x^{\prime}-x)/h_{\ell}\}dW(x^{\prime})$, where $W(x)$ is a Wiener process defined on $(0,\infty)$. By the Lemma 1 in \cite{zheng2016statistical}, one has
\begin{equation}
\label{EQU:Wiener}
	\lim_{n \rightarrow \infty} \Pr\left[ \sqrt{\ln(h_{\ell}^{-2})}
		\bigg\{\sup_{x \in \chi_{h_{\ell}}}|M_{h_{\ell}}(x)|/\|K\|_{L_2}^2-\tau_n\bigg\} < t \right] = e^{-2e^{-t}}.
\end{equation}
Recall the definition of $\widehat{\phi}_{\ell}^{o}(x_{\ell})$ in (15), we have
\begin{align*}
	\widehat{\phi}_{\ell}^{o}(x_{\ell}) -\phi_{0\ell}(x_{\ell})
	=&\left(1 ~~ 0\right) \left(\mathbf{X}_{\ell}^{\ast\top}\mathbf{W}_{\ell}\mathbf{X}_{\ell}^{\ast}\right)^{-1} \mathbf{X}_{\ell}^{\ast\top}\mathbf{W}_{\ell} \mathbf{Y}_{\ell}  - \phi_{0\ell}(x_{\ell}) \\
	=&f_{\ell}^{-1} (x_{\ell})\frac{1}{n} \sum_{i = 1}^n K_{h_{\ell}}(X_{i\ell} - x_{\ell}) \varepsilon_i + O_P(h_{\ell}^2).
\end{align*}
According to the proof of Theorem 1 in \cite{zheng2016statistical}, we have
\begin{equation}
\label{SCB2}
	\sup_{x_{\ell} \in \chi_{h_{\ell}}}\left| \frac{\sqrt{nh_{\ell}}}{v_{\ell}(x_{\ell})} \{\widehat{\phi}_{\ell}^{o}( x_{\ell}) -\phi_{0\ell}(x_{\ell})\}-M_{h_{\ell}}(x_{\ell})/\|K\|_{L_2}^2\right|=o_P(\ln^{-1/2} n).
\end{equation}
Consequently, we have
\[
\sup_{x_{\ell} \in \chi_{h_{\ell}}}\sqrt{\ln(h_{\ell}^{-2})}\bigg| \frac{\sqrt{nh_{\ell}}}{v_{\ell}(x_{\ell})} \{\widehat{\phi}_{\ell}^{o}( x_{\ell}) -\phi_{0\ell}(x_{\ell})\}-M_{h_{\ell}}(x_{\ell})/\|K\|_{L_2}^2\bigg|=o_P(1),
\]
as $\sqrt{\ln(h_{\ell}^{-2})}/\sqrt{\ln(n)}=O(1)$.
The uniformly asymptotic normality of the ``oracle" estimator $\widehat{\phi}_{\ell}^{o}(x_{\ell})$ follows from (\ref{EQU:Wiener}) and Slusky's Theorem.

Hence, the result in (\ref{EQN:SBLL_Order}) is established. Consequently, the result in (\ref{EQN:sbll_normal}) follows from (\ref{EQN:alphao_normal}), and the result in (\ref{EQN:scb_form}) follows from \cite{claeskens2003bootstrap}.
\end{proof}

\subsection{Technical Lemmas}

The following lemmas are used in the proofs of Theorem \ref{THM:LASSO-1} and Theorem \ref{THM:selection}.

\begin{lemma}
\label{LEM:gLASSO_select} 
Suppose that Assumptions (A1)--(A4) hold. Recall the definition of $\mathcal{S}$ and $\widetilde{\mathcal{S}}$ in (\ref{DEF:SN}) and (\ref{DEF:S_tilde}),
with probability approaching one, $|\widetilde{\mathcal{S}}| \leq M_1|{\mathcal{S}}| = M_1(|\mathcal{S}_z| + |\mathcal{S}_{x,L}| + |\mathcal{S}_{x,N}|)$ for a finite constant $M_1 > 1$.
\end{lemma}

\begin{proof} [Proof]
The basic idea of the proof is similar to the proofs of Theorem 1 of \cite{zhang2008sparsity} and Theorem 2.1 of \cite{wei2010consistent}. The main differences are the error term shown in (\ref{EQN:eta_form}) and we have a more complex data structure. By Lemma \ref{LEM:spline-approx}, for some constant $C_{4} > 0$, we have $\Vert \bs{\delta} \Vert_2 \leq C_{4} \sqrt{n |\mathcal{S}_{x,N}| N_n^{-2}} = C_{4} |\mathcal{S}_{x,N}|^{1/2} n^{1 / 2} N_n^{-1}$.
For any positive integers $s_1$, $s_2$ and $s_3$, pick some index sets $\mathcal{A}_1 \subseteq \{1, \ldots, p_{1}\}$, $\mathcal{A}_2 \subseteq \{1, \ldots, p_{2}\}$ and $\mathcal{A}_3 \subseteq \{1, \ldots, p_{2}\}$ such that the cardinalities of $\mathcal{A}_1$, $\mathcal{A}_2$ and $\mathcal{A}_3$ are $|\mathcal{A}_1| = s_1$, $|\mathcal{A}_2| = s_2$ and $|\mathcal{A}_3| = s_3$, respectively. Denote $\mathcal{A} = \mathcal{A}_1 \cup \mathcal{A}_2 \cup \mathcal{A}_3$. Define an $(s_1 + s_2 + s_3 N_n) \times 1$ vector $\mathbf{S}_\mathcal{A} = \left(\widetilde{\lambda}_{n1}\mathbf{u}_{s_1}^{\top}, \widetilde{\lambda}_{n2}\mathbf{u}_{s_2}^{\top}, \widetilde{\lambda}_{n3}\sqrt{N_n}\mathbf{U}_{1}^{\top}, \ldots, \widetilde{\lambda}_{n3}\sqrt{N_n}\mathbf{U}_{s_3}^{\top}\right)^{\top}$, where $\mathbf{u}_{s_1} \in \{\pm 1\}^{s_1}$, $\mathbf{u}_{s_2} \in \{\pm 1\}^{s_2}$ and $\mathbf{U}_j$ is in a unit ball with dimension $N_n$, that is, $\mathbf{U}_j \in \mathbb{R}^N$ and $\Vert \mathbf{U}_j \Vert_2 = 1$, $j = 1, \ldots, s_3$. Let $\mathbf{P}_\mathcal{A} = \mathbf{D}_\mathcal{A} (\mathbf{D}_\mathcal{A}^{\top}\mathbf{D}_\mathcal{A})^{-1}\mathbf{D}_\mathcal{A}^{\top}$ be the projection matrix of $\mathbf{D}_\mathcal{A}$. Define
\[
    \chi_{s_1, s_2, s_3} = \max_{\substack{|\mathcal{A}_1| = s_1, |\mathcal{A}_2| = s_2, \\ |\mathcal{A}_3| = s_3 \\ \mathcal{A} = \mathcal{A}_1 \cup \mathcal{A}_2 \cup \mathcal{A}_3}}
    \max\limits_{\substack{\mathbf{u}_{s_1} \in \{\pm 1\}^{s_1} \\ \mathbf{u}_{s_2} \in \{\pm 1\}^{s_2} \\ \Vert \mathbf{U}_j \Vert_2 = 1, \, 1 \leq j \leq s_3}}
    \frac{|\bs{\eta}^{\top} \{\mathbf{D}_\mathcal{A}\left(\mathbf{D}_\mathcal{A}^{\top}\mathbf{D}_\mathcal{A}\right)^{-1}\mathbf{S}_\mathcal{A} - (\mathbf{I} - \mathbf{P}_\mathcal{A})\mathbf{D}\bs{\theta}_0\}|}
    {\Vert \mathbf{D}_\mathcal{A}\left(\mathbf{D}_\mathcal{A}^{\top}\mathbf{D}_\mathcal{A}\right)^{-1}\mathbf{S}_\mathcal{A} - (\mathbf{I} - \mathbf{P}_\mathcal{A})\mathbf{D}\bs{\theta}_0 \Vert_2},
\]
\begin{eqnarray*}
    \Omega_{|\mathcal{S}_z|, |\mathcal{S}_{x,L}|, |\mathcal{S}_{x,N}|} = \left\{(\mathbf{D}, \bs{\eta}): {\chi_{s_1, s_2, s_3} \leq \sigma C_{2} \sqrt{s_1\ln(p_1) \vee s_2\ln(p_2) \vee s_3 N_n \ln(p_3N_n)}}, \right. \\
        \left. \forall \, s_1 \geq |\mathcal{S}_z|, \, s_2 \geq |\mathcal{S}_{x,L}|, \, s_3 \geq |\mathcal{S}_{x,N}| \right\},
\end{eqnarray*}
where $C_{2} > 0$ is some sufficiently large constant.
As shown in the proof of Theorem 1 of \citet[][]{zhang2008sparsity} and Theorem 2.1 of \citet[][]{wei2010consistent}, there exists a constant $M_1 > 1$, such that if $\left(\mathbf{D}, \bs{\eta}\right) \in \Omega_{|\mathcal{S}_z|, |\mathcal{S}_{x,L}|, |\mathcal{S}_{x,N}|}$, then $|\widetilde{\mathcal{S}}| \leq M_1|{\mathcal{S}}| = M_1(|\mathcal{S}_z| + |\mathcal{S}_{x,L}| + |\mathcal{S}_{x,N}|)$.

So it suffices to show that $\left(\mathbf{D}, \bs{\eta} \right) \in \Omega_{|\mathcal{S}_z|, |\mathcal{S}_{x,L}|, |\mathcal{S}_{x,N}|}$. Denote $\mathbf{V}_\mathcal{A} = \mathbf{D}_\mathcal{A}\left(\mathbf{D}_\mathcal{A}^{\top}\mathbf{D}_\mathcal{A}\right)^{-1}\mathbf{S}_\mathcal{A} - (\mathbf{I} - \mathbf{P}_\mathcal{A})\mathbf{D}\bs{\theta}_0$, then by the triangle and Cauchy-Schwarz inequalities,
  \[
    \frac{|\bs{\eta}^{\top} \, \mathbf{V}_\mathcal{A}|}{\Vert \mathbf{V}_\mathcal{A} \Vert_2} = \frac{|\bs{\varepsilon}^{\top} \, \mathbf{V}_\mathcal{A} + \bs{\delta}^{\top} \, \mathbf{V}_\mathcal{A}|}{\Vert \mathbf{V}_\mathcal{A} \Vert_2} \leq \frac{|\bs{\varepsilon}^{\top} \, \mathbf{V}_\mathcal{A}|}{\Vert \mathbf{V}_\mathcal{A}\Vert_2} + \Vert \bs{\delta} \Vert_2 .
  \]
  For the ${|\bs{\varepsilon}^{\top} \, \mathbf{V}_\mathcal{A}|} / {\Vert \mathbf{V}_\mathcal{A} \Vert_2}$ part, define
    \[
        \chi_{s_1, s_2, s_3} = \max_{\substack{|\mathcal{A}_1| = s_1, |\mathcal{A}_2| = s_2, \\ |\mathcal{A}_3| = s_3 \\ \mathcal{A} = \mathcal{A}_1 \cup \mathcal{A}_2 \cup \mathcal{A}_3}}
    \max\limits_{\substack{\mathbf{u}_{s_1} \in \{\pm 1\}^{s_1} \\ \mathbf{u}_{s_2} \in \{\pm 1\}^{s_2} \\ \Vert \mathbf{U}_j \Vert_2 = 1, \, 1 \leq j \leq s_3}}
        \frac{|\bs{\varepsilon}^{\top} \mathbf{V}_\mathcal{A}|}{\Vert \mathbf{V}_\mathcal{A} \Vert_2},
    \]
\begin{eqnarray*}
        \Omega_{|\mathcal{S}_z|, |\mathcal{S}_{x,L}|, |\mathcal{S}_{x,N}|}^{\ast} = \left\{(\mathbf{D}, \bs{\varepsilon}): \chi_{s_1, s_2, s_3}^{\ast} \leq \sigma C_{3} \sqrt{s_1\ln(p_1) \vee s_2\ln(p_2) \vee s_3 N_n \ln(p_3N_n)}, \right. \\
        \left. \forall \, s_1 \geq |\mathcal{S}_z|, \, s_2 \geq |\mathcal{S}_{x,L}|, \, s_3 \geq |\mathcal{S}_{x,N}| \right\},
\end{eqnarray*}
where $C_3 > 0$ is some sufficiently large constant. As shown in the proof of Theorem 1 of \citet[][]{zhang2008sparsity} and Theorem 2.1 of \citet[][]{wei2010consistent},
 $\Pr(\Omega_{|\mathcal{S}_z|, |\mathcal{S}_{x,L}|, |\mathcal{S}_{x,N}|}^{\ast}) \rightarrow 1$. And for $\Vert \bs{\delta} \Vert_2 $ part, for $n$ sufficiently large and $N_n \asymp n^{1/3}$,
  \[
	\Vert \bs{\delta} \Vert_2 \leq C_{4}|\mathcal{S}_{x,N}|^{1/2} n^{1 / 2} N_n^{-1}
	\leq \sigma C_{5} \sqrt{s_1\ln(p_1) \vee s_2\ln(p_2) \vee s_3 N_n \ln(p_2N_n)}.
  \]
  It follows that $\Pr(\Omega_{|\mathcal{S}_z|, |\mathcal{S}_{x,L}|, |\mathcal{S}_{x,N}|}) \rightarrow 1$. This completes the proof.
\end{proof}

For any random variable $X$, denote $\Vert X \Vert_p = (\mathrm{E}|X|^p)^{1 / p}$ as the $L_p$ norm for random variable $X$; and denote $\Vert X \Vert_{\varphi} = \inf \left\{C > 0: \mathrm{E}\{\varphi\left({|X|}/{C}\right)\} \leq 1 \right\}$ as the \textit{Orlicz} norm for random variable $X$, where $\varphi$ is required as a non-decreasing, convex function with $\varphi(0) = 0$.

\begin{lemma}
\label{LEM:2}
Suppose that Assumptions (A2) and (A4) hold. Let
\begin{align*}
	T_{1k} & = n^{-1/2} \sum_{i = 1}^{n} Z_{ik} \varepsilon_i , ~1 \leq k \leq p_1, \quad
	T_{2\ell} = n^{-1/2} \sum_{i = 1}^{n} X_{i\ell} \varepsilon_i , ~1 \leq \ell \leq p_2, \\
	T_{3J\ell} & = n^{-1/2} \sum_{i = 1}^{n} B_{J,\ell}^{(d)}(X_{i\ell}) \varepsilon_i,  ~1 \leq \ell \leq p_2, ~ 1 \leq J \leq N_n,
\end{align*}
and $T_1 = \max\limits_{1 \leq k \leq p_1} |T_{1k}|$, $T_2 = \max\limits_{1 \leq \ell \leq p_2} |T_{2\ell}|$ and $T_3 = \max\limits_{1 \leq \ell \leq p_2,  1 \leq J \leq N_n} |T_{3J\ell}|$. Then we have
\begin{eqnarray*}
	&\mathrm{E}(T_1) \leq {C_1\sqrt{\ln(p_1)}}, \quad \mathrm{E}(T_2) \leq {C_2\sqrt{\ln(p_2)}}, \\
	&\mathrm{E}(T_3) \leq {C_3 n^{-1/2} \sqrt{\ln(p_2 N_n)} \left(\sqrt{2 C_4\,n N_n \ln(2p_2N_n)} + C_5N_n^{1/2}\ln(2p_2N_n) + n\right)^{1/2} },
\end{eqnarray*}
where $C_1$, $C_2$, $C_3$, $C_4$ and $C_5$ are positive constants.

In particular, when $ {N_n \ln(p_{2}N_n) / n \rightarrow 0}$, we have
\[
	\mathrm{E}(T_1) = O\{\sqrt{\ln(p_1)}\}, \quad
	\mathrm{E}(T_2) = O\{\sqrt{\ln(p_2)}\}, \quad
	\mathrm{E}(T_3) = O\{\sqrt{\ln(p_2N_n)}\}.
\]
\end{lemma}

\begin{proof}
Denote $s_{1nk}^2 = \sum_{i = 1}^n Z_{ik}^2$, $1 \leq k \leq p_1$, $s_{2nl}^2 = \sum_{i = 1}^n X_{i\ell}^2$, $1 \leq \ell \leq p_2$, $s_{3nJ\ell}^2 = \sum_{i = 1}^n \{B_{J,\ell}^{(d)}(X_{i\ell})\}^2$, $1 \leq \ell \leq p_2$, $1 \leq J \leq N_n$.
Next let $s_{1n}^2 = \max_{1 \leq k \leq p_1} s_{1nk}^2$, $s_{2n}^2 = \max_{1 \leq \ell \leq p_2} s_{2nl}^2$ and $s_{3n}^2 = \max_{1 \leq \ell \leq p_2, 1 \leq J \leq N_n} s_{3nJ\ell}^2$.
By Assumption (A2), conditional on $\mathbb{Z}=\{Z_{ik}, \, 1 \leq i \leq n, \, 1 \leq k \leq p_1\}$, $\sqrt{n} \, T_{1k}$ is $b{\left(\sum_{i = 1}^n Z_{ik}^2\right)}^{1/2}$--subgaussian; and conditional on $\mathbb{X}=\left\{X_{i\ell}, \, 1 \leq i \leq n, \, 1 \leq \ell \leq p_2\right\}$, $\sqrt{n} \, T_{2\ell}$ is $b{\left(\sum_{i = 1}^n X_{i\ell}^2\right)}^{1/2}$-subgaussian, and $\sqrt{n} \, T_{3J\ell}$ is $b\left[\sum_{i = 1}^n \{B_{J,\ell}^{(d)} (X_{i\ell})\}^2\right]^{1 / 2}$-subgaussian.

    Define $\varphi_p(x) = \exp({x^p}) - 1, \, p \geq 1$. Then $\varphi_p^{-1}(m) = {\left\{\ln(1 + m)\right\}}^{1/p}$. By Assumption (A2) and the maximal inequality for sub-Gaussian random variables (as stated in Lemmas 2.2.1 and 2.2.2 of \citet[][]{van1996weak}),
\begin{align*}
         \mathrm{E}\left(T_{1} | \mathbb{Z}\right) & \!=\!
        \mathrm{E}\left(\max\limits_{1 \leq k \leq p_1} \big|T_{1k}| | \mathbb{Z}\right)
         \!=\!  \left\Vert \max\limits_{1 \leq k \leq p_1} |T_{1k}| \big| \mathbb{Z} \right\Vert_1
        \!\!\!\leq\! \left\Vert \max\limits_{1 \leq k \leq p_1} |T_{1k}| \big| \mathbb{Z} \right\Vert_{\varphi_1} \!\!\!\!\!\!\leq\! \sqrt{\ln (2)} \left\Vert \max\limits_{1 \leq k \leq p_1} |T_{1k}| \big| \mathbb{Z} \right\Vert_{\varphi_2} \\
 & \leq K_1 \sqrt{\ln 2} \sqrt{\ln(1 + p_1)} \, {n}^{-1/2} \max\limits_{1 \leq k \leq p_1} \left\Vert \sqrt{n} \, T_{1k} \big| \left\{Z_{ik}, 1 \leq i \leq n, \, 1 \leq k \leq p_1 \right\} \right\Vert_{\varphi_2} \\
 & \leq K_1 \sqrt{\ln 2} \sqrt{\ln(1 + p_1)} \, {n}^{-1/2} \max\limits_{1 \leq k \leq p_1} \left(6b^2 \sum_{i = 1}^n Z_{ik}^2\right)^{1/2} \\
        &\leq {C_{5} {n}^{-1/2}s_{1n} \sqrt{\ln(p_1)}}.
\end{align*}
Next,
\begin{align*}
         \mathrm{E}\left(T_{2} | \mathbb{X}\right) & =
        \mathrm{E}\left(\max\limits_{1 \leq \ell \leq p_2} \big|T_{2\ell}| | \mathbb{X}\right)
         =  \left\Vert \max\limits_{1 \leq \ell \leq p_2} |T_{2\ell}| \big| \mathbb{X} \right\Vert_1
        \leq \sqrt{\ln (2)} \left\Vert \max\limits_{1 \leq \ell \leq p_2} |T_{2\ell}| \big| \mathbb{X} \right\Vert_{\varphi_2} \\
 & \leq K_2 \sqrt{\ln 2} \sqrt{\ln(1 + p_2)} \, {n}^{-1/2} \max\limits_{1 \leq \ell \leq p_2} \left\Vert \sqrt{n} \, T_{2\ell} \big| \left\{X_{i\ell}, 1 \leq i \leq n, \, 1 \leq \ell \leq p_2 \right\} \right\Vert_{\varphi_2} \\
 & \leq K_2 \sqrt{\ln 2} \sqrt{\ln(1 + p_2)} \, {n}^{-1/2} \max\limits_{1 \leq \ell \leq p_2} \left(6b^2 \sum_{i = 1}^n X_{i\ell}^2\right)^{1/2} \\
        &\leq {C_{5} {n}^{-1/2}s_{2n} \sqrt{\ln(p_2)}}.
\end{align*}
\begin{align*}
        \mathrm{E}\left(T_{3} | \mathbb{X}\right) & =  \mathrm{E}\left(\max\limits_{1 \leq \ell \leq p_2, \, 1 \leq J \leq N_n} |T_{3J\ell}| \big| \mathbb{X}\right) = \left\Vert \max\limits_{1 \leq \ell \leq p_2, \, 1 \leq J \leq N_n} |T_{3J\ell}| \big| \mathbb{X} \right\Vert_1\\
        &\leq  \sqrt{\ln (2)} \Vert \max\limits_{1 \leq \ell \leq p_2, \, 1 \leq J \leq N_n} |T_{3J\ell}| \big| \mathbb{X} \Vert_{\varphi_2} \\
 & \leq  K_3 \sqrt{\ln 2} \sqrt{\ln(1 + p_2 N_n)} \, {n}^{-1/2} \max\limits_{1 \leq \ell \leq p_2, \, 1 \leq J \leq N_n} \left[6b^2 \sum_{i = 1}^n \{B_{J,\ell}^{(d)}(X_{i\ell})\}^2 \right]^{1/2} \\
        &= {C_{21}{n}^{-1/2} s_{3n} \sqrt{\ln(p_2 N_n)}}.
\end{align*}
Thus,
\begin{align*}
	\mathrm{E}\left(T_{1} \right) &\leq {C_{11} {n}^{-1/2} \sqrt{\ln(p_1)}  \mathrm{E} (s_{1n})},  ~
	\mathrm{E}\left(T_{2} \right)  \leq {C_{21} {n}^{-1/2} \sqrt{\ln(p_2)}  \mathrm{E} (s_{2n})},  \\
	\mathrm{E}\left(T_{3} \right)  &\leq {C_{31} {n}^{-1/2} \sqrt{\ln(p_2 N_n)} \, \mathrm{E}(s_{3n})},
\end{align*}
    where $K_{1}$, $K_{2}$, $K_3$, $C_{11}$, $C_{21}$ and $C_{31}$ are positive constants. By Assumption (A4), we have $\mathrm{E} (Z_{ik})^2 \leq  C_{13}^2 $ and $\mathrm{E} (X_{i\ell})^2 \leq  C_{23}^2 $. The properties of normalized B-splines imply that, for every $l, J$, there exist positive constants $C_{13}$, and $C_{4}$, such that $|B_{J,\ell}^{(d)}(X_{i\ell})| \leq C_4 N_n^{1 / 2}$ and $\mathrm{E} \left(B_{J,\ell}^{(d)}(X_{i\ell})\right)^2 = 1$. Therefore,
$\mathrm{E}(s_{1n}^2) = \max\limits_{1 \leq k \leq p_1} \mathrm{E}(s_{1nk}^2 )= \max\limits_{1 \leq k \leq p_1}  \sum\limits_{i = 1}^n   \mathrm{E} (Z_{ik}^2) \leq n C_{13}^2$, $\mathrm{E}(s_{2n}^2) = \max\limits_{1 \leq \ell \leq p_2} \mathrm{E}(s_{2nl}^2 )= \max\limits_{1 \leq \ell \leq p_2}  \sum\limits_{i = 1}^n   \mathrm{E} (X_{i\ell}^2)\leq n C_{23}^2$, and $\sum\limits_{i = 1}^n \mathrm{E} \left[\{B_{J,\ell}^{(d)}(X_{i\ell})\}^2 - \mathrm{E}^2 \{B_{J,\ell}^{(d)}(X_{i\ell})\}\right]^2$ $\leq n N_n C_3$.
Thus, by Lemma A.1 of \citet[][]{van2008high}, we have
\[
\mathrm{E} \Bigg[\max_{\substack{1 \leq \ell \leq p_2, \\ 1 \leq J \leq N_n}} \Bigg| \sum\limits_{i = 1}^n \Bigg\{B_{J,\ell}^{(d)}(X_{i\ell})\Bigg\} ^2 - \mathrm{E}^2 \left\{B_{J,\ell}^{(d)}(X_{i\ell})\right\} \Bigg| \Bigg] \leq \sqrt{2C_4 n N_n\ln(2p_2N_n)} + C_5 N_n^{1 / 2} \ln(2p_2N_n).
\]
Therefore, by triangle inequality, $\mathrm{E}(s_{3n}^2) \leq  \sqrt{2C_4\,n N_n \ln(2p_2N_n)} + C_5 N_n^{1 / 2} \ln(2p_2N_n) + n$.
Thus, $\mathrm{E}(s_{1n})\leq (\mathrm{E}s_{1n}^2)^{1/2} \leq {\left(C_{13}^2n\right)^{1/2}}$, $\mathrm{E}(s_{2n})\leq (\mathrm{E}s_{2n}^2)^{1/2} \leq {\left(C_{23}^2n\right)^{1/2}}$, and
    \[
        \mathrm{E}(s_{3n}) \leq  (\mathrm{E}s_{3n}^2)^{1/2} \leq {\left\{\sqrt{2C_4 n N_n\ln(2p_2N_n)} + C_5 N_n^{1 / 2} \ln(2p_2N_n) + n\right\}^{1/2}}.
    \]
    The lemma follows.
\end{proof}

\begin{lemma}
\label{LEM:4} 
For
\begin{align*}
	\bs{v}_1 &= \left\{\left(\frac{\omega_k^{\alpha} \overline{\theta}_{0,k}}{|\overline{\theta}_{0,k}|}, k\in\mathcal{S}_z\right)^{\top}, \bs{0}_{|\mathcal{S}_{x,L}|}^{\top} , \bs{0}_{|\mathcal{S}_{x,N}|N_n}^{\top} \right\}^{\top},\\
	\bs{v}_2 &= \left\{\bs{0}_{|\mathcal{S}_z|}^{\top}, \left(\frac{\omega_{\ell}^{\beta} \overline{\theta}_{0,|\mathcal{S}_z|+\ell}}{|\overline{\theta}_{0,|\mathcal{S}_z|+\ell}|}, \ell\in\mathcal{S}_{x,L}\right)^{\top}, \bs{0}_{|\mathcal{S}_{x,N}|N_n}^{\top} \right\}^{\top}, \\
	\bs{v}_3 &= \left\{\bs{0}_{|\mathcal{S}_z|}^{\top}, \bs{0}_{|\mathcal{S}_{x,L}|}^{\top}, \left(\frac{\omega_{\ell}^{\gamma}\overline{\bs{\theta}}_{0,|\mathcal{S}_z|+|\mathcal{S}_{x,L}|+\ell}^{\top}}{\Vert \overline{\bs{\theta}}_{0,|\mathcal{S}_z|+|\mathcal{S}_{x,L}|+\ell}\Vert_2}, \ell\in\mathcal{S}_{x,N}\right)^{\top} \right\}^{\top},
\end{align*}
under Assumption (A5),
\begin{align}
	\Vert \bs{v}_1 \Vert_2^2 &= O_P \left( h_{n1}^2 \right) = O_P \left(b_{n1}^{-4}c_{b1}^{-2}r_{n1}^{-2} + |\mathcal{S}_z|b_{n1}^{-2}\right), \label{EQN:v1_order} \\
	\Vert \bs{v}_2 \Vert_2^2 &= O_P \left( h_{n2}^2 \right) = O_P \left\{b_{n2}^{-4}c_{b2}^{-2}r_{n2}^{-2} + |\mathcal{S}_{x,L}|b_{n2}^{-2}\right\}, \label{EQN:v2_order} \\
	\Vert \bs{v}_3 \Vert_2^2 &= O_P \left( h_{n3}^2 \right) = O_P \left\{b_{n3}^{-4}c_{b3}^{-2}r_{n3}^{-2} + |\mathcal{S}_{x,N}|b_{n3}^{-2}\right\}. \label{EQN:v3_order}
\end{align}
\end{lemma}

\begin{proof}
Write
\begin{align*}
	\Vert \bs{v}_1 \Vert _2^2 &= \sum\limits_{k \in \mathcal{S}_z} \left(\omega_k^{\alpha}\right)^2
	= \sum\limits_{k \in \mathcal{S}_z} \, |\widetilde{\alpha}_k |^{-2}
	= \sum\limits_{k \in \mathcal{S}_z} \, \frac{\alpha_{0k}^2 - \widetilde{\alpha}_k^2}{\alpha_{0k}^2  \widetilde{\alpha}_k^2}
	+ \sum\limits_{k \in \mathcal{S}_z} \, |\alpha_{0k}|^{-2} , \\
	\Vert \bs{v}_2 \Vert _2^2 &= \sum\limits_{\ell \in \mathcal{S}_{x,L}} \left(\omega_{\ell}^{\beta}\right)^2
	= \sum\limits_{\ell \in \mathcal{S}_{x,L}} \, |\widetilde{\beta}_{\ell} |^{-2}
	= \sum\limits_{\ell \in \mathcal{S}_{x,L}} \, \frac{\beta_{0\ell}^2 - \widetilde{\beta}_{\ell}^2}{\beta_{0\ell}^2  \widetilde{\beta}_{\ell}^2}
	+ \sum\limits_{\ell \in \mathcal{S}_{x,L}} \, |\beta_{0\ell}|^{-2} , \\
	\Vert \bs{v}_3 \Vert _2^2 &= \sum\limits_{\ell \in \mathcal{S}_{x,N}} \left(\omega_{\ell}^{\gamma}\right)^2
	= \sum\limits_{\ell \in \mathcal{S}_{x,N}} \Vert \tbgamma_{\ell} \Vert_2^{-2}
	= \sum\limits_{\ell \in \mathcal{S}_{x,N}} \, \frac{\Vert\bs{\gamma}_{0\ell}\Vert_2^2 - \Vert\tbgamma_{\ell}\Vert^2}{\Vert\bs{\gamma}_{0\ell}\Vert_2^2  \Vert\tbgamma_{\ell}\Vert_2^2}
	+ \sum\limits_{\ell \in \mathcal{S}_{x,N}} \Vert \bs{\gamma}_{0\ell} \Vert_2^{-2} .
\end{align*}
Under (A5), there exist positive constants $M_1$, $M_2$ and $M_3$, such that
\begin{align*}
	\sum_{k \in \mathcal{S}_z} \frac{\left| \alpha_{0k}^2 - \widetilde{\alpha}_k^2 \right|}{\alpha_{0k}^2 \widetilde{\alpha}^2}
	&\leq M_1c_{b1}^{-2}b_{n1}^{-4} \|\widetilde{\bs{\alpha}} - \bs{\alpha}_0 \|^2
	= O_P \left( b_{n1}^{-4}c_{b1}^{-2} r_{n1}^{-2}\right), \\
	\sum_{\ell \in \mathcal{S}_{x,L}} \frac{\left| \beta_{0\ell}^2 - \widetilde{\beta}_{\ell}^2 \right|}{\beta_{0\ell}^2 \widetilde{\beta}^2}
	&\leq M_2c_{b2}^{-2}b_{n2}^{-4} \|\tbbeta - \bs{\beta}_0 \|^2
	= O_P \left( b_{n2}^{-4}c_{b2}^{-2} r_{n2}^{-2}\right), \\
	\sum_{\ell \in \mathcal{S}_{x,N}} \, \frac{\left|\|\bs{\gamma}_{0\ell}\|^2 - \Vert\tbgamma_{\ell}\Vert^2 \right|}{\Vert\bs{\gamma}_{0\ell}\Vert^2 \Vert\tbgamma\Vert^2}
	&\leq M_3 c_{b3}^{-2}b_{n3}^{-4} \Vert \tbgamma - \bs{\gamma}_0 \Vert^2
	= O_P \left( b_{n3}^{-4}c_{b3}^{-2} r_{n3}^{-2}\right),
\end{align*}
and the results follow from that $\sum_{k \in \mathcal{S}_z} |\alpha_k |^{-2} \leq |\mathcal{S}_z| b_{n1}^{-2}$, $\sum_{\ell \in \mathcal{S}_{x,L}} |\beta_{\ell} |^{-2} \leq |\mathcal{S}_{x,L}| b_{n2}^{-2}$ and $\sum_{\ell \in \mathcal{S}_{x,N}}\Vert \bs{\gamma}_{\ell} \Vert^{-2} \leq |\mathcal{S}_{x,N}|b_{n3}^{-2}$.
\end{proof}

{\baselineskip=12pt
\bibliographystyle{asa}
\bibliography{references_abbr}
}

\end{document}